\documentclass[journal,onecolumn]{IEEEtran}
\usepackage{amsfonts}
\usepackage{mathrsfs}
\usepackage{amsmath}
\usepackage{amssymb}
\usepackage{graphicx}
\usepackage{epic}
\usepackage{makecell}
\usepackage{epstopdf}
\usepackage{color}
\usepackage{cite}
\usepackage{bbm}
\usepackage{multirow}
\usepackage{makecell}
\usepackage{mathtools}
\usepackage{enumerate}
\usepackage{hyperref}
\usepackage{lineno}
\usepackage{url}
\usepackage{float}
\usepackage{color}
\usepackage{bm}
\usepackage{comment}
\usepackage{booktabs}
\usepackage{tikz}
\usepackage{indentfirst}
\usepackage{exscale} 
\usepackage{relsize}
\usepackage{pifont}
\usepackage{enumerate}
\newtheorem{claim}{Claim}[section]
\newtheorem{theorem}{Theorem}[section]
\newtheorem{corollary}[theorem]{Corollary}
\newtheorem{definition}[theorem]{Definition}

\newtheorem{lemma}[theorem]{Lemma}

\newtheorem{remark}[theorem]{Remark}
\newtheorem{observation}[theorem]{Observation}

\newcommand{\C}{\mathcal {C}}

\begin{document}
	
\title{Correcting One Deletion and One Substitution with a Constant Number of Reads}
\author{Yuling~Li, Yubo~Sun, Gennian~Ge%
\thanks{This research was supported by the National Key Research and Development Program of China under Grant 2025YFC3409900, the National Natural Science Foundation of China under Grant 12231014, and Beijing Scholars Program.}
\thanks{Y. Li ({\tt 2240501022@cnu.edu.cn}), Y. Sun ({\tt 2200502135@cnu.edu.cn}), and G. Ge ({\tt gnge@zju.edu.cn}) are with the School of Mathematical Sciences, Capital Normal University, Beijing 100048, China.}
}

\maketitle

\begin{abstract}
In this paper, we investigate the problem of designing \((n, N; \mathcal{B})\)-reconstruction codes for $N\in \{14,11,9,5\}$, where \(\mathcal{B}\) is the single-deletion single-substitution ball function that maps a sequence to the set of all sequences obtainable via one deletion and one substitution. Such a code is defined by the requirement that the intersection size of any two distinct single-deletion single-substitution balls is strictly less than the given number of noisy reads $N$. Note that for any $1\le N<N'$, an \((n, N; \mathcal{B})\)-reconstruction code is also an \((n, N'; \mathcal{B})\)-reconstruction code. It follows that the problem of designing \((n, N; \mathcal{B})\)-reconstruction codes with less redundancy becomes more challenging as $N$ decreases, particularly because the problem for $N=1$ already reduces to the coding problem of single-deletion and single-substitution correcting codes.
To the best of our knowledge, most existing results focus on the case where $N$ is a linear function of $n$, while only a limited number consider constant $N$. When $N=1$, the best known \((n, 1; \mathcal{B})\)-reconstruction codes (single-deletion and single-substitution correcting codes) require  $(4+o(1))\log n$ redundant bits. In this work, we show that this redundancy can be reduced to $3\log n+4$ when $N=5$. As $N$ increases further to $9$ and $11$, the redundancy can be improved to $2\log n+12\log\log n+O(1)$ and $\log n +12\log \log n+O(1)$, respectively. Finally,
for $N=14$, we provide a reconstruction code with $\log n+3$ bits of redundancy, which is only two bits more than the best known \((n, 18; \mathcal{B})\)-reconstruction codes. 
\end{abstract}

\begin{IEEEkeywords}
      Reconstruction codes, deletion, substitution, constant number of reads
\end{IEEEkeywords}

\section{Introduction}\label{I}

\IEEEPARstart{T}{he} \emph{sequence reconstruction problem} as a generalization of \emph{error-correcting codes} was first proposed by Levenshtein \cite{Levenshtein-01-IT-recons}, \cite{Levenshtein-01-JCTA-recons} in 2001. The model of this problem is as follows: a sender  transmits a codeword from a codebook over $N$ identical noisy channels each producing a distinct output. The receiver then attempts to recover the original codeword from these corrupted outputs. The central question in this framework is to determine the minimum number of noisy reads $N$ required to uniquely reconstruct the transmitted codeword. For any transmitted codeword $\bm{x}$, we denote its error ball as $\mathcal{B}(\bm{x})$, which contains a set of all possible channel outputs derived from $\bm{x}$. Levenshtein showed that if the intersection between the error balls of any two distinct codewords has size strictly less than $N$, then the original codeword $\bm{x}$ can be uniquely reconstructed from any $N$ distinct elements in $\mathcal{B}(\bm{x})$. Levenshtein studied this problem for channels involving deletions, insertions, substitutions, and transpositions. Subsequent works have extended these results to multiple insertions~\cite{Sala-17-IT-reconstr-ins,Ye-23-IT,Lan-25}, multiple deletions~\cite{Gabrys-18-IT-reconstr-del,Chrisnata-22-IT-reconstr-del,Pham-25-JCTA-reconstr-del,Lan-25}, and combined substitution–insertion channels~\cite{Abu-Sini-21-IT-recons}.

Recently, the dual problem known as the \emph{reconstruction codes problem} was introduced by Cai et al. \cite{Cai-22-IT-recon-edit} and Chrisnata et al. \cite{Chrisnata-22-IT-reconstr-del}.  This formulation is motivated by 
the applications in DNA storage \cite{Church-12-science-DNA, Goldman-13-nature-DNA, Yazdi-15-TMBMC-DNA, Organick-18-nature-DNA} and racetrack memory \cite{Chee-18-IT}. In this problem, the number of noisy reads $N$ is fixed, and the goal is to design codes with minimal redundancy guaranteeing every codeword can be uniquely reconstructed from the $N$ noisy reads. Thus, the core  question of this problem is to determine the maximum size of the intersection between the error balls of any two distinct codewords in a codebook. This problem has been studied for various channels \cite{Cai-22-IT-recon-edit, Chrisnata-22-IT-reconstr-del, Ye-23-IT, Wu-24-DCC-reconstr, Sun-23-IT-reconstr, Zhang-24-ISIT-reconstr, Sun-24-arXiv}. For more details regarding the reconstruction problem, we refer interested readers to recent surveys \cite{Sabary-24-TMBMC, Yang-25-JSAIT}.

In this paper, we study reconstruction codes problem with a constant number of reads for channels that introduce one deletion and one substitution. We refer to \(\mathcal{B}\) as the \emph{single-deletion single-substitution ball}, which maps a sequence to the set of all sequences obtainable via one deletion and one substitution.  It is well known in \cite{Smagloy-23-IT-DS} that the order of deletions and substitutions is commutative in terms of the final result. Therefore, throughout the paper we will usually choose the order that simplifies our proofs. We define  an \emph{\((n, N; \mathcal{B})\)-reconstruction code} as a set of binary sequences of length \(n\) such that any two distinct codewords \(\bm{x}\) and \(\bm{y}\) satisfy \(|\mathcal{B}(\bm{x}) \cap \mathcal{B}(\bm{y})| < N\). This condition ensures that we can uniquely reconstruct each codeword in a specific \emph{\((n, N; \mathcal{B})\)-reconstruction code} by using $N$ distinct elements from its single-deletion single-substitution ball.

For any two distinct sequences $\bm{x},\bm{y}\in \{0,1\}^n$, the authors in \cite{Li-25} recovered that
\[  
\big| \mathcal{B}(\bm{x}) \cap \mathcal{B}(\bm{y}) \big| \leq (d + s) n + O(1), 
\]  
where \(d\in \{0,1,2\}\) (respectively, \(s\in \{0,2\}\)) denotes the size of the intersection between the single-deletion (respectively, single-substitution) balls of \(\bm{x}\) and \(\bm{y}\). For any case of $(d,s)\in \{0,1,2\}\times \{0,2\}$, the upper bounds of $\big| \mathcal{B}(\bm{x}) \cap \mathcal{B}(\bm{y}) \big|$ have been determined in \cite{Li-25}. \cite{Li-25} demonstrated that when \( N \) is set to \( 4n - 8 \), \( 3n - 4 \), \( 2n + 9 \), \( n+21 \), \( 31 \),  and \( 7 \), the redundancy of binary \((n,N;\mathcal{B})\)-reconstruction codes can be \( 0 \), \( 1 \), \( 2 \), \( \log\log n + 3 \), \( \log n + 1 \), and \( 3\log n + 4 \), respectively. \cite{Cai-25-arxiv-sd} also investigated \( (n, N; \mathcal{B}) \)-reconstruction codes, constructing codes with redundancies of \( 1 \) and \( \log n + O(1) \) for \( N = 4n - 8 \) and \( 18 \), respectively, where the logarithm is on base two. 

Motivated by the work of \cite{Li-25} and \cite{Cai-25-arxiv-sd}, we continue to consider the binary reconstruction codes problem for case $(d,s)=(0,0)$, i.e., the intersections between distinct single-deletion balls and  distinct single-substitution balls are both empty. In this situation, $N$ is a constant number independent of the length $n$. The authors in \cite{Li-25} determined the size of $\mathcal{B}(\bm{x}) \cap \mathcal{B}(\bm{y})$ for this case by analyzing the possible positions where deletions occur in the sequences. We begin by dividing the set $\mathcal{B}(\bm{x}) \cap \mathcal{B}(\bm{y})$ into eighteen sets, which are unnecessarily disjoint just like \cite{Cai-25-arxiv-sd}. From  \cite{Li-25} and \cite{Cai-25-arxiv-sd}, we know that $\big| \mathcal{B}(\bm{x}) \cap \mathcal{B}(\bm{y}) \big|\le 17$ when $(d,s)=(0,0)$, and there exists an \( (n, 18; \mathcal{B}) \)-reconstruction code  with redundancy  \( \log n + 1 \). In this paper, we find that by imposing a constraint with only two additional bits of redundancy being required to the \( (n, 18; \mathcal{B}) \)-reconstruction code, we can obtain an \( (n, 14; \mathcal{B}) \)-reconstruction code with \(\log n + 3\) bits of redundancy. By further restricting the length of periodic substrings in each sequence belonging to our \( (n, 14; \mathcal{B}) \)-reconstruction code, we derive an \( (n, 11; \mathcal{B}) \)-reconstruction code with \(\log n +12 \log\log n + O(1)\) bits of redundancy. Furthermore, by imposing constraints to the differential sequence of each sequence in \( (n, 11; \mathcal{B}) \)-reconstruction code, we obtain an \( (n, 9; \mathcal{B}) \)-reconstruction code with \(2\log n +12 \log\log n + O(1)\) bits of redundancy. Finally, we show that the single-deletion single-substitution list-decodable code with list size two developed in \cite{Song-22-ISIT-LD}, which has redundancy \(3 \log n + 4\) and  is already known to be an \((n, 7; \mathcal{B})\)-reconstruction code in \cite{Li-25},  also qualifies as an \((n, 5; \mathcal{B})\)-reconstruction code.  

The remainder of this paper is organized as follows. Section~\ref{sec:pre} introduces the relevant notations, definitions, main tools, and conclusions used throughout the paper.  
In Section~\ref{sec:key}, we first divide the set $ \mathcal{B}(\bm{x}) \cap \mathcal{B}(\bm{y})$ into eighteen  subsets for any $\bm{x},\bm{y}\in \{0,1\}^n$ that satisfy $(d,s)=(0,0)$. Then, we   present the key theorems that characterize the intersection between two distinct single-deletion single-substitution balls under various constraints on these subsets.  
Section~\ref{sec:constr} describes the construction of \((n, N; \mathcal{B})\)-reconstruction codes for $N=14,11,9$, and $5$.  
Section~\ref{sec:proof} provides the proofs of the theorems stated in Section~\ref{sec:key}.  
Finally, Section~\ref{sec:concl} concludes the paper.   

\section{preliminaries}\label{sec:pre}

\subsection{Notations and Definitions}
Let $[n]$ denote the set $\{1,2,\ldots,n\}$ for any integer $n\ge 1$, and  we use $[i,j]$ to denote the set $\{i,i+1,\ldots,j\}$, for $i,j\in \mathbb{N}$ with $i\le j$,  especially, $[i,j]$ refers to the empty set when $i>j$. Such set $[i,j]$ is called an interval. Let \(\Sigma= \{0,1\}\) be the binary alphabet set, and let \(\Sigma^n\) be the set of all sequences of length \(n\) over \(\Sigma\). Furthermore, we define \(\Sigma^{\ast}= \bigcup_{n\geq 0} \Sigma^n\).  

For a sequence \(\bm{x}\in\Sigma^n\), its \(i\)-th entry is denoted by \(x_i\), and the sequence is written as  \(x_1x_2\cdots x_n\). Moreover, we define its \emph{complement} as \(\overline{\bm{x}}\triangleq \overline{x_1}\overline{x_2}\ldots \overline{x_n}\), where \(\overline{x_i}=1-x_i\) for \(i\in [n]\), and its \emph{reversal} as \(R(\bm{x})\triangleq x_n x_{n-1}\cdots x_1\). Let \(|\bm{x}|\) denote the length of \(\bm{x}\) when \(\bm{x}\) is a sequence, and let \(|\mathcal{C}|\) be the size of \(\mathcal{C}\) when \(\mathcal{C}\) is a set.  

For two sequences \(\bm{x}\in\Sigma^n\) and \(\bm{y}\in\Sigma^m\), we use \(\bm{x}\bm{y}=x_1\cdots x_n y_1\cdots y_m\) to denote their \emph{concatenation}. If there exists a set \(\mathcal{I}=\{i_1,i_2,\ldots,i_m\}\) with \(1\le i_1<i_2<\cdots <i_m \leq n\) such that \(\bm{y}=\bm{x}_{\mathcal{I}}\triangleq x_{i_1}x_{i_2}\cdots x_{i_m}\), we say that \(\bm{y}\) is a \emph{subsequence} of \(\bm{x}\). Specifically, when \(\mathcal{I}\) is an interval, \(\bm{y}\) is called a \emph{substring} of \(\bm{x}\).  

For a positive integer \(t\), we say that \(\bm{x}\in \Sigma^n\) has a \emph{period} of \(t\) or is \(t\)-\emph{periodic} if \(t\) is the smallest integer such that \(x_i=x_{i+t}\) for every \(i\in [n-t]\). Moreover, we say that \(\bm{x}\) is \emph{\(^{\le}t\)-periodic} if it has a period of at most \(t\). A \emph{run} in \(\bm{x}\) is a maximal substring of \(\bm{x}\) of period one, and an \emph{alternating sequence} in \(\bm{x}\) is a maximal substring of \(\bm{x}\) of period two. We use \(r(\bm{x})\) to denote the number of runs in \(\bm{x}\). 

For two sequences \(\bm{x}, \bm{y} \in \Sigma^n\), we define the \emph{weight} of \(\bm{x}\) as \( \mathrm{wt}(\bm{x})\triangleq \sum_{i=1}^n x_i\), which records the number of ones in $\bm{x}$. The \emph{Hamming distance} between \(\bm{x}\) and \(\bm{y}\) is given by  \(d_H(\bm{x}, \bm{y}) \triangleq |\{i\in [n]: x_i\neq y_i\}|\), 
counting the number of coordinates where \(\bm{x}\) and \(\bm{y}\) differ. When \(\bm{x}\) and \(\bm{y}\) are clear from the context, we write \(d_H \triangleq d_H(\bm{x}, \bm{y})\). For \(i \in [d_H]\), let   \(j_i \triangleq j_i(\bm{x}, \bm{y})  \) 
denote the \(i\)-th smallest index where $\bm{x}$ and $\bm{y}$ disagree.
Moreover, define \(\bm{a} \triangleq \bm{a}(\bm{x}, \bm{y}) = \bm{x}_{[1, j_1 - 1]} = \bm{y}_{[1, j_1 - 1]}\) and \(\bm{b} \triangleq \bm{b}(\bm{x}, \bm{y}) = \bm{x}_{[j_{d_H} + 1, n]} = \bm{y}_{[j_{d_H} + 1, n]}\)
as the longest common prefix and suffix of \(\bm{x}\) and \(\bm{y}\), respectively. 

For any sequence \(\bm{x}\in \Sigma^n\), we say $\bm{x}$ incurs a \emph{deletion} if it is transformed into one of its subsequences of length $n-1$, and we say $\bm{x}$ incurs a \emph{substitution} if it is transformed into a sequence $\bm{y}\in \Sigma^n$ satisfying $d_H(\bm{x},\bm{y})\le 1$. The \emph{single-deletion ball} of \(\bm{x}\) is defined as $\mathcal{D}(\bm{x})\triangleq \{\bm{y}\in \Sigma^{n-1}:\bm{y}~\text{is a subsequence of}~\bm{x} \}$, and the \emph{single-substitution ball} as $\mathcal{S}(\bm{x})\triangleq \{\bm{y}\in \Sigma^{n}:d_H(\bm{x},\bm{y})\le 1 \}$. Furthermore, the \emph{single-deletion single-substitution ball} of \(\bm{x}\), defined by \(\mathcal{B}(\bm{x})\), is the set of all sequences obtainable from \(\bm{x}\) after one deletion and one substitution. For simplicity, for any sequences $\bm{x}$ and $\bm{y}$, we define  \(\mathcal{D}(\bm{x},\bm{y})\triangleq \mathcal{D}(\bm{x})\cap \mathcal{D}(\bm{y})\), \(\mathcal{S}(\bm{x},\bm{y})\triangleq \mathcal{S}(\bm{x})\cap \mathcal{S}(\bm{y})\), and \(\mathcal{B}(\bm{x},\bm{y})\triangleq \mathcal{B}(\bm{x})\cap \mathcal{B}(\bm{y})\).

For any ${d_x,e_x}\in [n]$, let $\bm{x}(d_x,e_x)$ stand the sequence obtained from $\bm{x}\in \Sigma^n$ after deleting $x_{d_x}$ and substituting $x_{e_x}$ with $\overline{x_{e_x}}$, if $\bm{x}$ only suffers one deletion (or one substitution), we write it as $\bm{x}(d_x,0)$ (or $\bm{x}(0,e_x)$). A set of sequences \(\mathcal{C}\subseteq \Sigma^n\) is referred to as an \emph{\((n,N;\mathcal{B})\)-reconstruction code} if for any two distinct sequences \(\bm{x},\bm{y}\in \mathcal{C}\), it holds that \(|\mathcal{B}(\bm{x},\bm{y})|<N\). To evaluate a code \(\mathcal{C}\subseteq \Sigma^n\), we calculate its \emph{redundancy} \(\mathrm{red}
(\mathcal{C})\triangleq n-\log |\mathcal{C}|\), where the logarithm is on base two.    

The following definition will be crucial in our code constructions.

\begin{definition}
    For any \( \bm{x}\in\Sigma^n\), and  non-negative integer $k$, the \emph{$k$-th order VT syndrome} of $\bm{x}$ is defined as $\mathrm{VT}^{k}(\bm{x})=\sum_{i=1}^n\sum_{j=1}^i j^{k-1}x_i$ when $k\ge 1$, and $\mathrm{VT}^{0}(\bm{x})=\mathrm{wt}(\bm{x})$ when $k=0$.
\end{definition}

\subsection{Useful Tools and Conclusions}\label{sub:conclu}
In this subsection, we review several well-known tools and conclusions that will be critical for our later discussion.

\subsubsection{Differential Sequence}
For any sequence $ \bm{x}\in \Sigma^n$, let $x_{n+1}=0$, $x_0=0$, its differential sequence is defined as
    \[
    \psi(\bm{x})=(\psi(x)_1,\psi(x)_2,\cdots,\psi(x)_{n+1})\in \Sigma^{n+1},
    \]
    where $\psi(x)_i=x_i\oplus x_{i+1}\triangleq x_i-x_{i-1} \pmod{ 2}$, $1\le i\le n+1$.

Now, we summarize how a single deletion or  a substitution in \(\bm{x}\) affects the number of ones in the differential sequence \(\psi(\bm{x})\).

\begin{observation}\cite[Observation 1]{Sun-24-IT}\label{obs:del}
After a deletion at position \(i\) in \(\bm{x}\), i.e.,
\[
\cdots x_{i-1}x_ix_{i+1}\cdots \longrightarrow \cdots x_{i-1}x_{i+1}\cdots,
\]
the pair \(\psi(x)_i \psi(x)_{i+1}\) becomes \(\psi(x)_i \oplus \psi(x)_{i+1}\). The rest of \(\psi(\bm{x})\) remains unchanged. Thus, the number of ones in \(\psi(\bm{x})\) either decreases by two or stays the same.

\begin{itemize}
    \item If the number of ones remains unchanged, the error type in \(\psi(\bm{x})\) is one of: \(00 \to 0\), \(10 \to 1\), or \(01 \to 1\).
    \item If the number of ones decreases by two, the error type is \(11 \to 0\).
\end{itemize}
\end{observation}

\begin{observation}\cite[Observation 2]{Sun-24-IT}\label{obs:sub}
After a substitution at position \(i\) in \(\bm{x}\), i.e.,
\[
\cdots x_{i-1}x_ix_{i+1}\cdots \longrightarrow \cdots x_{i-1}\overline{x_i}x_{i+1}\cdots,
\]
the pair \(\psi(x)_i \psi(x)_{i+1}\) becomes \(\overline{\psi(x)_i \psi(x)_{i+1}}\). The rest of \(\psi(\bm{x})\) remains unchanged. Thus, the number of ones in \(\psi(\bm{x})\) may change by \(-2\), \(0\), or \(+2\).

\begin{itemize}
    \item If the number of ones increases by two, the error type is \(00 \to 11\).
    \item If the number of ones remains unchanged, the error type in \(\psi(\bm{x})\) is one of: \(01 \to 10\) or \(10 \to 01\).
    \item If the number of ones decreases by two, the error type is \(11 \to 00\).
\end{itemize}
\end{observation}

\subsubsection{Strong-Locally-Balanced Sequence} 
Let $l$ be a positive integer, $0\le \epsilon\le \frac{1}{2}$. A sequence $ \bm{x}\in \Sigma^n$ is called a strong-$(l,\epsilon)$-locally-balanced sequence if for  $1\le i\le n$, and all $l'\ge l$, such that $i+l'-1\le n$, we have $(\frac{1}{2}-\epsilon)l'\le wt(\bm{x}_{[i,i+l'-1]})\le (\frac{1}{2}+\epsilon)l'$. Furthermore, a code $\C\subseteq \Sigma^n$ is termed a strong-($l,\epsilon$)-locally-balanced code, if every $ \bm{x}\in \C$  is a strong-($l,\epsilon$)-locally-balanced sequence. 

In this paper, we will construct an \( (n, 9; \mathcal{B}) \)-reconstruction code by utilizing strong-$(l,\epsilon)$-locally-balanced sequence and differential sequence.

\begin{lemma}\cite[Lemma 1]{Sun-24-IT} \label{lem:N=9_1}
    Suppose $n\ge 10$ and $b=6$, set $l=b^4\log n=1296\log n$, and $\epsilon=\frac{1}{3b}=\frac{1}{18}$, the size of $\mathcal{M}\triangleq \{\bm{x}\in \Sigma^n: \bm{x}~ \text{and}~ \psi(\bm{x})~\text{are strong-}(l,\epsilon)\text{-locally-}\text{balanced sequences}\}$ is at least $2^n\left(1-4n^{2-\frac{2}{9}b^2\log e}\right)$,  especially, this size is no less than $2^{n-1}$. 
\end{lemma}

\subsubsection{Error Correcting Codes}
We first review a construction of binary codes that can correct one deletion or one substitution.

\begin{lemma}\cite[Theorem 2]{Levenshtein-66-SPD-1D}\label{lem:VT_1}
    For any $s \in [0,2n-1]$, define $\mathrm{VT}_{s}(n) = \left\{\bm{x} \in \Sigma^n : \mathrm{VT}^{1}(\bm{x}) \equiv s \pmod{2n}\right\}$, it holds that $\mathcal{D}(\bm{x},\bm{y})=\emptyset$ and $\mathcal{S}(\bm{x},\bm{y})=\emptyset$ for any two distinct sequences $\bm{x}, \bm{y}\in \mathrm{VT}_s(n)$.
\end{lemma}

We now review a construction of single-deletion single-substitution list-decodable codes with list size two.

\begin{lemma}\cite[Definition 1, Theorem 1, and Lemma 1]{Song-22-ISIT-LD}\label{lem:DS_L}
For any $s_0 \in [0,3]$, $s_1 \in [0,2n-1]$, and $s_2 \in [0,2n^2-1]$, define the code 
\begin{align*}
    \mathcal{C}_L\triangleq \big\{\bm{x}\in \Sigma^n :\mathrm{VT}^{0}(\bm{x})\equiv s_0 \pmod{4},~ \mathrm{VT}^k(\bm{x})\equiv s_{k} \pmod{2n^k} \text{ for } k\in \{1,2\} \big\}.
\end{align*}
For any three distinct sequences $\bm{x},\bm{y},\bm{z}\in \mathcal{C}_L$, it holds that $\mathcal{B}(\bm{x})\cap \mathcal{B}(\bm{y})\cap \mathcal{B}(\bm{z})=\emptyset$.
Moreover, if $\bm{x}(d_x,e_x)=\bm{y}(d_y,e_y)$, then $e_x, e_y\in [d_x,d_y]\cup [d_y,d_x]$.
\end{lemma}

\subsubsection{$P$-bounded Error-Correcting Codes}  
A code $\mathcal{C}\subseteq\Sigma^n$ is referred to as a \emph{$P$-bounded single-deletion single-substitution correcting code} if, for any two distinct sequences $\bm{x}=\bm{a}\tilde{\bm{x}}\bm{b},\bm{y}=\bm{a}\tilde{\bm{y}}\bm{b}\in \C$ with $|\tilde{\bm{x}}|=|\tilde{\bm{y}}|\le P$, it holds that $\mathcal{B}(\bm{x},\bm{y})=\emptyset$.

Then, we review a construction of $P$-bounded single-deletion single-substitution correcting codes.

\begin{lemma}\cite[Lemma 10]{Sun-24-IT}  \label{lem:N=11}
Let $P<n$, for fixed integers \(g_1, g_1' \in [0,6P-1]\), \(g_2, g_2' \in [0,12P^2-1]\), and \(g_3, g_3' \in [0,24P^3-1]\), we define \(\C_{\mathcal{DS}}^P \subseteq \Sigma^n\) satisfying the following constraints:
\begin{itemize}
    \item $\sum_i\mathrm{VT}^{k}(\bm{x}^{(P,2i-1)}\bm{x}^{(P,2i)})\equiv g_k\pmod{3(2P)^k}$, $k\in \{1,2,3\}$;
    \item $\sum_i\mathrm{VT}^{k}(\bm{x}^{(P,2i)}\bm{x}^{(P,2i+1)})\equiv g_k'\pmod{3(2P)^k}$, $k\in \{1,2,3\}$;
\end{itemize}
    where $\bm{x}^{(P,i)}=\bm{x}_{[(i-1)P+1,iP]}$ if $1\le i\le \left \lfloor n/P \right \rfloor$ and $\bm{x}^{(P,\left \lfloor n/P \right \rfloor+1)}=\bm{x}_{[\left \lfloor n/P \right \rfloor P+1,n]}$. Then $\C_{\mathcal{DS}}^P$ is a $P$-bounded single-deletion single-substitution correcting code.
\end{lemma}

\subsubsection{Error Ball}
 We first consider the intersection sizes between distinct single-substitution balls and single-deletion balls, respectively.

\begin{lemma}\cite[Corollary 1]{Levenshtein-01-IT-recons}\label{lem:sub}
    Let \(\bm{x}\) and \(\bm{y}\) be distinct sequences in \(\Sigma^{n}\). It holds that \(|S(\bm{x},\bm{y})|\in \{0,2\}\). Specifically,
    \begin{itemize}
        \item If \(d_H(\bm{x},\bm{y})=1\), we can write \(\bm{x}=\bm{a}\alpha \bm{b}\) and \(\bm{y}= \bm{a} \overline{\alpha} \bm{b}\) for some \(\alpha\in \Sigma\). Then, \(|S(\bm{x},\bm{y})|=2\) and \(S(\bm{x},\bm{y})= \{\bm{a}\alpha \bm{b}, \bm{a}\overline{\alpha} \bm{b}\}\).
        \item If \(d_H(\bm{x},\bm{y})=2\), we can write \(\bm{x}=\bm{a}\alpha \bm{c} \beta \bm{b}\) and \(\bm{y}= \bm{a} \overline{\alpha} \bm{c} \overline{\beta} \bm{b}\) for some \(\alpha, \beta\in \Sigma\) and \(\bm{c}\in \Sigma^{\ast}\). Then, \(|S(\bm{x},\bm{y})|=2\) and \(S(\bm{x},\bm{y})= \{\bm{a}\overline{\alpha} \bm{c} \beta \bm{b}, \bm{a}\alpha \bm{c} \overline{\beta} \bm{b}\}\).
        \item If \(d_H(\bm{x},\bm{y})\geq 3\), then \(|S(\bm{x},\bm{y})|=0\).
    \end{itemize}
\end{lemma}

\begin{lemma}\cite[Lemma 2.12]{Li-25}\label{lem:del_int} 
    Let \(\bm{x}\) and \(\bm{y}\) be distinct sequences in \(\Sigma^{n}\). We have that $|\mathcal{D}(\bm{x},\bm{y})|=0$ if and only if $d_H(\bm{x}_{[j_1+1,j_{d_H}]},\bm{y}_{[j_1,j_{d_H}-1]})\ge 1$  and $d_H(\bm{x}_{[j_1,j_{d_H}-1]},\bm{y}_{[j_1+1,j_{d_H}]})\ge 1$.
\end{lemma}

Then, we review a property about the periodicity of sequences.

\begin{lemma}\label{lem:period}\cite[Lemma 3.7]{Sun-23-IT-reconstr}
 For two distinct sequences $\bm{x}$ and $\bm{y}$ in $\Sigma^m$, if $\bm{x}_{[1,m-1]}=\bm{y}_{[2,m]}$ and $\bm{x}_{[2,m]}=\bm{y}_{[1,m-1]}$, then $\bm{x}$ and $\bm{y}$ are alternating sequences.
\end{lemma}

The following lemma characterizes the Hamming distance between distinct sequences within a deletion ball. 
\begin{lemma}\label{lem:del}\cite[Lemma 5]{Abu-Sini-21-IT-recons}
    Let $\bm{x}\in \Sigma^n$ be a word consisting of $r=r(\bm{x})$ runs. For $1\le i\le r$, let $\bm{x}(i)$ be the word received by deleting a bit from the $i$-th run, then it holds that $\mathcal{D}(\bm{x})=\{\bm{x}(1),\bm{x}(2),\cdots,\bm{x}(r)\}$ and $d_{H}(\bm{x}(i),\bm{x}(j))=j-i $.
\end{lemma}

Based on Lemma \ref{lem:del}, we can derive the following conclusion.
\begin{lemma}\label{lem:run_number}
For sequences $\alpha\bm{w}\beta,\overline{\alpha}\bm{w}\overline{\beta}\in \Sigma^{*}$, where $\alpha,\beta\in \{0,1\}$ and $|\bm{w}|\ge 0$, let 
\begin{align*}
    d_1\triangleq d_H(\bm{w}\beta, \overline{\alpha}\bm{w}) ,\quad d_2\triangleq d_H(\alpha\bm{w}, \bm{w}\overline{\beta}).
\end{align*}
Then the following results hold:
\begin{itemize}
    \item $\alpha=\beta$ if and only if $d_1, d_2$ are both odd;
    \item $\alpha\neq\beta$ if and only if $d_1, d_2$ are both even.
\end{itemize}
\end{lemma}
	
\begin{IEEEproof}
Let $\bm{u}=\overline{\alpha}\bm{w}\beta,\bm{v}=\alpha\bm{w}\overline{\beta}$, by Lemma \ref{lem:del}, we have that $d_1=r(\bm{u})-1$ and $d_2=r(\bm{v})-1$. Below we only prove the first statement, as the second statement can be discussed similarly.

If $\alpha=\beta$, then $\overline{\alpha}\neq \beta$ and $\alpha\neq \overline{\beta}$.
This implies that $r(\bm{u}),r(\bm{v})\geq 2$ are even and then $d_1, d_2$ are odd. On the other hand, if $d_1, d_2$ are odd, we have that $r(\bm{u}),r(\bm{v})\geq 2$ are even.
This implies that $\overline{\alpha}\neq \beta$, that is, $\alpha=\beta$.
\end{IEEEproof}

\begin{remark}\label{rem:run_number}
From the proof of Lemma \ref{lem:run_number}, we know that $d_1=r(\overline{\alpha}\bm{w}\beta)-1$ and $d_2=r(\alpha\bm{w}\overline{\beta})-1$. Then the following conclusions can be easily checked: 
\begin{itemize}
    \item When $\alpha=\beta$, we have $d_1=d_2=1$ if and only if $r(\bm{w})\le 1$, and  $d_1,d_2\ge 3$ if and only if $r(\bm{w})\ge 3$. Moreover, we have $d_1\neq d_2\in\{1,3\}$ if and only if $r(\bm{w})=2$.
   \item When $\alpha\neq \beta$, we have $d_1=d_2=0$ if and only if $r(\bm{w})=0$, $d_1=d_2=2$ if and only if $r(\bm{w})=2$, and  $d_1,d_2\ge 4$ if and only if $r(\bm{w})\ge 4$. Moreover, $d_1\neq d_2\in\{0,2\}$ if and only if $r(\bm{w})=1$, and $d_1\neq d_2\in\{2,4\}$ if and only if $r(\bm{w})=3$.
\end{itemize}
\end{remark}

Then we can also obtain the following conclusion.
\begin{corollary}\label{cor:run_number}
  For any two distinct sequences $\bm{u},\bm{v}\in \Sigma^*$, where $|\bm{u}|=|\bm{v}|$. Let $d_H\triangleq d_H(\bm{u},\bm{v})$, then for any $1\le i<i'\le d_H$, we can obtain that  $d_H(\bm{u}_{[j_i,j_{i'}-1]},\bm{v}_{[j_i+1,j_{i'}]})\ge 0$ is even if $d_H(\bm{u}_{[j_i+1,j_{i'}]},\bm{v}_{[j_i,j_{i'}-1]})=0$, and $d_H(\bm{u}_{[j_i+1,j_{i'}]},\bm{v}_{[j_i,j_{i'}-1]})\ge 0$ is even if $d_H(\bm{u}_{[j_i,j_{i'}-1]},\bm{v}_{[j_i+1,j_{i'}]})=0$.
\end{corollary}

\begin{IEEEproof}
  For any $1\le i<i'\le d_H$, we know that    
\begin{align*}
d_H(\bm{u}_{[j_i+1,j_{i'}]},\bm{v}_{[j_i,j_{i'}-1]})= d_H(\bm{u}_{[j_i+1,j_{i+1}]},\bm{v}_{[j_i,j_{i+1}-1]})+\cdots +d_H(\bm{u}_{[j_{i'-1}+1,j_{i'}]},\bm{v}_{[j_{i'-1},j_{i'}-1]}),  
\end{align*}
and
\begin{align*}
d_H(\bm{u}_{[j_i,j_{i'}-1]},\bm{v}_{[j_i+1,j_{i'}]})= d_H(\bm{u}_{[j_i,j_{i+1}-1]},\bm{v}_{[j_i+1,j_{i+1}]})+\cdots +d_H(\bm{u}_{[j_{i'-1},j_{i'}-1]},\bm{v}_{[j_{i'-1}+1,j_{i'}]}). 
\end{align*}
Thus, if $d_H(\bm{u}_{[j_i+1,j_{i'}]},\bm{v}_{[j_i,j_{i'}-1]})=0$, then we can obtain that $u_{j_l}\neq u_{j_{l+1}}$ for all $l\in [i,i'-1]$ according to Lemma \ref{lem:run_number}. Consequently, we have $d_H(\bm{u}_{[j_i,j_{i'}-1]},\bm{v}_{[j_i+1,j_{i'}]})\ge 0$ is even. Similarly, by considering $R(\bm{u})$ and $R(\bm{v})$, we can obtain that $d_H(\bm{u}_{[j_i+1,j_{i'}]},\bm{v}_{[j_i,j_{i'}-1]})\ge 0$ is even if $d_H(\bm{u}_{[j_i,j_{i'}-1]},\bm{v}_{[j_i+1,j_{i'}]})=0$. 
\end{IEEEproof}

\section{key conclusions}\label{sec:key}
In this section, we provide several key lemmas, which serve to the construction of our reconstruction codes in the next section. In addition, we always assume $|\mathcal{D}(\bm{x},\bm{y})|=0$ and $|\mathcal{S}(\bm{x},\bm{y})|= 0$ for any two distinct sequences $\bm{x},\bm{y}\in \Sigma^n$ throughout the remainder, unless otherwise specified.

 Since $|\mathcal{S}(\bm{x},\bm{y})|=0$, then by Lemma \ref{lem:sub}, we have $d_H(\bm{x},\bm{y})\ge 3$. On the other hand, since $|\mathcal{D}(\bm{x},\bm{y})|=0$, then by definition, we have $d_H(\bm{x}(d_x,0),\bm{y}(d_y,0))\ge 1$ for any $(d_x,d_y)\in [n]\times [n]$. Moreover, if $\mathcal{B}(\bm{x},\bm{y})$ is not empty, then by Lemma \ref{lem:sub}, there exists some $(d_x,d_y)\in [n]\times [n]$ such that $d_H(\bm{x}(d_x,0),\bm{y}(d_y,0)) \in \{1,2\}$.
  Note that $d_x\neq d_y$ when $d_H(\bm{x}(d_x,0),\bm{y}(d_y,0))=1$, otherwise,  we have $d_H(\bm{x},\bm{y})=2$, which leads a contradiction. Thus, we default $d_x\neq d_y$ when $d_H(\bm{x}(d_x,0),\bm{y}(d_y,0))=1$ in the following discussion. For any $d_x,d_y\in [n]$, let 
 \begin{align*}
     E\triangleq \big\{(d_x,d_y):d_H(\bm{x}(d_x,0),\bm{y}(d_y,0)) \in \{1,2\}\big\}.
 \end{align*}
  Then, we can obtain that
 \begin{align*}
   \mathcal{B}(\bm{x},\bm{y})=\bigcup_{ (d_x,d_y)\in E} \mathcal{S}(\bm{x}(d_x,0),\bm{y}(d_y,0)). 
 \end{align*}
On the other hand, we note that
\begin{align*}
   d_H(\bm{x}(d_x,0),\bm{y}(d_y,0))=d_H(\bm{x}_{[1,d_x-1]},\bm{y}_{[1,d_x-1]})+ d_H(\bm{x}_{[d_x+1,d_y]},\bm{y}_{[d_x,d_y-1]})+d_H(\bm{x}_{[d_y+1,n]},\bm{y}_{[d_y+1,n]}),
\end{align*}
when $d_x\le d_y$, and
\begin{align*}
   d_H(\bm{x}(d_x,0),\bm{y}(d_y,0))=d_H(\bm{x}_{[1,d_y-1]},\bm{y}_{[1,d_y-1]})+ d_H(\bm{x}_{[d_y,d_x-1]},\bm{y}_{[d_y+1,d_x]})+d_H(\bm{x}_{[d_x+1,n]},\bm{y}_{[d_x+1,n]}),
\end{align*}
when $d_x\ge d_y$. Thus, the set $E$ can be determined by the values of tuples
\begin{align*}
   \big(d_H(\bm{x}_{[1,d_x-1]},\bm{y}_{[1,d_x-1]}), d_H(\bm{x}_{[d_x+1,d_y]},\bm{y}_{[d_x,d_y-1]}),d_H(\bm{x}_{[d_y+1,n]},\bm{y}_{[d_y+1,n]})\big),
\end{align*}
and
\begin{align*}
   \big(d_H(\bm{x}_{[1,d_y-1]},\bm{y}_{[1,d_y-1]}),d_H(\bm{x}_{[d_y,d_x-1]},\bm{y}_{[d_y+1,d_x]}),d_H(\bm{x}_{[d_x+1,n]},\bm{y}_{[d_x+1,n]})\big).
\end{align*}
Correspondingly, the set $\mathcal{B}(\bm{x},\bm{y})$ can be determined.
The values of these two tuples are summarized in Table \ref{tab:tuple}.\footnote{In fact, the values of these two tuples have also been studied in \cite{Cai-25-arxiv-sd} and summarized in Table 1. Comparing with the definitions in \cite{Cai-25-arxiv-sd}, it is clear that the tuple $\big(d_H(\bm{x}_{[1,j-1]},\bm{y}_{[1,j-1]}), d_H(\bm{x}_{[j+1,j']},\bm{y}_{[j,j'-1]}),d_H(\bm{x}_{[j'+1,n]},\bm{y}_{[j'+1,n]})\big)$ in this paper is equivalent to the tuple $\big(|S\cap[j-1]|,|T^L\cap[j+1,j']|,|S\cap [j'+1,n]|\big)$ in \cite{Cai-25-arxiv-sd}. Similarly, the tuple $\big(d_H(\bm{x}_{[1,j-1]},\bm{y}_{[1,j-1]}),d_H(\bm{x}_{[j,j'-1]},\bm{y}_{[j+1,j']}),d_H(\bm{x}_{[j'+1,n]},\bm{y}_{[j'+1,n]})\big)$ is equivalent to the tuple $\big(|S\cap[j-1]|,|T^R\cap[j+1,j']|,|S\cap [j'+1,n]|\big)$, where $j\le j'$.} 

\begin{table*}[ht]
\caption{In this table, we display the possible values of two tuples. When $d_x\le d_y$, let $d_1\triangleq d_H(\bm{x}_{[1,d_x-1]},\bm{y}_{[1,d_x-1]})$, $d_2\triangleq d_H(\bm{x}_{[d_x+1,d_y]},\bm{y}_{[d_x,d_y-1]})$ and $d_3 \triangleq d_H(\bm{x}_{[d_y+1,n]},\bm{y}_{[d_y+1,n]})$. Moreover, when $d_x\ge d_y$, let $d_1\triangleq d_H(\bm{x}_{[1,d_y-1]},\bm{y}_{[1,d_y-1]})$, $d_2\triangleq d_H(\bm{x}_{[d_y,d_x-1]},\bm{y}_{[d_y+1,d_x]})$ and $d_3 \triangleq d_H(\bm{x}_{[d_x+1,n]},\bm{y}_{[d_x+1,n]})$. Note that we have $d_1+d_3=2$ when $d_x=d_y$.} 
\centering
\begin{tabular}{|c|c|c|c|c|c|}  
\hline
 \multicolumn{3}{|c}{$d_x\le d_y$} & \multicolumn{3}{|c|}{$d_x\ge d_y$}\\
\hline 
$d_1$ & $d_2$& $d_3$ & $d_1$ & $d_2$ & $d_3$\\
\hline
$1$ & $0$ & $0$ & $1$ & $0$ & $0$\\
\hline
$0$ & $0$ & $1$ & $0$ & $0$ & $1$\\
\hline
$0$ & $1$ & $0$ & $0$ & $1$ & $0$\\
\hline 
$2$ & $0$ & $0$ & $2$ & $0$ & $0$ \\
\hline
$0$ & $0$ & $2$ & $0$ & $0$ & $2$ \\
\hline
$1$ & $0$ & $1$ & $1$ & $0$ & $1$ \\
\hline
$1$ & $1$ & $0$ & $1$ & $1$ & $0$ \\
\hline
$0$ & $1$ & $1$ & $0$ & $1$ & $1$ \\
\hline
$0$ & $2$ & $0$ & $0$ & $2$ & $0$ \\
\hline
\end{tabular}
\label{tab:tuple}
\end{table*}

 Suppose $\bm{z}=\bm{x}(d_x,e_x)=\bm{y}(d_y,e_y)$, according to the Table \ref{tab:tuple} and Lemma \ref{lem:sub}, we can divide $\mathcal{B}(\bm{x},\bm{y})$ into the following eighteen subsets, which are unnecessarily disjoint.
\begin{align*}
  &B_1\triangleq \big\{\bm{z}\in \mathcal{B}(\bm{x},\bm{y}):e_y=0, e_x<d_x<d_y~\text{or}~e_x=0, e_y<d_x<d_y\big\}, \\ 
   &B_2\triangleq \big\{\bm{z}\in \mathcal{B}(\bm{x},\bm{y}):e_y=0, e_x<d_y<d_x~\text{or}~e_x=0, e_y<d_y<d_x\big\},\\
   &B_3\triangleq \big\{\bm{z}\in \mathcal{B}(\bm{x},\bm{y}):e_y=0, d_x<d_y<e_x~\text{or}~e_x=0, d_x<d_y<e_y\big\},\\
   &B_4\triangleq \big\{\bm{z}\in \mathcal{B}(\bm{x},\bm{y}):e_y=0, d_y<d_x<e_x~\text{or}~e_x=0, d_y<d_x<e_y\big\},\\
   &B_5\triangleq \big\{\bm{z}\in \mathcal{B}(\bm{x},\bm{y}):e_y=0, d_x<e_x\leq d_y~\text{or}~e_x=0, d_x\leq e_y<d_y\big\},\\
   &B_6\triangleq \big\{\bm{z}\in \mathcal{B}(\bm{x},\bm{y}):e_y=0, d_y\leq e_x< d_x~\text{or}~e_x=0, d_y< e_y\le d_x\big\},\\
   &B_7\triangleq \big\{\bm{z}\in \mathcal{B}(\bm{x},\bm{y}):e_x< d_x\le d_y~\text{and}~e_y< d_x\le d_y~\text{and}~e_x, e_y\neq 0\big\},\\
   &B_8\triangleq \big\{\bm{z}\in \mathcal{B}(\bm{x},\bm{y}):e_x< d_y\le d_x~\text{and}~e_y< d_y\le d_x~\text{and}~e_x, e_y\neq 0\big\},\\
   &B_9\triangleq \big\{\bm{z}\in \mathcal{B}(\bm{x},\bm{y}): d_x\le d_y<e_x~\text{and}~d_x\le d_y<e_y~\text{and}~e_x, e_y\neq 0\big\},\\
   &B_{10}\triangleq \big\{\bm{z}\in \mathcal{B}(\bm{x},\bm{y}):d_y\le d_x<e_x~\text{and}~ d_y\le d_x<e_y~\text{and}~e_x, e_y\neq 0\big\},\\
   &B_{11}\triangleq \big\{\bm{z}\in \mathcal{B}(\bm{x},\bm{y}):e_x<d_x\leq d_y<e_y~\text{or}~e_y<d_x\leq d_y<e_x\big\}, \\
   &B_{12}\triangleq \big\{\bm{z}\in \mathcal{B}(\bm{x},\bm{y}):e_x<d_y\leq d_x<e_y~\text{or}~e_y<d_y\leq d_x<e_x\big\},\\
   &B_{13}\triangleq \big\{\bm{z}\in \mathcal{B}(\bm{x},\bm{y}):e_x<d_x\le e_y<d_y~\text{or}~e_y<d_x< e_x\le d_y\big\},\\
   &B_{14}\triangleq \big\{\bm{z}\in \mathcal{B}(\bm{x},\bm{y}):e_x<d_y< e_y\le d_x~\text{or}~e_y<d_y\le e_x< d_x\big\},
   \end{align*}
   \begin{align*}
   &B_{15}\triangleq \big\{\bm{z}\in \mathcal{B}(\bm{x},\bm{y}):d_x<e_x\le d_y<e_y~\text{or}~d_x\le e_y< d_y<e_x\big\},\\
   &B_{16}\triangleq \big\{\bm{z}\in \mathcal{B}(\bm{x},\bm{y}):d_y\le e_x< d_x<e_y~\text{or}~d_y<e_y\le d_x<e_x\big\},\\
   &B_{17}\triangleq \big\{\bm{z}\in \mathcal{B}(\bm{x},\bm{y}):d_x<e_x\le d_y~\text{and}~d_x\le e_y<d_y~\text{and}~e_x, e_y\neq 0\big\},\\
   &B_{18}\triangleq \big\{\bm{z}\in \mathcal{B}(\bm{x},\bm{y}):d_y\le e_x< d_x~\text{and}~d_y< e_y\le d_x~\text{and}~e_x, e_y\neq 0\big\}.\\
    \end{align*}  
In other words, the following lemma holds. 

\begin{lemma}\label{lem:sets}
    Let $\bm{x}, \bm{y} \in \Sigma^n$ be two distinct sequences such that $|\mathcal{D}(\bm{x},\bm{y})|=0$ and $|\mathcal{S}(\bm{x},\bm{y})|=0$. If $|\mathcal{B}(\bm{x},\bm{y})|\neq 0$, then  
    \[
    \mathcal{B}(\bm{x},\bm{y}) = \bigcup_{k=1}^{18} B_k.
    \]
\end{lemma}

 The following lemma suggests that we can forbid four sets so as to simplify  this problem by adding only two bits of redundancy. Moreover, the proof of this lemma will be provided in Appendix A.

\begin{lemma}\label{lem:VT}
   Let 
   \begin{align*}
   \C_1\triangleq\left\{\bm{x}\in \Sigma^n :~  \mathrm{VT}^{0}(\bm{x})\equiv s_0  \pmod{4},~ \mathrm{VT}^{1}(\bm{x})\equiv s_1\pmod{2n} \right\}, 
   \end{align*} 
where $s_0 \in [0,3]$ and $s_1 \in [0,2n-1]$. Then, for any two distinct sequences $\bm{x},\bm{y}\in \C_1$, we have $|\mathcal{D}(\bm{x},\bm{y})|=0$ and $|\mathcal{S}(\bm{x},\bm{y})|=0$, and $d_H(\bm{x},\bm{y})\ge 4$. Moreover, 
\[
    \mathcal{B}(\bm{x},\bm{y}) = \bigcup_{k=5}^{18} B_k.
    \]
\end{lemma}

Based on Lemma \ref{lem:VT}, we refine our analysis of $\mathcal{B}(\bm{x},\bm{y})$ by the constrained code $\C_1$. Under this setting, only the remaining fourteen sets $B_i$ for $5 \leq i \leq 18$ need to be considered.

The following theorems and lemmas will be proved in Section \ref{sec:proof}. Recall that for any two distinct sequences $\bm{x}$ and $\bm{y}$ with $d_H\triangleq d_H(\bm{x},\bm{y})$, we define \(\bm{a} \triangleq \bm{a}(\bm{x}, \bm{y}) = \bm{x}_{[1, j_1 - 1]} = \bm{y}_{[1, j_1 - 1]}\) and \(\bm{b} \triangleq \bm{b}(\bm{x}, \bm{y}) = \bm{x}_{[j_{d_H} + 1, n]} = \bm{y}_{[j_{d_H} + 1, n]}\)
as the longest common prefix and suffix of \(\bm{x}\) and \(\bm{y}\), respectively. In the subsequent discussions, we will use $\tilde{\bm{x}}$ and $\tilde{\bm{y}}$ to represent $\bm{x}_{[j_1, j_{d_H}]}$ and $\bm{y}_{[j_1, j_{d_H}]}$, respectively.

\begin{theorem}\label{th:B5_nonempty}
    For any two distinct sequences $\bm{x}, \bm{y} \in \C_1$, let $\bm{x}=\bm{a}\tilde{\bm{x}}\bm{b}$ and $\bm{y}=\bm{a}\tilde{\bm{y}}\bm{b}$ with $d_H\triangleq d_H(\bm{x},\bm{y})$. If $B_5$ or $B_6$ is a nonempty set, then we have $|\mathcal{B}(\bm{x},\bm{y})| \le 13$. Moreover, one of the following holds:
\begin{itemize}
  \item $\tilde{\bm{x}},\tilde{\bm{y}}$ are combinations of at most four $^{\le}2$-periodic sequences;
  \item $|\mathcal{B}(\bm{x},\bm{y})| \le 10$.
\end{itemize}
\end{theorem}

\begin{theorem}\label{th:B11_nonempty}
    For any two distinct sequences $\bm{x},\bm{y}$ in $\C_1$, let $\bm{x}=\bm{a}\tilde{\bm{x}}\bm{b}$ and $\bm{y}=\bm{a}\tilde{\bm{y}}\bm{b}$ and $d_H\triangleq d_H(\bm{x},\bm{y})$. If $B_5$ and $B_6$ are two empty sets, and $B_{11}$ or $B_{12}$ is a nonempty set, then we have $|\mathcal{B}(\bm{x},\bm{y})| \le 12$. Moreover, one of the following holds:
\begin{itemize}
  \item $\tilde{\bm{x}},\tilde{\bm{y}}$ are combinations of at most five $^{\le}2$-periodic sequences and two additional symbols from $\{0,1\}$;
  \item $|\mathcal{B}(\bm{x},\bm{y})| \le 10$.
\end{itemize}
\end{theorem}

\begin{theorem}\label{th:B7_nonempty}
    For any two distinct sequences $\bm{x},\bm{y}\in \C_1$, let $\bm{x}=\bm{a}\tilde{\bm{x}}\bm{b}, \bm{y}=\bm{a}\tilde{\bm{y}}\bm{b}$ and $d_H\triangleq d_H(\bm{x},\bm{y})$.  If $B_5$ and $B_6$ are two empty sets, and $B_{7}$ or $B_{8}$ is a nonempty set, then $|\mathcal{B}(\bm{x},\bm{y})|\le 13$. Moreover, one of the following holds:
\begin{itemize}
  \item $\tilde{\bm{x}}$ and $\tilde{\bm{y}}$ are combinations of at most six $^{\le}2$-periodic sequences and four additional symbols from $\{0,1\}$;
  \item $|\mathcal{B}(\bm{x},\bm{y})|\le 8$.
\end{itemize} 
\end{theorem}
        
By considering $R(\bm{x})$ and $R(\bm{y})$, and using the conclusions of Theorem \ref{th:B7_nonempty}, we can also obtain Corollary \ref{cor:B9B10_nonempty}.
\begin{corollary}\label{cor:B9B10_nonempty}
    For any two distinct sequences $\bm{x},\bm{y}\in \C_1$, let $\bm{x}=\bm{a}\tilde{\bm{x}}\bm{b}, \bm{y}=\bm{a}\tilde{\bm{y}}\bm{b}$ and $d_H\triangleq d_H(\bm{x},\bm{y})$. If $B_5$ and $B_6$ are two empty sets, and $B_{9}$ or $B_{10}$ is a nonempty set, then $|\mathcal{B}(\bm{x},\bm{y})|\le 13$. Moreover, one of the following holds:
  \begin{itemize}
   \item $\tilde{\bm{x}}$ and $\tilde{\bm{y}}$ are combinations of at most six $^{\le}2$-periodic sequences and four additional symbols from $\{0,1\}$;
  \item $|\mathcal{B}(\bm{x},\bm{y})|\le 8$.
  \end{itemize} 
\end{corollary}

\begin{lemma}\label{lem:B12_empty}
    For any two distinct sequences $\bm{x},\bm{y}\in \C_1$, if $B_k$ with $5\le k\le 12$ are empty, then $|\mathcal{B}(\bm{x},\bm{y})|\le 8$.
\end{lemma}

\begin{lemma}\label{lem:B16_empty}
    For any two distinct sequences $\bm{x},\bm{y}\in \C_1$, if $B_k$ with $5\le k\le 16$ are empty, then $|\mathcal{B}(\bm{x},\bm{y})|\le 4$.
\end{lemma}

\begin{remark}\label{rem:CONST}
  From these theorems and lemmas above, we can conclude that $|\mathcal{B}(\boldsymbol{x}) \cap \mathcal{B}(\boldsymbol{y})|\leq 13$ for any two different sequences $\bm{x},\bm{y}\in \C_1$. Furthermore,  let $\bm{x}=\bm{a}\tilde{\bm{x}}\bm{b}, \bm{y}=\bm{a}\tilde{\bm{y}}\bm{b}$, then we have the following conclusions:
  \begin{itemize}
 \item If $|\mathcal{B}(\bm{x},\bm{y})|>10$, then $\tilde{\bm{x}}$ and $\tilde{\bm{y}}$ are combinations of at most six $^{\le}2$-periodic sequences and four additional symbols from $\{0,1\}$.
 \item When $B_5,B_6,B_{11},B_{12}$ are empty, if $|\mathcal{B}(\bm{x},\bm{y})|>8$, then $\tilde{\bm{x}}$ and $\tilde{\bm{y}}$ are combinations of at most six $^{\le}2$-periodic sequences and four additional symbols from $\{0,1\}$.
 \end{itemize}
\end{remark}

\section{reconstruction codes}\label{sec:constr}
In this section, we will provide several $(n,N;\mathcal{B})$-reconstruction codes for $N\in \{14,11,9,5\}$.

 \subsection{Reconstruction Codes for N=14}
    
 Based on Remark \ref{rem:CONST}, we know that for {any two distinct sequences} \( \bm{x}, \bm{y} \in \C_1\), we have $|\mathcal{B}(\boldsymbol{x}) \cap \mathcal{B}(\boldsymbol{y})|\leq 13$, therefore, we have the following theorem.
			
\begin{theorem}\label{th:N=14}
    For the given integers $s_0\in [0,3]$, $s_1\in [0,2n-1]$, we define $\C(n,14,\mathcal{B})\subseteq \Sigma^n$ as
    \[
    \C(n,14,\mathcal{B})=\C_1\triangleq\left\{\bm{x}\in \Sigma^n :~  \mathrm{VT}^{0}(\bm{x})\equiv s_0  \pmod{4},~ \mathrm{VT}^{1}(\bm{x})\equiv s_1\pmod{2n} \right\}.
    \]
    Then, $\C(n,14,\mathcal{B})$ is an $(n,14;\mathcal{B})$-reconstruction code and by the pigeonhole principle, there exists some choice of $s_0$ and $s_1$ such that $\mathrm{red}(\C_1)\le \log n+3$.
\end{theorem}
\begin{IEEEproof}
    For any two distinct sequences $\bm{x},\bm{y}\in \C_1$, by Remark \ref{rem:CONST}, we can obtain that $|\mathcal{B}(\boldsymbol{x}) \cap \mathcal{B}(\boldsymbol{y})|\leq 13$. Thus, $\C(n,14,\mathcal{B})$ forms an $(n,14;\mathcal{B})$-reconstruction code. Applying the pigeonhole principle, there exist $s_0$ and $s_1$ such that $|\C_1|\ge \frac{2^n}{4\cdot 2n}$, which implies $\mathrm{red}(\C_1)\le \log n+3$. This completes the proof.
\end{IEEEproof}  

 \subsection{Reconstruction Codes for N=11}

In this subsection we will give an $(n,11,\mathcal{B})$-reconstruction code by restricting the lengths of  periodic substrings in each sequence from the $(n,14,\mathcal{B})$-reconstruction code and utilizing the $P$-bounded single-deletion single-substitution correcting code $\C_{\mathcal{DS}}^P$, which are mentioned in Section \ref{sub:conclu}.

We use \(\mathcal{R}(n, t', t)\) to denote the set of all binary sequences of length \(n\) where each substring with a period at most \(t'\) has a length at most \(t\).  The following lemma gives the size of \(\mathcal{R}(n, t', t)\). Moreover, algorithms that encode an arbitrary sequence into a sequence belonging to \(\mathcal{R}(n, t', t)\) have been given by \cite{Kobovich-26-IT-constraint}.

\begin{lemma}\cite[Theorem 13]{Chee-18-IT}\label{lem:N=9_2}
 For all positive integers $n,t'$ and $t$, we have
 \begin{align*}
 \left|\mathcal{R}(n, t', t)\right|\ge 2^n\left(1-n(\frac{1}{2})^{t-t'}\right).    
 \end{align*}
 In particular, for $t= \lceil \log n   \rceil+t'+1 $, the size of the set $\mathcal{R}(n, t', t)$ is at least $2^{n-1}$.
\end{lemma}

By using the conclusion of Lemmas \ref{lem:N=11} and \ref{lem:N=9_2}, and Theorem \ref{th:N=14}, we can give an $(n,11;\mathcal{B})$-reconstruction code with $\log n+12\log\log n+O(1)$ bits of redundancy.
\begin{theorem}\label{th:N=11}
    Let $t= \lceil \log n   \rceil+3 $, $P\triangleq 6t+4$, for fixed sequences $\bm{s}=s_0s_1, \bm{g}=g_1g_2g_3$ and $\bm{g}'=g_1'g_2'g_3'$, where  $s_0\in [0,3]$, $s_1\in [0,2n-1]$, \(g_1, g_1' \in [0,6P-1]\), \(g_2, g_2' \in [0,12P^2-1]\), and \(g_3, g_3' \in [0,24P^3-1]\),  we define the code $\C(n,11,\mathcal{B})\subseteq \Sigma^n$ such that any sequence $\bm{x}\in \C(n,11,\mathcal{B})$ satisfies the following conditions:
    \begin{itemize}
    \item $\bm{x}\in \mathcal{R}(n,2,t)$;
    \item $\mathrm{VT}^{0}(\bm{x})\equiv s_0  \pmod{4}$;
    \item $\mathrm{VT}^{1}(\bm{x})\equiv s_1\pmod{2n}$;  
    \item $\sum_i\mathrm{VT}^{k}(\bm{x}^{(P,2i-1)}\bm{x}^{(P,2i)})\equiv g_k\pmod{3(2P)^k}$, $k\in \{1,2,3\}$;
    \item $\sum_i\mathrm{VT}^{k}(\bm{x}^{(P,2i)}\bm{x}^{(P,2i+1)})\equiv g_k'\pmod{3(2P)^k}$, $k\in \{1,2,3\}$;
   \end{itemize}
    then $\C(n,11,\mathcal{B})$ is an $(n,11;\mathcal{B})$-reconstruction code and  by the pigeonhole principle, there exist appropriate sequences $\bm{s}, \bm{g}$ and $\bm{g}'$ such that $\mathrm{red}(\C(n,11,\mathcal{B}))\le \log n+12\log\log n+O(1)$.
  \end{theorem}
  \begin{IEEEproof} 
   Given any two distinct sequences $\bm{x},\bm{y}\in \C(n,11,\mathcal{B})$, we note that $\C(n,11,\mathcal{B})$ is a subcode of $\C(n,14,\mathcal{B})$ and $\C_{\mathcal{DS}}^P$, so from Theorem \ref{th:N=14}, we have $|\mathcal{B}(\boldsymbol{x}) \cap \mathcal{B}(\boldsymbol{y})|\leq 13$. Suppose that there exist a pair of sequences $\bm{x}$ and $\bm{y}$ in $\C(n,11,\mathcal{B})$ such that $|\mathcal{B}(\boldsymbol{x}) \cap \mathcal{B}(\boldsymbol{y})|>10$, then let \(\bm{x} = \bm{a}\tilde{\bm{x}}\bm{b}\) and \(\bm{y} = \bm{a}\tilde{\bm{y}}\bm{b}\). According to Remark \ref{rem:CONST},  $\tilde{\bm{x}}$ and $\tilde{\bm{y}}$ are composed of at most six $^{\le}2$-periodic sequences and four additional symbols from $\{0,1\}$. On the other hand, since $\bm{x},\bm{y}\in \mathcal{R}(n,2,t)$, it follows that $|\tilde{\bm{x}}|=|\tilde{\bm{y}}|\le P$,
    so from Lemma \ref{lem:N=11}, we know that \(\bm{x}, \bm{y} \notin \C_{\mathcal{DS}}^P\), which leads to a contradiction. Thus, for any two distinct sequences $\bm{x}, \bm{y}\in \C(n, 11, \mathcal{B})$, we have $|\mathcal{B}(\boldsymbol{x}) \cap \mathcal{B}(\boldsymbol{y})|\le 10$, i.e., \(\C(n, 11, \mathcal{B})\) is an \((n, 11; \mathcal{B})\)-reconstruction code. From Lemma \ref{lem:N=9_2}, we know that $|\mathcal{R}(n,2,t)|\ge 2^{n-1}$.  By the pigeonhole principle, there exist appropriate sequences $\bm{s}, \bm{g}$ and $\bm{g}'$ such that $|\C(n, 11, \mathcal{B})|\ge \frac{2^{n-1}}{4\cdot 2n\cdot {\textstyle \prod_{k=1}^{3}}(3(2P)^k)^2}$, which implies \(\mathrm{red}(\C(n, 11, \mathcal{B}))\le \log n + 12\log \log n + O(1)\). This completes the proof.
\end{IEEEproof}

 \subsection{Reconstruction Codes for N=9}

Let \(\bm{x} = \bm{a}\tilde{\bm{x}}\bm{b}\) and \(\bm{y} = \bm{a}\tilde{\bm{y}}\bm{b}\) be two distinct sequences in $\Sigma^n$. From Remark \ref{rem:CONST}, we know that when $B_5,B_6,B_{11}$, and $B_{12}$ are all empty, if \(|\mathcal{B}(\bm{x},\bm{y})| > 8\), then $\tilde{\bm{x}}$ and $\tilde{\bm{y}}$ are composed of at most six $^{\le}2$-periodic sequences and four additional symbols from $\{0,1\}$. This structure can be forbidden by $\mathcal{R}(n,2,t)\cap\C_{\mathcal{DS}}^P$ when we set $t=\lceil \log n   \rceil+3$ and $P\ge  6t+4$. Therefore, we shall construct a reconstruction code which is a subcode of $\mathcal{R}(n,2,t)\cap\C_{\mathcal{DS}}^P$, where $t=\lceil \log n   \rceil+3$ and $P\ge  6t+4$, such that any two different sequences in it satisfying $B_5,B_6,B_{11}$, and $B_{12}$ are empty.
To do this, we need to utilize the tools mentioned earlier, i.e. strong-locally-balanced sequence and differential sequence.

\begin{theorem}\label{th:N=9}
Let $n\ge 10, \epsilon =\frac{1}{18}$, $l=1296 \log n$, $P=4l$, $t=\lceil\log n\rceil+3$, and 
\begin{align*}
    \mathcal{M}\triangleq \{\bm{x}\in \Sigma^n: \bm{x}~ \text{and}~ \psi(\bm{x})~\text{are strong-}(l,\epsilon)\text{-locally-}\text{balanced sequences}\}.
\end{align*}
For fixed sequences $\bm{s}=s_0s_1, \bm{h}=h_0h_1, \bm{g}=g_1g_2g_3$ and $\bm{g}'=g_1'g_2'g_3'$, where $s_0\in [0,3]$, $s_1\in [0,3n]$, $h_0\in [0,6]$, $h_1\in [0,6(n+1)]$, $g_1,g_1'\in [0,6P-1]$, $g_2,g_2'\in [0,12P^2-1]$, and $g_3,g_3'\in [0,24P^3-1]$, we define $\C(n,9,\mathcal{B})\subseteq \Sigma^n$ such that any sequence $ \bm{x}\in \C(n,9,\mathcal{B})$ satisfies the following conditions:
\begin{itemize}
    \item $\bm{x}\in \mathcal{M}\cap \mathcal{R}(n,2,t)$;
    \item $\mathrm{VT}^{0}(\bm{x})\equiv s_0\pmod {4}$;
    \item $\mathrm{VT}^{1}(\bm{x})\equiv s_1\pmod {3n+1}$;
    \item $\mathrm{VT}^{0}(\psi(\bm{x}))\equiv h_0\pmod {7}$;
    \item $\mathrm{VT}^{1}(\psi(\bm{x}))\equiv h_1\pmod {6(n+1)+1}$;
    \item $\sum_i\mathrm{VT}^{k}(\bm{x}^{(P,2i-1)}\bm{x}^{(P,2i)})\equiv g_k\pmod{3(2P)^k}$, $k\in \{1,2,3\}$;
    \item $\sum_i\mathrm{VT}^{k}(\bm{x}^{(P,2i)}\bm{x}^{(P,2i+1)})\equiv g_k'\pmod{3(2P)^k}$, $k\in \{1,2,3\}$.
\end{itemize}
Then, $\C(n,9,\mathcal{B})$ is an $(n,9;\mathcal{B})$-reconstruction code, and by the pigeonhole principle, there exist appropriate sequences $\bm{s}, \bm{h}, \bm{g}$ and $\bm{g}'$ such that $\mathrm{red}(\C(n,9,\mathcal{B}))\le 2\log n+12\log \log n+O(1)$.
    \end{theorem}

Before proving Theorem \ref{th:N=9}, we need the following conclusions and observations.
\begin{corollary}\label{cor:N=9}
Suppose $n\ge 10$, and $b=6$, set $l=b^4\log n=1296\log n$, and $\epsilon=\frac{1}{3b}=\frac{1}{18}$. Then, when $t= \lceil \log n   \rceil+3 $, the size of $\mathcal{M}\cap \mathcal{R}(n, 2, t)$ is at least $2^{n-1}$.
\end{corollary}
\begin{IEEEproof}
 From Lemmas \ref{lem:N=9_1} and \ref{lem:N=9_2}, we can compute 
 \begin{align*}
  \left|\mathcal{M}\cap \mathcal{R}(n, 2, t)\right| &=  \left|\mathcal{M}\right|+\left|\mathcal{R}(n, 2, t)\right|-\left|\mathcal{M}\cup \mathcal{R}(n, 2, t)\right|\\
  &\ge 2^n\left(1-4n^{2-\frac{2}{9}b^2\log e}\right)+ 2^n\left(1-n(\frac{1}{2})^{t-2}\right)-2^n\\
  &=2^n\left(1-4n^{2-\frac{2}{9}b^2\log e}-n(\frac{1}{2})^{t-2}\right)\\
  &\ge 2^{n-1}.
 \end{align*}
\end{IEEEproof}

\begin{claim}\cite[Claims 12 and 16]{Sun-24-IT}\label{cla:N=9_1}
For any two distinct sequences $\bm{x},\bm{y}\in \C(n,9,\mathcal{B})$, if there exist $d_x,d_y,e_x,e_y\in [n]$ such that $\bm{x}(d_x,e_x)=\bm{y}(d_y,e_y)$, then
$\mathrm{VT}^{k}(\bm{x})=\mathrm{VT}^{k}(\bm{y})$ for $k\in\{0,1\}$, and $\mathrm{VT}^{k}(R(\bm{x}))=\mathrm{VT}^{k}(R(\bm{y}))$. 
\end{claim}

\begin{claim}\cite[Claims 14 and 17]{Sun-24-IT}\label{cla:N=9_2}
    For any two distinct sequences $\bm{x},\bm{y}\in \C(n,9,\mathcal{B})$, if there exist $d_x,d_y,e_x,e_y\in [n]$ such that $\bm{x}(d_x,e_x)=\bm{y}(d_y,e_y)$, then $\mathrm{VT}^{k}(\psi(\bm{x})=\mathrm{VT}^{k}(\psi(\bm{y}))$ for $k\in \{0,1\}$, and $\mathrm{VT}^{k}(\psi(R(\bm{x})))=\mathrm{VT}^{k}(\psi(R(\bm{y})))$.
\end{claim}

\begin{claim}\cite[Claim 18]{Sun-24-IT}\label{cla:N=9_3}
   For any two distinct sequences $\bm{x},\bm{y}\in \C(n,9,\mathcal{B})$,
    $R(\bm{x})$ and $R(\bm{y})$ also satisfy the constraints in $\C(n,9,\mathcal{B})$.
\end{claim}

For any \(\bm{x} \in \Sigma^n\) and $e_x\in [n]$, from Observation \ref{obs:sub}, we have
\begin{equation}\label{equ:e_x}
    \psi(x(0,e_x))_i=\left\{\begin{aligned}
        &\psi(x)_i,\quad &&i\notin \{e_x,e_x+1\},\\
        &1-\psi(x)_i,\quad &&i\in \{e_x,e_x+1\}.
    \end{aligned}\right.
\end{equation}
Furthermore, when $\bm{x}(d_x,e_x)=\bm{y}(d_y,e_y)$ and $d_x > d_y$, we can deduce from Observation \ref{obs:del} that 
\begin{align*}
    \psi(x(0,e_x))_{d_x} &\oplus \psi(x(0,e_x))_{d_x+1}=\psi(y(0,e_y))_{d_x+1},
\end{align*}
\begin{align*}
    \psi(y(0,e_y))_{d_y}&\oplus \psi(y(0,e_y))_{d_y+1}=\psi(x(0,e_x))_{d_y}.
\end{align*}
More detailedly,
\begin{flalign*}
    &\psi(x(0,e_x))_i=\left\{
    \begin{array}{ll}
        \psi(y(0,e_y))_i,& i\in [1,d_y-1]\cup[d_x+2,n+1],\\
        \psi(y(0,e_y))_{i+1},&  i\in [d_y+1,d_x-1],
    \end{array}\right.
\end{flalign*}
as shown in Fig \ref{fig:1}.

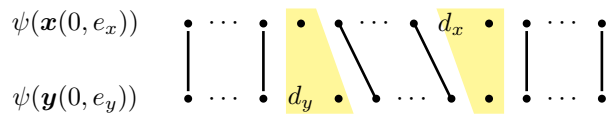
\begin{figure*}[ht]
    \centering
    \begin{tikzpicture}[line width=1pt,scale=1]
        \fill[fill=yellow!50]
        (-1.7,-0.2)--(-0.8,-0.2)--(-1.3,1.2)--(-1.7,1.2)--cycle	(0.3,1.2)--(1.2,1.2)--(1.2,-0.2)--(0.8,-0.2)--cycle;
        \draw 
        (-4.5,0) node []{$\psi(\bm{y}(0,e_y))$}
        (-4.5,1) node []{$\psi(\bm{x}(0,e_x))$}
        (-1-2,0.9)--(-1-2,0.1)
        (0-2,0.9)--(0-2,0.1)
        (-3.5+2,0)   node []{$d_y$}
        (-2.5+2,1)   node []{$\cdots$}
        (2,0)    node []{$\cdots$}
        (0,0)    node []{$\cdots$}
        (2,1)    node []{$\cdots$}
        (-0.5-2,0)   node []{$\cdots$}
        (-0.5-2,1)   node []{$\cdots$}
        (0.5,1) node[]{$d_x$};
        \filldraw [black]
        (-1-2,0)   circle [radius=1pt]
        (-1-2,1)   circle [radius=1pt]
        (0-2,0)   circle [radius=1pt]
        (0-2,1)   circle [radius=1pt]
        (-3.5+2,1)  circle [radius=1pt]
        (-3+2,0)   circle [radius=1pt]
        (-3+2,1)   circle [radius=1pt]
        (-2+2,1)   circle [radius=1pt]
        (-0.5,0)   circle [radius=1pt]
        (0.5,0)   circle [radius=1pt]
        (1,1)   circle [radius=1pt]
        (1,0)   circle [radius=1pt]
        (1.5,0)   circle [radius=1pt]
        (1.5,1)   circle [radius=1pt]
        (2.5,0)   circle [radius=1pt]
        (2.5,1)   circle [radius=1pt];
        \draw 
        (0.01,0.99)--(0.49,0.01)
        (-0.51,0.01)--(-0.99,0.99)
        (2.5,0.1)--(2.5,0.9)
        (1.5,0.1)--(1.5,0.9);
    \end{tikzpicture}
    \caption{Illustrations of $\psi(\bm{x}(0,e_x))$ and $\psi(\bm{y}(0,e_y))$ when $\psi(\bm{x}(0,e_x))_{[n+1] \setminus d_x}=\psi(\bm{y}(0,e_y))_{[n+1]\setminus d_y}$, where $d_x>d_y$. Each pair of bits connected by a solid line are of equal value, while those connected by a yellow undertone signify a connection resulted from deletions.}\label{fig:1}
\end{figure*}

Similarly, when \(\bm{x}(d_x,e_x) = \bm{y}(d_y,e_y)\) and \(d_x < d_y\), we have
\begin{align*}
    \psi(x(0,e_x))_{d_x}&\oplus \psi(x(0,e_x))_{d_x+1}=\psi(y(0,e_y))_{d_x}, 
\end{align*}
\begin{align*}
    \psi(y(0,e_y))_{d_y}&\oplus \psi(y(0,e_y))_{d_y+1}=\psi(x(0,e_x))_{d_y+1},
\end{align*}
and 
\begin{equation*}
    \psi(x(0,e_x))_i=\left\{\begin{aligned}
        &\psi(y(0,e_y))_i,\quad &&i\in [1,d_x-1]\cup [d_y+2,n+1],\\
        &\psi(y(0,e_y))_{i-1},\quad &&i\in [d_x+2,d_y].
    \end{aligned}\right.
\end{equation*}
Then, we can calculate
\begin{align*}
    \Delta_{\psi}^{(1)}\triangleq &\mathrm{VT}^{1}(\psi(\bm{x}))-\mathrm{VT}^{1}(\psi(\bm{y}))\\=&\mathrm{VT}^{1}(\psi(\bm{x}))-\mathrm{VT}^{1}(\psi(\bm{x}(0,e_x)))\notag\\
    &+\mathrm{VT}^{1}(\psi(\bm{x}(0,e_x)))-\mathrm{VT}^{1}(\psi(\bm{y}(0,e_y)))\\&+\mathrm{VT}^{1}(\psi(\bm{y}(0,e_y)))-\mathrm{VT}^{1}(\psi(\bm{y})).
\end{align*}

Let 
\begin{align*}
  \eta_e^x\triangleq \mathrm{VT}^{1}(\psi(\bm{x}))-\mathrm{VT}^{1}(\psi(\bm{x}(0,e_x))),  
\end{align*}
\begin{align*}
   \eta_e^y\triangleq \mathrm{VT}^{1}(\psi(\bm{y}(0,e_y)))-\mathrm{VT}^{1}(\psi(\bm{y})),
\end{align*}
and
\begin{align*}
  \xi_d\triangleq \mathrm{VT}^{1}(\psi(\bm{x}(0,e_x)))-\mathrm{VT}^{1}(\psi(\bm{y}(0,e_y))).  
\end{align*}
 Then $\Delta_\psi^{(1)}=\eta_e^x+\eta_e^y+\xi_d$. According to the previous analysis, we know that
\begin{align*}
    &\eta_e^x=e_x(2\psi(x)_{e_x}-1)+(e_x+1)(2\psi(x)_{e_x+1}-1),
\end{align*}
\begin{align*}
    &\eta_e^y=e_y(1-2\psi(y)_{e_y})+(e_y+1)(1-2\psi(y)_{e_y+1}).
\end{align*}

\begin{itemize}
    \item	When $d_x>d_y$
    \begin{align*}
        \xi_d=&d_y\psi(x(0,e_x))_{d_y}-d_y\psi(y(0,e_y))_{d_y}\\
        &-(d_y+1)\psi(y(0,e_y))_{d_y+1}+d_x\psi(x(0,e_x))_{d_x}\notag\\
        &+(d_x+1)\psi(x(0,e_x))_{d_x+1}-(d_x+1)\psi(y(0,e_y))_{d_x+1}\\
        &-\sum_{i=d_y+1}^{d_x-1}\psi(x(0,e_x))_i.
    \end{align*}
    \item	When $d_x<d_y$
    \begin{align*}\label{q1}
        \xi_d=&(d_y+1)\psi(x(0,e_y))_{d_y+1}-d_y\psi(y(0,e_y))_{d_y}\\
        &-(d_y+1)\psi(y(0,e_y))_{d_y+1}+d_x\psi(x(0,e_x))_{d_x}\notag\\
        &+(d_x+1)\psi(x(0,e_x))_{d_x+1}-d_x\psi(y(0,e_y))_{d_x}\\
        &+\sum_{i=d_x+2}^{d_y}\psi(x(0,e_x))_i.
    \end{align*} 
\end{itemize}
When $d_x>d_y$, let
\begin{align*}
 \zeta_d^x&\triangleq d_x\psi(x(0,e_x))_{d_x}+(d_x+1)\psi(x(0,e_x))_{d_x+1}-(d_x+1)\psi(y(0,e_y))_{d_x+1},
\end{align*}
\begin{align*}
 \zeta_d^y\triangleq d_y\psi(x(0,e_x))_{d_y}-d_y\psi(y(0,e_y))_{d_y}-(d_y+1)\psi(y(0,e_y))_{d_y+1}.
\end{align*}
When $d_x<d_y$, let
\begin{align*}
 &\xi_d^x\triangleq d_x\psi(x(0,e_x))_{d_x}\notag+(d_x+1)\psi(x(0,e_x))_{d_x+1}-d_x\psi(y(0,e_y))_{d_x},
\end{align*}
\begin{align*}
\xi_d^y&\triangleq(d_y+1)\psi(x(0,e_y))_{d_y+1}-d_y\psi(y(0,e_y))_{d_y}-(d_y+1)\psi(y(0,e_y))_{d_y+1}.
\end{align*}
Then, we will obtain that
\begin{equation}\label{equ:xy}
\Delta_\psi^{(1)}=\eta_e^x+\eta_e^y+\zeta_d^x+ \zeta_d^y-\sum_{i=d_y+1}^{d_x-1}\psi(x(0,e_x))_i,
\end{equation}
when $d_x>d_y$, and
\begin{equation}\label{equ:yx}
\Delta_\psi^{(1)}=\eta_e^x+\eta_e^y+\xi_d^x+ \xi_d^y+\sum_{i=d_x+2}^{d_y}\psi(x(0,e_x))_i,
\end{equation}
when $d_x<d_y$.

\begin{table*}
\caption{In this table, we calculate the values of $\zeta_d^x,\zeta_d^y$ for $d_x>d_y$, and calculate the values of $\xi_d^x,\xi_d^y$ for $d_x<d_y$. When we calculate $\zeta_d^x$ and $\xi_d^x$, $N_d$ represents the possible values of $N(\bm{x})_d$, and the errors represent the transformation $\psi(x(0,e_x))_{d_x}  \psi(x(0,e_x))_{d_x+1}\longrightarrow\psi(y(0,e_y))_{d_x+1}$ and $\psi(x(0,e_x))_{d_x} \psi(x(0,e_x))_{d_x+1}\longrightarrow\psi(y(0,e_y))_{d_x}$, respectively. On the other hand, when we calculate $\zeta_d^y$ and $\xi_d^y$, $N_d$ represents the possible value of $N(\bm{y})_d$, and the errors represent the transformation $\psi(y(0,e_y))_{d_y} \psi(y(0,e_y))_{d_y+1}\longrightarrow\psi(x(0,e_x))_{d_y}$ and $\psi(y(0,e_y))_{d_y} \psi(y(0,e_y))_{d_y+1}\longrightarrow\psi(x(0,e_x))_{d_y+1}$, respectively.} 
\centering
\begin{tabular}{c|c|c|c|c|c}  
\hline
\multicolumn{2}{c|}{} & \multicolumn{2}{c|}{$d_x>d_y$} & \multicolumn{2}{c}{$d_x<d_y$}\\
\hline 
$N_d$ & Error & $\zeta_d^x$ & $\zeta_d^y$ & $\xi_d^x$ & $\xi_d^y$\\
\hline
\multirow{3}{*}{$0$} & $00\to 0$ & $0$ & $0$ & $0$ & $0$\\
\cline{2-6}
& $01\to 1$ & $0$ & $-1$ & $1$ & $0$\\
\cline{2-6}
& $10\to 1$ & $-1$ & $0$ & $0$ & $1$\\
\hline 
$2$ & $11\to 0$ & $2d_x+1$ & $-2d_y-1$ & $2d_x+1$ & $-2d_y-1$ \\
\hline
\end{tabular}
\label{tab:del}
\end{table*}

\begin{table*}
\caption{In this table, we calculate the sizes of $\eta_e^x$ and $\eta_e^y$. When we calculate $\eta_e^x$, $N_e$ represents the possible values of $N(\bm{x})_e$, and the errors represent the transformation $\psi(x)_{e_x}\psi(x)_{e_x+1}\longrightarrow \overline{\psi(x)_{e_x}\psi(x)_{e_x+1}}$. On the other hand, when we calculate $\eta_e^y$, $N_e$ represents the possible values of $N(\bm{y})_e$, and the errors represent the transformation $\psi(y)_{e_y}\psi(y)_{e_y+1}\longrightarrow \overline{\psi(y)_{e_y}\psi(y)_{e_y+1}}$.} 
\centering
\begin{tabular}{c|c|c|c}  
\hline 
$N_e$ & Error & $\eta_e^x$ & $\eta_e^y$\\
\hline
$2$ & $11\to 00$ & $2e_x+1$ & $-2e_y-1$\\
\hline
\multirow{2}{*}{0} & $01\to 10$ & $1$ & $-1$\\
\cline{2-4}
& $10\to 01$ & $-1$ & $1$\\
\hline 
$-2$ & $00\to 11$ & $-2e_x-1$ & $2e_y+1$ \\
\hline
\end{tabular}
\label{tab:sub}
\end{table*}

Now, we will prove Theorem \ref{th:N=9} by using above observations.

\begin{IEEEproof}[Proof of Theorem \ref{th:N=9}]
Given  any two distinct sequences $\bm{x}$ and $\bm{y}$ in \(\C(n, 9, \mathcal{B})\), let \(\bm{x} = \bm{a}\tilde{\bm{x}}\bm{b}\) and \(\bm{y} = \bm{a}\tilde{\bm{y}}\bm{b}\). From Theorem \ref{th:N=11} we can conclude that \( \C(n,9,\mathcal{B})\) forms an \((n,11,\mathcal{B})\)-reconstruction code, thus we have that $|\mathcal{B}(\bm{x},\bm{y})|\le 10$. Furthermore, if $B_5,B_6,B_{11}$ and $B_{12}$ are empty, then we have that $|\mathcal{B}(\bm{x},\bm{y})|\le 8$, otherwise, we can obtain that $\tilde{\bm{x}}$ and $\tilde{\bm{y}}$ are composed of at most six $^{\le}2$-periodic sequences and four additional symbols from $\{0,1\}$ by Remark \ref{rem:CONST}. On the other hand, we note that \( \C(n,9,\mathcal{B})\) is a subcode of $\C_{\mathcal{DS}}^{\mathcal{P}}$. However, since $\bm{x},\bm{y}\in \mathcal{R}(n,2,t)$, we have $|\tilde{\bm{x}}|=|\tilde{\bm{y}}|\le 6(\lceil \log n   \rceil+3)+4< P=4l$, so from Lemma \ref{lem:N=11}, we know that $\bm{x},\bm{y}\notin \C_{\mathcal{DS}}^{\mathcal{P}}$, which is in contradiction with $\bm{x},\bm{y}\in \C(n, 9, \mathcal{B})$. Therefore, we only need to demonstrate that $B_5,B_6,B_{11}$ and $B_{12}$  are empty for any two distinct sequences $\bm{x}$ and $\bm{y}$ in \(\C(n, 9, \mathcal{B})\). 

Let $N(\bm{x}) \triangleq (N(\bm{x})_d, N(\bm{x})_e)$ denote that one deletion occurs in $\bm{x}$ resulting in a decrease in the number of ones in $\psi(\bm{x})$ by $N(\bm{x})_d$, while one substitution occurs in $\bm{x}$ resulting in a decrease in the number of ones in $\psi(\bm{x})$ by $N(\bm{x})_e$. From Observations \ref{obs:del} and \ref{obs:sub}, we know that $N(\bm{x})_d\in \{0,2\}$ and $N(\bm{x})_e\in \{-2,0,2\}$. Then the value of $N_{\bm{x}\bm{y}} \triangleq N(\bm{x})_d + N(\bm{x})_e = N(\bm{y})_d+ N(\bm{y})_e\in \{-2,0,2,4\}$ can be determined by $\mathrm{VT}^{0}(\psi(\bm{x})) \pmod{7}$ and $\mathrm{VT}^{0}(\psi(\bm{x}(d_x,e_x)))$. Now we will prove our conclusion based on Equation (\ref{equ:xy}) and Equation (\ref{equ:yx}), and the value of $N_{\bm{x}\bm{y}}$. We first calculate $\zeta_d^x,\zeta_d^y,\xi_d^x$ and $\xi_d^y$ for all possible values of $N(\bm{x})_d$ and $N(\bm{y})_d$ (summarized in Table \ref{tab:del}), and calculate $\eta_e^x$ and $\eta_e^y$ for all possible values of $N(\bm{x})_e$ and $N(\bm{y})_e$ (summarized in Table \ref{tab:sub}). Furthermore,  
let $\eta = \min\{e_x, e_y, d_x, d_y\}$ and $\xi = \max\{e_x, e_y, d_x, d_y\}$, it is reasonable to assume $|\eta - \xi| > P$. 

\begin{enumerate}[1)]
    \item If $B_5$ is a nonempty set, then there exist $\bm{x}(d_x,e_x) = \bm{y}(d_y,e_y)\in B_5$, where $d_x,d_y,e_x$ and $e_y$ satisfy $e_y=0$ and $d_x< e_x<d_y$ or $e_x=0$ and $d_x\le e_y<d_y$ according to the definition of $B_5$. 
    
    If $e_y=0$ and $d_x< e_x<d_y$, i.e.,  
    $\bm{y}$ only suffers one deletion, then we have $\eta_e^y = 0$. In addition, by comparing the indices of  $\bm{x}(d_x,e_x)$ and $\bm{y}(d_y,e_y)$, it can be easily checked that $x_i=y_i$ for $i\in [1,d_x-1]\cup [d_y+1,n]$ and $x_i=y_{i-1}$ for $i\in [d_x+1,d_y]\setminus{e_x}$, and $x_{e_x}\neq y_{e_x-1}$. Moreover, we have $N(\bm{x})_d,N(\bm{y})_d \in \{0, 2\}$, $N(\bm{x})_e \in \{-2, 0, 2\}$, and $N(\bm{y})_e = 0$, thus $N_{\bm{x}\bm{y}}= N(\bm{x})_d + N(\bm{x})_e = N(\bm{y})_d \in \{0, 2\}$. Since $\psi(\bm{x})$ is a strong-$(l, \epsilon)$-locally balanced sequence and $|\eta - \xi|=|d_y-d_x| > P=4l$, we have that $\sum_{i=d_x+2}^{d_y}\psi(x)_i\le  \left(\frac{1}{2} + \epsilon\right)(d_y - d_x )+1$. In addition, from Equation (\ref{equ:e_x}), we know that $\sum_{i=d_x+2}^{d_y}\psi(x(0,e_x))_i\le \sum_{i=d_x+2}^{d_y}\psi(x)_i+2\le \left(\frac{1}{2} + \epsilon\right)(d_y - d_x )+3$. 
    
    By using Equation (\ref{equ:yx}), we can obtain the following results.
\begin{itemize}
    \item When $N_{\bm{x}\bm{y}}=0$, we have  $N(\bm{x})=N(\bm{y})=(0,0)$ or $N(\bm{y})=(0,0)$, $N(\bm{x})=(2,-2)$.  
        \begin{itemize}      
         \item	If $N(\bm{x})=N(\bm{y})=(0,0)$, then $\eta_e^x\in \{1,-1\}$ and $\xi_d^x,\xi_d^y\in \{0,1\}$ according to Table \ref{tab:del} and Table \ref{tab:sub}, and
    \begin{align*}  \Delta_\psi^{(1)}=&\xi^x_d+\eta_e^x+\xi^y_d+\sum_{i=d_x+2}^{d_y}\psi(x(0,e_x))_i> 0.\notag  
    \end{align*}
    \item If $N(\bm{x})=(2,-2)$, $N(\bm{y})=(0,0)$, then $\xi_d^y\in \{0,1\}$ according to Table \ref{tab:del} and Table \ref{tab:sub}, and
    \begin{align}\label{equ:(2,2)} 
     \Delta_\psi^{(1)}=&2d_x+1+\xi_d^y-2e_x-1+\sum_{i=d_x+2}^{d_y}\psi(x(0,e_x))_i\notag\\
        \le& -2(e_x-d_x)+(\frac{1}{2}+\epsilon)(d_y-d_x)+4.   
    \end{align} 
           
    From Claim \ref{cla:N=9_1}, we have $\mathrm{VT}^{0}(\bm{x}) = \mathrm{VT}^{0}(\bm{y})$ and $\mathrm{VT}^{1}(\bm{x}) = \mathrm{VT}^{1}(\bm{y})$. In addition, since $x_{e_x}\neq y_{e_x-1}$, the possible values for $x_{d_x}, x_{e_x}, y_{d_y}, y_{e_x-1}$ are either $(x_{d_x}, x_{e_x}) = (y_{e_x-1}, y_{d_y}) = (1, 0)$ or $(x_{d_x}, x_{e_x}) = (y_{e_x-1}, y_{d_y})= (0, 1)$. 
    \begin{itemize}
        \item  When $(x_{d_x}, x_{e_x}) = (y_{e_x-1}, y_{d_y}) = (1, 0)$, we have
        \begin{align*}
        \mathrm{VT}^{1}(\bm{x}) - \mathrm{VT}^{1}(\bm{y})=d_x-e_x+1+\sum_{i=d_x+1}^{e_x-1} x_i + \sum_{i=e_x+1}^{d_y} x_i=0,
    \end{align*}
    which means that 
    \begin{equation}\label{equ:(1,0)}
      e_x - d_x = \sum_{i=d_x+1}^{e_x-1} x_i + \sum_{i=e_x+1}^{d_y} x_i+1. 
    \end{equation} 
    \item When $(x_{d_x}, x_{e_x}) = (y_{e_x-1}, y_{d_y})= (0, 1)$, we have 
    \begin{align*}
        \mathrm{VT}^{1}(\bm{x}) - \mathrm{VT}^{1}(\bm{y})=e_x-d_y+\sum_{i=d_x+1}^{e_x-1} x_i + \sum_{i=e_x+1}^{d_y} x_i=0,
    \end{align*}
    which means that
    \begin{equation*}
       \sum_{i=d_x+1}^{e_x-1} x_i + \sum_{i=e_x+1}^{d_y} x_i = d_y - e_x. 
    \end{equation*}
     In this case, we have 
    \begin{equation}\label{equ:(0,1)}
        e_x-d_x=(d_y-d_x)-(d_y-e_x)=(d_y-d_x)-\big(\sum_{i=d_x+1}^{e_x-1} x_i + \sum_{i=e_x+1}^{d_y} x_i\big).
    \end{equation}
    \end{itemize}  
        \end{itemize}
   On the other hand, since $\bm{x}$ is a strong-$(l, \epsilon)$-locally balanced sequence, we can obtain $\sum_{i=d_x+1}^{e_x-1} x_i + \sum_{i=e_x+1}^{d_y} x_i \in \left[\left(\frac{1}{2} - \epsilon\right)(d_y - d_x )- 1, \left(\frac{1}{2} + \epsilon\right)(d_y - d_x) + 1\right]$. Thus, for any case of $(x_{d_x}, x_{e_x}) = (y_{e_y}, y_{d_y})\in \{(1,0),(0,1)\}$, by Equations (\ref{equ:(1,0)}) and (\ref{equ:(0,1)}), we can derive  
    \begin{align*}
        -2(e_x-d_x)\le -2\left(\frac{1}{2} - \epsilon\right)(d_y - d_x)+2.  
    \end{align*}
    Combining with Equation (\ref{equ:(2,2)}), we can obtain that
    \begin{align*}
     \Delta_\psi^{(1)}&\le -2\left(\frac{1}{2} - \epsilon\right)(d_y - d_x )+2+ (\frac{1}{2}+\epsilon)(d_y-d_x)+4\\
     &=(-\frac{1}{2}+3\epsilon)(d_y-d_x)+6<0.
    \end{align*}
    \item When $N_{\bm{x}\bm{y}}=2$, we have $N(\bm{x})=N(\bm{y})=(2,0)$ or $N(\bm{x})=(0,2), N(\bm{y})=(2,0)$.
\begin{itemize}
    \item If $N(\bm{x})=N(\bm{y})=(2,0)$, then $\eta_e^x\in \{-1,1\}$ according to Table \ref{tab:del} and Table \ref{tab:sub}, and
    \begin{align*}
       \Delta_\psi^{(1)}=&2d_x+1-2d_y-1+\eta_e^x+\sum_{i=d_x+2}^{d_y}\psi(x(0,e_x))_i\notag\\
            \le& -2(d_y-d_x)+1+(\frac{1}{2}+\epsilon)(d_y-d_x)+3<0.\notag\\
    \end{align*}
    \item If $N(\bm{x})=(0,2), N(\bm{y})=(2,0)$, then $\xi_d^x\in \{0,1\}$ according to Table \ref{tab:del} and Table \ref{tab:sub}, and 
    \begin{align*}
       \Delta_\psi^{(1)}=&2e_x+1-2d_y-1+\xi_d^x+\sum_{i=d_x+2}^{d_y}\psi(x(0,e_x))_i\notag\\
            \le& -2(d_y-e_x)+(\frac{1}{2}+\epsilon)(d_y-d_x)+4.\notag\\
    \end{align*}
   Similar to the case $N(\bm{x})=(2,-2)$, $N(\bm{y})=(0,0)$, we can also prove that $\Delta_\psi^{(1)}<0$ in this case.
\end{itemize}
\end{itemize}
 In conclusion, we have $\Delta_\psi^{(1)} \neq 0$ if $e_y=0$. By symmetry, considering $R(\bm{y})$ and $R(\bm{x})$ instead of $\bm{x}$ and $\bm{y}$, from Claim \ref{cla:N=9_3}, we can also obtain $\Delta_\psi^{(1)} \neq 0$ if $e_x=0$. Thus, we have $\Delta_\psi^{(1)} \neq 0$ when $B_{5}$ is a nonempty set, which contradicts to $\bm{x},\bm{y}\in \C(n,9,\mathcal{B})$. Consequently, $B_5$ is an empty set.     
\item From Claims \ref{cla:N=9_1}, \ref{cla:N=9_2} and \ref{cla:N=9_3}, similarly, we can also prove $B_6$ is an empty set by considering $R(\bm{x})$ and $R(\bm{y})$ and using the conclusion of the previous case.
    \item  If $B_{11}$ is a nonempty set, then there exist $d_x,d_y,e_x$ and $e_y$ such that $\bm{x}(d_x,e_x) = \bm{y}(d_y,e_y)\in B_{11}$, where $e_x<d_x<d_y<e_y$ or $e_y<d_x<d_y<e_x$ according to the definition of $B_{11}$. If $e_x<d_x<d_y<e_y$, then we have
    \[
    N_{\bm{x}\bm{y}}\triangleq N(\bm{x})_d+N(\bm{x})_e=N(\bm{y})_d+N(\bm{y})_e\in \{4,-2,0,2\}.
    \]
    \begin{itemize}
						
    \item When $N_{\bm{x}\bm{y}}=4$, it  holds that $N(\bm{x})=N(\bm{y})=(2,2)$,  from Table \ref{tab:del} and Table \ref{tab:sub}, we have that
    \begin{align*}
    \Delta_\psi^{(1)}&=2e_x+1+2d_x+1-2d_y-1-2e_y-1+\sum_{i=d_x+2}^{dy}\psi(x(0,e_x))_i\\
    &=-2(d_y-d_x)-2(e_y-e_x)+ \sum_{i=d_x+2}^{dy}\psi(x(0,e_x))_i<0.
    \end{align*}
    \item When $N_{\bm{x}\bm{y}}=-2$, we have $N(\bm{x})=(0,-2)=N(\bm{y})$, from Table \ref{tab:del} and Table \ref{tab:sub}, we know that $\xi_d^x,\xi_d^y\in \{0,1\}$ and
    \[
\Delta_\psi^{(1)}=\xi_d^x+\xi_d^y-2e_x-1+2e_y+1+\sum_{i=d_x+2}^{dy}\psi(x(0,e_x))_i>0.
    \]
    \item When $N_{\bm{x}\bm{y}}=2$, we have $N(\bm{x}), N(\bm{y})\in \{(0,2), (2,0)\}$.
    \begin{itemize}
        
        \item If $N(\bm{x})=N(\bm{y})=(0,2)$, then $\xi_d^x,\xi_d^y\in \{0,1\}$ according to Table \ref{tab:del} and Table \ref{tab:sub}, and
        \[
        \Delta_\psi^{(1)}=2e_x+1-2e_y-1+\xi_d^x+\xi_d^y+\sum_{i=d_x+2}^{dy}\psi(x(0,e_x))_i<0.
        \]

        \item If $N(\bm{x})=N(\bm{y})=(2,0)$, then $\eta_e^y,\eta_e^x\in \{-1,1\}$ according to Table \ref{tab:del} and Table \ref{tab:sub}, and 
        \[
        \Delta_\psi^{(1)}=2d_x+1-2d_y-1+\eta_e^x+\eta_e^y+\sum_{i=d_x+2}^{dy}\psi(x(0,e_x))_i<0.
        \]
        
        \item If $N(\bm{x})=(0,2),N(\bm{y})=(2,0)$, then $\eta_e^y\in \{-1,1\}$ and $\xi_d^x\in \{0,1\}$ according to Table \ref{tab:del} and Table \ref{tab:sub}, and
        \[
        \Delta_\psi^{(1)}=2e_x+1-2d_y-1+\xi_d^x+\eta_e^y+\sum_{i=d_x+2}^{dy}\psi(x(0,e_x))_i<0.
        \]
        
        \item If $N(\bm{x})=(2,0),N(\bm{y})=(0,2)$, then $\eta_e^x\in \{-1,1\}$ and $\xi_d^y\in \{0,1\}$ according to Table \ref{tab:del} and Table \ref{tab:sub}, and
        \[
        \Delta_\psi^{(1)}=2d_x+1-2e_y-1+\eta_e^x+\xi_d^y+\sum_{i=d_x+2}^{dy}\psi(x(0,e_x))_i<0.
        \]
    \end{itemize}
    \item When $N_{\bm{x}\bm{y}}=0$, it  holds that $N(\bm{x}), N(\bm{y})\in \{(0,0),(2,-2)\}$.
    \begin{itemize}
        \item	If $N(\bm{x})=N(\bm{y})=(0,0)$, then $\eta_e^y,\eta_e^x\in \{-1,1\}$ and $\xi_d^y,\xi_d^x\in \{0,1\}$ according to Table \ref{tab:del} and Table \ref{tab:sub}, and
        \[ \Delta_\psi^{(1)}=\xi_d^x+\xi_d^y+\eta_e^x+\eta_e^y+\sum_{i=d_x+2}^{dy}\psi(x(0,e_x))_i>0.
        \]
							
    \item	If $N(\bm{x})=N(\bm{y})=(2,-2)$, then from Table \ref{tab:del} and Table \ref{tab:sub}, we have
    \begin{align*}
        \Delta_\psi^{(1)}&=2d_x+1-2d_y-1-2e_x-1+2e_y+1+\sum_{i=d_x+2}^{dy}\psi(x)_i\notag\\
        &=2(e_y-d_y)+2(d_x-e_x)+\sum_{i=d_x+2}^{dy}\psi(x(0,e_x))_i>0.
    \end{align*}
    
    \item	If $N(\bm{x})=(0,0), N(\bm{y})=(2,-2)$, then $\eta_e^x\in \{-1,1\}$ and $\xi_d^x\in \{0,1\}$ according to Table \ref{tab:del} and Table \ref{tab:sub}, and
    \begin{align*}
 \Delta_\psi^{(1)}&=\xi^x_d+\eta_e^x-2d_y-1+2e_y+1+\sum_{i=d_x+2}^{dy}\psi(x(0,e_x))_i>0.
    \end{align*}
    
    \item	If $N(\bm{x})=(2,-2), N(\bm{y})=(0,0)$, then $\eta_e^y\in \{-1,1\}$ and $\xi_d^y\in \{0,1\}$ according to Table \ref{tab:del} and Table \ref{tab:sub}, and
    \begin{align*}
\Delta_\psi^{(1)}&=\xi^y_d+\eta_e^y+2d_x+1-2e_x-1+\sum_{i=d_x+2}^{dy}\psi(x(0,e_x))_i>0.
    \end{align*}
\end{itemize}
    \end{itemize}
   In conclusion, we have $\Delta_\psi^{(1)} \neq 0$ if $e_x<d_x<d_y<e_y$. By symmetry, considering $R(\bm{y})$ and $R(\bm{x})$ instead of $\bm{x}$ and $\bm{y}$, from Claim \ref{cla:N=9_3}, we can also obtain $\Delta_\psi^{(1)} \neq 0$ if $e_y<d_x<d_y<e_x$. Thus, we have $\Delta_\psi^{(1)} \neq 0$ when $B_{11}$ is a nonempty set, which contradicts $\bm{x},\bm{y}\in \C(n,9,\mathcal{B})$. Consequently, $B_{11}$ is an empty set.
    \item By considering $R(\bm{x})$ and $R(\bm{y})$ and using Claims \ref{cla:N=9_1}, \ref{cla:N=9_2} and \ref{cla:N=9_3}, similarly, we can also prove $B_{12}$ is an empty set.
\end{enumerate}	
From the analysis above we know that $B_5,B_6,B_{11}$ and $B_{12}$ are empty for any two distinct sequences $\bm{x},\bm{y}\in \C(n, 9, B_{1, 1})$, thus $|\mathcal{B}(\bm{x},\bm{y})|\le 8$, i.e., $\C(n, 9, B_{1, 1})$ forms an $(n, 9; \mathcal{B})$-reconstruction code. Furthermore, by Corollary \ref{cor:N=9}, we have that $|\mathcal{M}\cap \mathcal{R}(n,2,t)|\ge 2^{n-1}$. Applying the pigeonhole principle, there exist appropriate sequences $\bm{s}, \bm{h}, \bm{g}$ and $\bm{g}'$ such that $|\C(n, 9, B_{1, 1})|\ge \frac{2^{n-1}}{4\cdot(3n+1)\cdot 7\cdot(6(n+1)+1)\cdot {\textstyle \prod_{k=1}^{3}}(3(2P)^k)^2}$, which implies that $\mathrm{red}(\C(n, 9, B_{1, 1}))\le 2\log n + 12\log \log n + O(1)$. This completes the proof.  
\end{IEEEproof}

\subsection{Reconstruction Codes for N=5}
			
As for the $(n, 5; \mathcal{B})$-reconstruction code, we directly utilize the single-deletion single-substitution list-decodable code with a list size of two mentioned in Lemma \ref{lem:DS_L}.

\begin{theorem}\label{th:N=5}
For any $s_0 \in [0,3]$, $s_1 \in [0,2n-1]$, and $s_2 \in [0,2n^2-1]$, let 
\begin{align*}
  \C(n,5,\mathcal{B})=\C_L\triangleq \big\{\bm{x}\in \Sigma^n :\mathrm{VT}^{0}(\bm{x})\equiv s_0 \pmod{4},~ \mathrm{VT}^k(\bm{x})\equiv s_{k} \pmod{2n^k} \text{ for } k\in \{1,2\} \big\}. 
\end{align*}
Then $\C(n,5,\mathcal{B})$ is an $(n,5;\mathcal{B})$-reconstruction code, and by the pigeonhole principle, there exists some choice of $s_0, s_1$ and $s_2$ such that $\mathrm{red}(\C(n,5,\mathcal{B}))\le 3\log n+4$.
\end{theorem}
\begin{IEEEproof}
     Given any two different sequences $\bm{x}, \bm{y} \in \C(n, 5, B_{1, 1})$, and positive integers $d_x,d_y,e_x,e_y\in [n]$. If $\bm{z} = \bm{x}(d_x,e_x) = \bm{y}(d_y,e_y)$, then from Lemma \ref{lem:DS_L}, we can obtain that $\bm{z} \in B_{17} \cup B_{18}$. Furthermore, note that $\C(n, 5, B_{1, 1})$ is a subcode of $\C_1$, so from Lemma \ref{lem:B16_empty}, it follows that $|B_{1, 1}(\bm{x}) \cap B_{1, 1}(\bm{y})| \leq 4$. Therefore, $\C(n, 5, B_{1, 1})$ forms an $(n, 5; \mathcal{B})$-reconstruction code. By the pigeonhole principle, there exist $s_0,s_1$ and $s_2$ such that $|\C(n, 5, B_{1, 1})|\ge \frac{2^n}{4\cdot 2n\cdot 2n^2}$, which implies that $\mathrm{red}(\C(n, 5, B_{1, 1}))\le 3\log n+4$. This completes the proof.
\end{IEEEproof}

So far, we have provided several reconstruction codes for the case where the intersection between single-deletion and single-substitution balls of sequences $\bm{x},\bm{y} \in \Sigma^n$ is constant.

\section{missing proofs}\label{sec:proof}

In this section, we will provide the proofs of Theorems \ref{th:B5_nonempty}, \ref{th:B11_nonempty}, \ref{th:B7_nonempty}, and Lemmas \ref{lem:B12_empty} and \ref{lem:B16_empty}. Before proving them, we shall give some conclusions about sets $B_k$ with $5\le k\le 18$ which are defined in Section \ref{sec:key}. 

For \(1 \leq k \leq 18\), let
\begin{align*}
  E_k \triangleq \{(d_x, d_y):  d_x, d_y ~\text{satisfy the conditions in}~ B_k\},
\end{align*}
i.e., $E_k$ collects deletion index pairs $(d_x,d_y)$ such that $\bm{z}=\bm{x}(d_x,e_x) = \bm{y}(d_y,e_y)\in B_k$.\footnote{We note that for \(X \in \{L, R\}\), the set \(E_k\) with \(1 \le k \le 6\) in this paper corresponds to \(\Lambda^X_{1,i}\) with \(1 \le i \le 3\) in \cite{Cai-25-arxiv-sd}, while \(E_k\) with \(7 \le k \le 18\) corresponds to \(\Lambda^X_{2,i}\) with \(1 \le i \le 6\). More precisely, given \(l \in \{1,2\}\), \(p_1 \in [3]\), \(p_2 \in [6]\), and \(X \in \{L, R\}\), for any \(d_x, d_y \in [n]\), if \((\bm{x}_{[n]\setminus d_x}, \bm{y}_{[n]\setminus d_y}) \in \Lambda^X_{l,p_l}\), then \((d_x, d_y)\) belongs to the corresponding set \(E_k\).} It can be easily checked that $B_k=\bigcup_{(d_x,d_y)\in E_k}\mathcal{S}(\bm{x}(d_x,0),\bm{y}(d_y,0))$, which suggests that the size of $B_k$ is related to $E_k$.  
Thus, we shall first determine \(E_k\), and then compute the size of \(\mathcal{B}(\bm{x},\bm{y})\) under different constraints. Let $d_H\triangleq d_H(\bm{x},\bm{y})$, we recall that \(j_i \triangleq j_i(\bm{x}, \bm{y})  \) 
denotes the \(i\)-th smallest index where $\bm{x}$ and $\bm{y}$ differ, for \(i \in [d_H]\). Then, we need the following two lemmas, which are useful for the remaining space.
\begin{lemma}\cite[Claims 1.1-1.3 and 2.1-2.6]{Cai-25-arxiv-sd}\label{lem:E_k_odd}
  For any two distinct sequences $\bm{x},\bm{y}\in\Sigma^n$, let $d_H\triangleq d_H(\bm{x},\bm{y})$, then the following results hold:
  \begin{itemize}
      \item $E_1=\{(j_2,j_{d_H})\}$ if $\bm{x}_{[j_2+1,j_{d_H}]}= \bm{y}_{[j_2,j_{d_H}-1]}$, and $E_1=\emptyset$ if $d_H(\bm{x}_{[j_2+1,j_{d_H}]}, \bm{y}_{[j_2,j_{d_H}-1]})\ge 1$.
      \item $E_3=\{(j_1,j_{d_H-1})\}$ if $\bm{x}_{[j_1+1,j_{d_H-1}]}= \bm{y}_{[j_1,j_{d_H-1}-1]}$, and $E_3=\emptyset$ if $d_H(\bm{x}_{[j_1+1,j_{d_H-1}]}, \bm{y}_{[j_1,j_{d_H-1}-1]})\ge 1$.
      \item $E_5=\{(j_1,j_{d_H})\}$ if $d_H(\bm{x}_{[j_1+1,j_{d_H}]}, \bm{y}_{[j_1,j_{d_H}-1]})= 1$, and $E_5=\emptyset$ if $d_H(\bm{x}_{[j_1+1,j_{d_H}]}, \bm{y}_{[j_1,j_{d_H}-1]})\ge 2$.
      \item If $d_H\ge 3$, then $E_7=\{(j_3,j_{d_H})\}$ if $\bm{x}_{[j_3+1,j_{d_H}]}= \bm{y}_{[j_3,j_{d_H}-1]}$, and $E_7=\emptyset$ if $d_H(\bm{x}_{[j_3+1,j_{d_H}]}, \bm{y}_{[j_3,j_{d_H}-1]})\ge 1$.
      \item If $d_H\ge 3$, then $E_9=\{(j_1,j_{d_H-2})\}$ if $\bm{x}_{[j_1+1,j_{d_H-2}]}= \bm{y}_{[j_1,j_{d_H-2}-1]}$, and $E_9=\emptyset$ if $d_H(\bm{x}_{[j_1+1,j_{d_H-2}]}, \bm{y}_{[j_1,j_{d_H-2}-1]})\ge 1$.
      \item If $d_H\ge 3$, then $E_{11}=\{(j_2,j_{d_H-1})\}$ if $\bm{x}_{[j_2+1,j_{d_H-1}]}= \bm{y}_{[j_2,j_{d_H-1}-1]}$, and $E_{11}=\emptyset$ if $d_H(\bm{x}_{[j_2+1,j_{d_H-1}]}, \bm{y}_{[j_2,j_{d_H-1}-1]})\ge 1$.
      \item $E_{13}=\{(j_2,j_{d_H})\}$ if $d_H(\bm{x}_{[j_2+1,j_{d_H}]}, \bm{y}_{[j_2,j_{d_H}-1]})=1$, and $E_{13}=\emptyset$ if $d_H(\bm{x}_{[j_2+1,j_{d_H}]}, \bm{y}_{[j_2,j_{d_H}-1]})\ge 2$.
      \item $E_{15}=\{(j_1,j_{d_H-1})\}$ if $d_H(\bm{x}_{[j_1+1,j_{d_H-1}]}, \bm{y}_{[j_1,j_{d_H-1}-1]})=1$, and $E_{15}=\emptyset$ if $d_H(\bm{x}_{[j_1+1,j_{d_H-1}]}, \bm{y}_{[j_1,j_{d_H-1}-1]})\ge 2$.
      \item $E_{17}\subseteq\{(k_1-1,j_{d_H}),(j_1,k_1')\}$ if $d_H(\bm{x}_{[j_1+1,j_{d_H}]}, \bm{y}_{[j_1,j_{d_H}-1]})=1$, where $k_1\le j_1<j_{d_H}<k_1'$. Moreover, $E_{17}=\{(j_1,j_{d_H})\}$ if $d_H(\bm{x}_{[j_1+1,j_{d_H}]}, \bm{y}_{[j_1,j_{d_H}-1]})= 2$, and $E_{17}=\emptyset$ if $d_H(\bm{x}_{[j_1+1,j_{d_H}]}, \bm{y}_{[j_1,j_{d_H}-1]})\ge 3$.
  \end{itemize}
\end{lemma}

\begin{lemma}\cite[Claims 1.1'-1.3' and 2.1'-2.6']{Cai-25-arxiv-sd}\label{lem:E_k_even}
  For any two distinct sequences $\bm{x},\bm{y}\in\Sigma^n$, let $d_H\triangleq d_H(\bm{x},\bm{y})$, then the following results hold:
  \begin{itemize}
      \item $E_2=\{(j_{d_H},j_{2})\}$ if $\bm{x}_{[j_2,j_{d_H}-1]}= \bm{y}_{[j_2+1,j_{d_H}]}$, and $E_2=\emptyset$ if $d_H(\bm{x}_{[j_2,j_{d_H}-1]}, \bm{y}_{[j_2+1,j_{d_H}]})\ge 1$.
      \item $E_4=\{(j_{d_H-1},j_1)\}$ if $\bm{x}_{[j_1,j_{d_H-1}-1]}= \bm{y}_{[j_1+1,j_{d_H-1}]}$, and $E_4=\emptyset$ if $d_H(\bm{x}_{[j_1,j_{d_H-1}-1]}, \bm{y}_{[j_1+1,j_{d_H-1}]})\ge 1$.
      \item $E_6=\{(j_{d_H},j_{1})\}$ if $d_H(\bm{x}_{[j_1,j_{d_H}-1]}, \bm{y}_{[j_1+1,j_{d_H}]})= 1$, and $E_6=\emptyset$ if $d_H(\bm{x}_{[j_1,j_{d_H}-1]}, \bm{y}_{[j_1+1,j_{d_H}]})\ge 2$.
      \item If $d_H\ge 3$, then $E_8=\{(j_{d_H},j_{3})\}$ if $\bm{x}_{[j_3,j_{d_H}-1]}= \bm{y}_{[j_3+1,j_{d_H}]}$, and $E_8=\emptyset$ if $d_H(\bm{x}_{[j_3,j_{d_H}-1]}= \bm{y}_{[j_3+1,j_{d_H}]})\ge 1$.
      \item If $d_H\ge 3$, then $E_{10}=\{(j_{d_H-2},j_1)\}$ if $\bm{x}_{[j_1,j_{d_H-2}-1]}= \bm{y}_{[j_1+1,j_{d_H-2}]}$, and $E_{10}=\emptyset$ if $d_H(\bm{x}_{[j_1,j_{d_H-2}-1]}= \bm{y}_{[j_1+1,j_{d_H-2}]})\ge 1$.
      \item If $d_H\ge 3$, then $E_{12}=\{(j_{d_H-1},j_2)\}$ if $\bm{x}_{[j_2,j_{d_H-1}-1]}= \bm{y}_{[j_2+1,j_{d_H-1}]}$, and $E_{12}=\emptyset$ if $d_H(\bm{x}_{[j_2,j_{d_H-1}-1]}, \bm{y}_{[j_2+1,j_{d_H-1}]})\ge 1$.
      \item $E_{14}=\{(j_{d_H},j_{2})\}$ if $d_H(\bm{x}_{[j_2,j_{d_H}-1]}, \bm{y}_{[j_2+1,j_{d_H}]})=1$, and $E_{14}=\emptyset$ if $d_H(\bm{x}_{[j_2,j_{d_H}-1]}, \bm{y}_{[j_2+1,j_{d_H}]})\ge 2$.
      \item $E_{16}=\{(j_{d_H-1},j_1)\}$ if $d_H(\bm{x}_{[j_1,j_{d_H-1}-1]}, \bm{y}_{[j_1+1,j_{d_H-1}]})=1$, and $E_{16}=\emptyset$ if $d_H(\bm{x}_{[j_1,j_{d_H-1}-1]}, \bm{y}_{[j_1+1,j_{d_H-1}]})\ge 2$.
      \item $E_{18}\subseteq\{(j_{d_H},m_1-1),(m_1',j_1)\}$ if $d_H(\bm{x}_{[j_1,j_{d_H}-1]}, \bm{y}_{[j_1+1,j_{d_H}]})=1$, where, $m_1\le j_1<j_{d_H}<m_1'$. Moreover, $E_{18}=\{(j_{d_H},j_1)\}$ if $d_H(\bm{x}_{[j_1,j_{d_H}-1]}, \bm{y}_{[j_1+1,j_{d_H}]})= 2$, and $E_{18}=\emptyset$ if $d_H(\bm{x}_{[j_1,j_{d_H}-1]}, \bm{y}_{[j_1+1,j_{d_H}]})\ge 3$.
  \end{itemize}
\end{lemma}

We note that  $E_k=\emptyset$ if and only if $B_k=\emptyset$. Thus, by using Lemmas \ref{lem:E_k_odd} and \ref{lem:E_k_even}, we can prove the following lemma.
\begin{lemma}\label{lem:E_k}
  For any two distinct sequences $\bm{x}, \bm{y} \in \C_1$, we can determine $E_k$ for $5\le k\le 18$ in Table \ref {tab:E_k}.
  
\begin{table*}[ht]
    \centering
    \caption{ For any two distinct sequences $\bm{x},\bm{y} \in \C_1$, and $5\le k\le 18$, we determine the set $E_k$, which collects the deletion index pairs $(d_x,d_y)$ such that $\bm{z}=\bm{x}(d_x,e_x) = \bm{y}(d_y,e_y)\in B_k$, where $d_H\triangleq d_H(\bm{x},\bm{y})\ge 4$.} 
    \begin{tabular}{| >{\centering\arraybackslash}p{6.5cm} | >{\centering\arraybackslash}p{6.5cm} |} 
        \hline
        $E_5\subseteq\{(j_1,j_{d_H})\}$ & $E_6\subseteq\{(j_{d_H},j_1)\}$ \\
        \hline 
        $E_7\subseteq\{(j_3,j_{d_H})\}$ & $E_8\subseteq\{(j_{d_H},j_3)\}$ \\
        \hline 
        $E_9\subseteq\{(j_1,j_{d_{H}-2})\}$ & $E_{10}\subseteq\{(j_{d_{H}-2},j_1)\}$ \\
        \hline 
        $E_{11}\subseteq\{(j_2,j_{d_H-1})\}$ & $E_{12}\subseteq\{(j_{d_H-1},j_2)\}$ \\
        \hline 
        $E_{13}\subseteq\{(j_2,j_{d_H})\}$ & $E_{14}\subseteq\{(j_{d_H},j_2)\}$ \\
        \hline 
        $E_{15}\subseteq\{(j_1,j_{d_H-1})\}$ & $E_{16}\subseteq\{(j_{d_H-1},j_1)\}$ \\
        \hline 
        \multicolumn{2}{|c|}{$E_{17}\subseteq\{(j_1,\ge j_{d_H}+1),(\le j_1-1,j_{d_H})\}$ when $B_5$ is nonempty and $E_{17}\subseteq\{(j_1,j_{d_H})\}$ when $B_5$ is empty} \\
        \hline 
        \multicolumn{2}{|c|}{$E_{18}\subseteq\{(\ge j_{d_H}+1,j_1),(j_{d_H},\le j_1-1)\}$ when $B_6$ is nonempty and $E_{18}\subseteq\{(j_{d_H},j_{1})\}$ when $B_6$ is empty}  \\
        \hline
    \end{tabular}
    \label{tab:E_k}
\end{table*}
  \end{lemma}
\begin{IEEEproof}
From Lemma \ref{lem:VT}, we know that $|\mathcal{D}(\bm{x},\bm{y})|=0$ and $|\mathcal{S}(\bm{x},\bm{y})|=0$. Moreover, $B_1,B_2,B_3$ and $B_4$ are empty, thus, $E_1,E_2,E_3$ and $E_4$ are empty. In addition, we have $d_H((\bm{x},\bm{y}))\ge 4$,  then we can prove this lemma as follows.
  \begin{itemize}
      \item The sets $E_k$ for $7\le k\le 12$ can be obtained directly by Lemmas \ref{lem:E_k_odd} and \ref{lem:E_k_even}.
      \item For sets $E_5,E_6,E_{17}$ and $E_{18}$. Since $|\mathcal{D}(\bm{x},\bm{y})|=0$, then we can obtain that $d_H(\bm{x}_{[j_1+1,j_{d_H}]}, \bm{y}_{[j_1,j_{d_H}-1]})\ge 1$ and $d_H(\bm{x}_{[j_1,j_{d_H}-1]}, \bm{y}_{[j_1+1,j_{d_H}]})\ge 1$ by Lemma \ref{lem:del_int}. Thus, we can derive that $E_5\subseteq\{(j_1,j_{d_H})\}$ and $E_6\subseteq\{(j_{d_H},j_1)\}$ by Lemmas \ref{lem:E_k_odd} and \ref{lem:E_k_even}. More detailedly, we have that $d_H(\bm{x}_{[j_1+1,j_{d_H}]}, \bm{y}_{[j_1,j_{d_H}-1]})=1$ when $E_5$ is nonempty and $d_H(\bm{x}_{[j_1+1,j_{d_H}]}, \bm{y}_{[j_1,j_{d_H}-1]})\ge 2$ when $E_5$ is empty. Moreover, we have $d_H(\bm{x}_{[j_1,j_{d_H}-1]}, \bm{y}_{[j_1+1,j_{d_H}]})=1$ when $E_6$ is nonempty and $_H(\bm{x}_{[j_1,j_{d_H}-1]}, \bm{y}_{[j_1+1,j_{d_H}]})\ge 2$ when $E_6$ is empty. Consequently, the sets $E_{17}$ and $E_{18}$ can be determined by Lemmas \ref{lem:E_k_odd} and \ref{lem:E_k_even}.
      \item For sets $E_{13},E_{14},E_{15}$ and $E_{16}$. Since $E_1$ is empty,  we have  $d_H(\bm{x}_{[j_2+1,j_{d_H}]}, \bm{y}_{[j_2,j_{d_H}-1]})\ge 1$ by Lemma \ref{lem:E_k_odd}. Thereby, we can obtain that $E_{13}\subseteq\{(j_2,j_{d_H})\}$ according to Lemma \ref{lem:E_k_odd}. By considering $R(\bm{x})$ and $R(\bm{y})$, or changing the positions of $\bm{x}$ and $\bm{y}$, we can also prove the conclusion for $E_{14},E_{15}$ and $E_{16}$ since $E_2,E_3$ and $E_4$ are empty.
  \end{itemize}
\end{IEEEproof}

The following lemma can be easily proved by Lemma \ref{lem:E_k}. 

\begin{lemma}\label{lem:set size}
    For any two distinct sequences $\bm{x},\bm{y} \in \C_1$, we have the following results:
    \begin{itemize}
        \item $|B_i|\le 2$ for $5\le i\le 16$;
        \item  $|B_i| \leq 2$ when $B_{i-12} = \emptyset$, and $|B_i| \leq 4$ when $B_{i-12} \neq \emptyset$, for $i\in \{17,18\}$.
    \end{itemize}
\end{lemma}
\begin{IEEEproof}
    Given any two distinct sequences $\bm{x},\bm{y} \in \C_1$, by Lemma \ref{lem:VT}, we have that $|\mathcal{D}(\bm{x},\bm{y})|=0$. Thus, we can conclude that $\bm{x}(d_x,0)\neq \bm{y}(d_y,0)$ for any $d_x,d_y\in [n]$. Then, from Lemma \ref{lem:sub}, we can obtain that $|\mathcal{S}(\bm{x}(d_x,0), \bm{y}(d_y,0))|\le 2$ for any given $d_x,d_y$. Furthermore, since $B_k=\bigcup_{(d_x,d_y)\in E_k}\mathcal{S}(\bm{x}(d_x,0),\bm{y}(d_y,0))$, we can obtain that $|B_k|\le 2|E_k|$. Thus, this lemma can be proved immediately by Lemma \ref{lem:E_k}.
\end{IEEEproof}

Then, by observation, we find some relationship between these sets, which is summarized in the following lemmas and corollaries. Moreover, some of them will be proved in the Appendix.

\begin{lemma}\label{lem:B11B9}
    For any  two distinct sequences \(\bm{x},\bm{y} \in \C_1\), if \(B_{11}\) is a nonempty set, then $B_5,B_7$ and $B_9$ are empty.  
    By symmetry, considering \(R(\bm{x})\) and \(R(\bm{y})\), if $B_{12}$ is a nonempty set, $B_6,B_8$ and $B_{10}$ are empty.
\end{lemma}

The following can be obtained directly from Lemma \ref{lem:B11B9}.
\begin{corollary}\label{cor:B11B9}
 For any  two distinct sequences \(\bm{x},\bm{y} \in \C_1\), if one of $B_5,B_7,B_9$ is a nonempty set, then \(B_{11}\) must be empty. By symmetry, considering \(R(\bm{x})\) and \(R(\bm{y})\), if one of $B_6,B_8,B_{10}$ is a nonempty set, then $B_{12}$ must be empty.    
\end{corollary}
\begin{lemma}\label{lem:B7B9}
    For any  two distinct sequences \(\bm{x},\bm{y} \in \C_1\), if \(d_H(\bm{x},\bm{y}) \geq 5\), then we have
    \begin{itemize}
        \item \(B_7\) and \(B_9\) cannot be nonempty  simultaneously;
        \item \(B_8\) and \(B_{10}\) cannot be nonempty simultaneously.
    \end{itemize}
\end{lemma}

\begin{IEEEproof}
    Assume both \(B_7\) and \(B_9\) are nonempty sets. Then, 
    by Lemma \ref{lem:E_k_odd}, we have that
$\bm{x}_{[j_3+1,j_{d_H}]}=\bm{y}_{[j_3,j_{d_H}-1]}$ and  $\bm{x}_{[j_1+1,j_{d_H-2}]}=\bm{x}_{[j_1,j_{d_H-2}-1]}$.
 Given \(d_H(\bm{x},\bm{y}) \geq 5\), it follows that \(j_{d_{H}-2} \geq j_3\), consequently, $\bm{x}_{[j_1+1,j_{d_H}]}=\bm{y}_{[j_1,j_{d_H}-1]}$, implying \(|\mathcal{D}(\bm{x}) \cap \mathcal{D}(\bm{y})| \geq 1\), which contradicts \(\bm{x},\bm{y}\in \C_1\) by Lemma \ref{lem:VT}.  
Thus, \(B_7\) and \(B_9\) cannot be nonempty simultaneously. By symmetry, changing the positions of $\bm{x}$ and $\bm{y}$, we can also obtain that \(B_8\) and \(B_{10}\) cannot be nonempty simultaneously.
\end{IEEEproof}

Lemmas \ref{lem:B11B9}, \ref{lem:B7B9} and Corollary \ref{cor:B11B9} show that some sets can not be nonempty at the same time under constraints about the Hamming distance between $\bm{x}$ and $\bm{y}$. In addition, the following Lemmas \ref{lem:B7B13} and \ref{lem:intersect}, and Corollaries \ref{cor:B7B13} and \ref{cor:intersect} suggest that there are common elements in some sets.  
  
\begin{lemma}\label{lem:B7B13}
For any two distinct sequences $\bm{x}, \bm{y} \in \C_1$, if $B_7$ and $B_{13}$ are nonempty, then $|B_7 \cup B_{13}| \leq 3$. Similarly, if $B_{10}$ and $B_{16}$ are nonempty, then $|B_{10} \cup B_{16}| \leq 3$.  
\end{lemma}

By swapping $\bm{x}$ and $\bm{y}$ and using the conclusion of Lemma \ref{lem:B7B13}, we can obtain the following corollary.
\begin{corollary}\label{cor:B7B13}
 For any two distinct sequences $\bm{x}, \bm{y} \in \C_1$, if $B_8,B_{14}$ are nonempty, then $|B_8\cup B_{14}|\le 3$. Similarly, if $B_{9},B_{15}$ are nonempty, then $|B_{9}\cup B_{15}|\le 3$.   
\end{corollary}
\begin{lemma}\label{lem:intersect}
For any two distinct sequences $\bm{x}, \bm{y} \in \C_1$, let $I = \{5,7,9,13,15,17\}$. If $B_5$ is a nonempty set, then $B_{13},B_{15}$ and $B_{17}$ are nonempty. Moreover, we have that
$|\bigcup_{i\in I} B_i|\le 8,|\bigcup_{i\in I\setminus\{7\}} B_i|\le 7$, $|\bigcup_{i\in I\setminus\{9\}} B_i|\le 7$ and $|\bigcup_{i\in I\setminus\{7,9\}} B_i|\le 6$.
\end{lemma}

By considering $R(\bm{x})$ and $R(\bm{y})$ and using the conclusion of Lemma \ref{lem:intersect}, we immediately obtain the following corollary.  

\begin{corollary}\label{cor:intersect}
For any two distinct sequences $\bm{x}, \bm{y} \in \C_1$, let $I = \{6,8,10,14,16,18\}$. If $B_6$ is a nonempty set, then $B_{14},B_{16}$ and $B_{18}$ are nonempty. Moreover, we have that $|\bigcup_{i\in I} B_i|\le 8,|\bigcup_{i\in I\setminus\{8\}} B_i|\le 7$,  $|\bigcup_{i\in I\setminus\{10\}} B_i|\le 7$ and $|\bigcup_{i\in I\setminus\{8,10\}} B_i|\le 6$.
\end{corollary}

The proofs of Lemmas \ref{lem:B11B9}, \ref{lem:B7B13} and \ref{lem:intersect} will be provided in the Appendix. Then we will use these lemmas to prove Theorems \ref{th:B5_nonempty}, \ref{th:B11_nonempty}, \ref{th:B7_nonempty} and Lemmas \ref{lem:B12_empty} and \ref{lem:B16_empty}.

\subsection{The proof of Theorem \ref{th:B5_nonempty}}

\begin{IEEEproof}
The conclusion is proved only for the case where $B_5$ is nonempty. The case where $B_6$ is nonempty follows similarly by considering $R(\bm{x})$ and $R(\bm{y})$.

Let $\bm{x}_i\triangleq \bm{x}_{[n]\setminus i}$ for $i\in [n]$ and define $I\triangleq \{5,6,7,8,9,10,13,14,15,16,17,18\}$.  
Since $\bm{x}, \bm{y} \in \C_1$, it follows by Lemma \ref{lem:VT} that $B_1,B_2,B_3$ and $B_4$ are empty and $d_H\ge 4$. Then from Lemma \ref{lem:E_k_odd}, we have that $d_H(\bm{x}_{[j_2+1,j_{d_H}]},\bm{y}_{[j_2,j_{d_H}-1]})\ge 1$ and $d_H(\bm{x}_{[j_1+1,j_{d_H-1}]},\bm{y}_{[j_1,j_{d_H-1}-1]})\ge 1$. Moreover, since $B_5$ is nonempty, we also have $d_H(\bm{x}_{[j_1+1,j_{d_H}]},\bm{y}_{[j_1,j_{d_H}-1]})= 1$, which implies that $d_H(\bm{x}_{[j_2+1,j_{d_H-1}]},\bm{y}_{[j_2,j_{d_H-1}-1]})= 1$. Thus, there exists $2 \leq t \leq d_H-2$ such that $d_H(\bm{x}_{[j_t+1,j_{t+1}]},\bm{y}_{[j_t,j_{t+1}-1]})= 1$. Consequently, we have $x_{j_t}=x_{j_{t+1}}$, $y_{j_t}=y_{j_{t+1}}$, and $r(\bm{x}_{[j_t+1,j_{t+1}-1]})\leq 2$  according to Lemma \ref{lem:run_number} and Remark \ref{rem:run_number}. In addition, we also have $\bm{x}_{[j_1+1,j_t]}=\bm{y}_{[j_1,j_t-1]}$ and $\bm{x}_{[j_{t+1}+1,j_{d_H}]}=\bm{y}_{[j_{t+1},j_{d_H}-1]}$, and $d_H(\bm{x}_{[j_t,j_{t+1}-1]},\bm{y}_{[j_t+1,j_{t+1}]})\ge 1$ is odd. 

Let 
\begin{align*}
   &\bm{c} \triangleq \bm{x}_{[j_1,j_t]}, \quad \bm{d}\triangleq \bm{y}_{[j_1,j_t]}, 
\end{align*}
\begin{align*}
   \bm{w} \triangleq &\bm{x}_{[j_{t+1},j_{d_H}]}, \quad \bm{v}\triangleq \bm{y}_{[j_{t+1},j_{d_H}]},
\end{align*}
\begin{align*}
   \bm{u} \triangleq &\bm{x}_{[j_t+1,j_{t+1}-1]}=\bm{y}_{[j_t+1,j_{t+1}-1]}.
\end{align*}
Then, the structures of $\bm{x}$ and $\bm{y}$ can be rewritten as
\begin{align*}
   \bm{x} = \bm{a}\bm{c} \bm{u}\bm{w}\bm{b}, \quad \bm{y} = \bm{a} \bm{d}\bm{u}\bm{v}\bm{b}. 
\end{align*}
In addition, let $m\triangleq |\bm{a}|$, $l_1\triangleq|\bm{c}|=j_t-j_1+1$, $l_2\triangleq |\bm{u}|$, and $l_3\triangleq |\bm{w}|=j_{d_H}-j_{t+1}+1$. Then, according to the analysis above, we have $l_1,l_3\ge 2$ and $c_{l_1}=w_1,d_{l_1}=v_1=\overline{c_{l_1}}=\overline{w_1}$ and $r(\bm{u})\le 2$. Moreover,
\begin{align*}
    \bm{c}_{[2,l_1]}=\bm{d}_{[1,l_1-1]},\quad  \bm{w}_{[2,l_3]}=\bm{v}_{[1,l_3-1]},\quad d_H(\bm{u}w_1,d_{l_1}\bm{u})=1,\quad d_H(c_{l_1}\bm{u},\bm{u}v_1)\ge 1.
\end{align*} 
Furthermore, since $d_H\geq 4$, we can obtain that
\begin{align*}
  d_H(\bm{x}_{j_{d_H-1}},\bm{y}_{j_2})=2+d_H(c_{l_1}\bm{u},\bm{u}v_1)\ge 3.
\end{align*}
Then, from Lemmas \ref{lem:sub} and \ref{lem:E_k}, and Corollary \ref{cor:B11B9}, we have that $B_{11}$ and $B_{12}$ are empty.  
Consequently, in this case, we can derive   
\[
\mathcal{B}(\bm{x},\bm{y})=\bigcup_{i\in I} B_i.
\]  

Next, we will prove this lemma based on the sizes of $d_H(\bm{c}_{[1,l_1-1]},\bm{d}_{[2,l_1]})$ and $d_H(\bm{w}_{[1,l_3-1]},\bm{v}_{[2,l_3]})$. 

\begin{enumerate}
    \item If $d_H(\bm{c}_{[1,l_1-1]},\bm{d}_{[2,l_1]})=0$ and $d_H(\bm{w}_{[1,l_3-1]},\bm{v}_{[2,l_3]})=0$, i.e., $\bm{c}_{[1,l_1-1]}=\bm{d}_{[2,l_1]}$ and $\bm{w}_{[1,l_3-1]}=\bm{v}_{[2,l_3]}$. Then combining with $\bm{c}_{[2,l_1]}=\bm{d}_{[1,l_1-1]}$ and $\bm{w}_{[2,l_3]}=\bm{v}_{[1,l_3-1]}$, we can obtain that $\bm{c}$ and $\bm{d}$ are alternating sequences by Lemma \ref{lem:period}. Similarly,  $\bm{w}$ and $\bm{v}$ are alternating sequences. Thus, the structures of $\bm{x}$ and $\bm{y}$ can be written as  
 \[
\bm{x}=\bm{a}\bm{c} \bm{u}\bm{w}\bm{b}, \quad \bm{y}=\bm{a} \overline{\bm{c}}\bm{u}\overline{\bm{w}}\bm{b},
\]
where $\bm{c}$ and $\bm{w}$ are alternating sequences of length at least two. In this case, $\tilde{\bm{x}}=\bm{c} \bm{u}\bm{w}$ and $\tilde{\bm{y}}=\overline{\bm{c}}\bm{u}\overline{\bm{w}}$ consist of at most four $^{\le}2$-periodic sequences.  
 Setting $l\triangleq l_1+l_2+l_3$, it is clear that $m+l=j_{d_H}$. Next, we will prove 
$|\mathcal{B}(\bm{x},\bm{y})| \le 13$.

\begin{itemize}
\item If $r(\bm{u})=2$.  
From Remark \ref{rem:run_number}, we can obtain that 
\[
d_H(c_{l_1}\bm{u},\bm{u}\overline{w_1})=3.
\]
Consequently, for any indices $i\ge j_{d_H-1}$ and $j\le j_2$, we have  
\[
d_H(\bm{x}_i,\bm{y}_j) \ge d_H(c_{l_1}\bm{u},\bm{u}\overline{w_1})=3.
\]
Thus, by Lemmas \ref{lem:sub} and \ref{lem:E_k}, the sets  $B_6,B_{14},B_{16}$, and $B_{18}$ are empty. Applying Lemmas \ref{lem:set size} and \ref{lem:intersect}, we conclude  that
\[
|\mathcal{B}(\bm{x},\bm{y})|\le 8+2\times2=12.
\]

\item If $r(\bm{u})\le 1$ and $l_1\ge 3$. Then, from Remark \ref{rem:run_number}, we know that  $d_H(c_{l_1}\bm{u},\bm{u}\overline{w_1})=1.$ Furthermore, since $c_{l_1}=w_1$ and $\overline{c_{l_1}}=\overline{w_1}$, then either $c_{l_1}\bm{u}w_1$ or $\overline{c_{l_1}}\bm{u}\overline{w_1}$ forms a run. Without loss of generality, assume $c_{l_1}\bm{u}w_1$ is a run (if $\overline{c_{l_1}}\bm{u}\overline{w_1}$ is a run, swapping $\bm{x}$ and $\bm{y}$, we can also prove the conclusion).  This implies 
\begin{align*}
  d_H(\bm{x}_{j_{d_H}},\bm{y}_{j_1})= d_H(\bm{c}\bm{u}\bm{w}_{[1,l_3-1]},\overline{\bm{c}}_{[2,l_1]}\bm{u}\overline{\bm{w}})=d_H(c_{l_1}\bm{u},\bm{u}\overline{w_1})=1,  
\end{align*}
 making $B_6$ nonempty. Since $l_1\ge 3$, we further derive that 
\begin{align*}
 d_H(\bm{x}_{j_3},\bm{y}_{j_{d_H}})=2+ d_H(\bm{u}w_1,\overline{c_{l_1}}\bm{u})=3,   
\end{align*}
\begin{align*}
  d_H(\bm{x}_{j_{d_H}},\bm{y}_{j_3})=2+ d_H(c_{l_1}\bm{u},\bm{u}\overline{w_1})=3.  
\end{align*}
By Lemmas \ref{lem:sub} and \ref{lem:E_k}, we have that $B_7$ and $B_8$ are empty. Thus, when $l_1\ge 3$ and $r(\bm{u})\le 1$, we can obtain
\begin{align*}
    \mathcal{B}(\bm{x},\bm{y})=\bigcup_{i\in I\setminus\{7,8\}} B_i.
\end{align*}
Similarly, by considering $R(\bm{x})$ and $R(\bm{y})$, we have that $B_9$ and $B_{10}$ are empty if $l_3\ge 3$ in this case. Consequently, from Lemma \ref{lem:intersect} and Corollary \ref{cor:intersect}, we can obtain 
\[
|\mathcal{B}(\bm{x},\bm{y})|\le 6+6=12.
\]
If $l_3=2$, in this case, Lemma \ref{lem:E_k} gives $E_{10}\subseteq\{(m+l_1,m+1)\}$ and $E_{16}\subseteq\{(m+l_1+l_2+1,m+1)\}$. Consequently, $B_{10}=B_{16}$ since $c_{l_1}\bm{u}w_1$ is a run. Then also from Lemma \ref{lem:intersect} and Corollary \ref{cor:intersect}, we have
\[
|\mathcal{B}(\bm{x},\bm{y})|\le 7+6=13.
\]
Thus, when $l_1\ge 3$ and $r(\bm{u})\le 1$, we have $|\mathcal{B}(\bm{x},\bm{y})|\le 13$.

\item If $r(\bm{u})\le 1$ and $l_3\ge 3$.  
By considering $R(\bm{x})$ and $R(\bm{y})$, and using the conclusion of previous case, we can also obtain that
$|\mathcal{B}(\bm{x},\bm{y})|\le 13$.

\item If $r(\bm{u})\le 1$ and $l_1=l_3=2$. 
Here,  
\[
\bm{x}=\bm{a}c_1c_2\bm{u}w_1w_2\bm{b}, \quad \bm{y}=\bm{a}\overline{c_1}\overline{c_2}\bm{u}\overline{w_1}\overline{w_2}\bm{b}.
\]
Similar to the case $r(\bm{u})\le 1$ and $l_1\ge 3$, we assume $c_2\bm{u}w_1$ forms a run. Then, $B_6$ is nonempty, and we can obtain $B_{10}=B_{16}$. By considering $R(\bm{x})$ and $R(\bm{y})$, we also have $B_7=B_{13}$.  Thus, when $l_1=l_3=2$ and $r(\bm{u})\le 1$, we have
\begin{align*}
    \mathcal{B}(\bm{x},\bm{y})=\bigcup_{i\in I\setminus\{7,10\}} B_i.
\end{align*}

From Lemma \ref{lem:E_k},  we can obtain that
\[
E_8\subseteq\{(m+l,m+l-1)\}, \quad E_9\subseteq\{(m+1,m+2)\}.
\]
It follows by Lemma \ref{lem:sub} that
\begin{align*}
 B_8\subseteq\{\bm{a}c_1\overline{c_2}\bm{u}w_1\bm{b},\bm{a}\overline{c_1}c_2\bm{u}w_1\bm{b}\},   
\end{align*}
\begin{align*}
 B_9\subseteq\{\bm{a}c_2\bm{u}w_1\overline{w_2}\bm{b},\bm{a}c_2\bm{u}\overline{w_1}w_2\bm{b}\}.   
\end{align*}
Clearly,  
\[
B_8\cap B_9=\{\bm{a}\overline{c_1}c_2\bm{u}w_1\bm{b}\}=\{\bm{a}c_2\bm{u}w_1\overline{w_2}\bm{b}\}.
\]
Applying Lemma \ref{lem:intersect} and Corollary \ref{cor:intersect}, we conclude  
\[
|\mathcal{B}(\bm{x},\bm{y})|\le 7+7-1=13.
\]
\end{itemize}
In conclusion, we have $|\mathcal{B}(\bm{x},\bm{y})|\le 13$ when $d_H(\bm{x}_{[j_1,j_t-1]},\bm{y}_{[j_1+1,j_t]})=0$ and $d_H(\bm{x}_{[j_{t+1},j_{d_H}-1]},\bm{y}_{[j_{t+1}+1,j_{d_H}]})=0$.

\item If $d_H(\bm{c}_{[1,l_1-1]},\bm{d}_{[2,l_1]})=0$ and $d_H(\bm{w}_{[1,l_3-1]},\bm{v}_{[2,l_3]})\neq 0$, i.e., $\bm{c}_{[1,l_1-1]}=\bm{d}_{[2,l_1]}$. Then combining with $\bm{c}_{[2,l_1]}=\bm{d}_{[1,l_1-1]}$, we can obtain that $\bm{c}$ and $\bm{d}$ are alternating sequences by Lemma \ref{lem:period}. Thus, the structures of $\bm{x}$ and $\bm{y}$ in this case are given by
\[\bm{x} = \bm{a}\bm{c} \bm{u}\bm{w}\bm{b}, \quad \bm{y} = \bm{a} \overline{\bm{c}}\bm{u}\bm{v}\bm{b},\]
where $\bm{c}$ is an alternating sequence of length at least two.  
Since $d_H(\bm{w}_{[2,l_3]},\bm{v}_{[1,l_3-1]})=0$ and $d_H(\bm{w}_{[1,l_3-1]},\bm{v}_{[2,l_3]})\neq 0$, and $d_H(\bm{u}w_1,\overline{c_{l_1}}\bm{u})=1$, it follows by Lemma \ref{lem:run_number} and Corollary \ref{cor:run_number} that
\[ d_H(\bm{w}_{[1,l_3-1]}, \bm{v}_{[2,l_3]}) \geq 2, \quad d_H(c_{l_1}\bm{u}, \bm{u}v_1) \geq 1. \]
Consequently, we can deduce that 
\begin{align*}
 d_H(\bm{x}_{j_{d_H}},\bm{y}_{j_3}) \ge 2 + d_H(\bm{w}_{[1,l_3-1]}, \bm{v}_{[2,l_3]}) \geq 4,   
\end{align*}
\begin{align*}
d_H(\bm{x}_{j_{d_H}},\bm{y}_{j_2}) = 1 + d_H(c_{l_1}\bm{u}, \bm{u}v_1) + d_H(\bm{w}_{[1,l_3-1]}, \bm{v}_{[2,l_3]}) \geq 4,  
\end{align*}
\begin{align*}
d_H(\bm{x}_{j_{d_H}},\bm{y}_{j_1}) = d_H(c_{l_1}\bm{u}, \bm{u}v_1) + d_H(\bm{w}_{[1,l_3-1]}, \bm{v}_{[2,l_3]}) \geq 3. 
\end{align*}
Then, by Lemmas \ref{lem:sub} and \ref{lem:E_k}, we know that $B_6,B_8,B_{14}$ and $B_{18}$ are  empty. In this case, let $I_1\triangleq I\setminus\{6,8,14,18\}=\{5,7,9,10,13,15,16,17\}$, then, we have
 \[
\mathcal{B}(\bm{x},\bm{y})=\bigcup_{i\in I_1} B_i.
\]  
From Lemmas \ref{lem:B7B13}  and \ref{lem:intersect}, we can obtain that
    \[ |\mathcal{B}(\bm{x},\bm{y})| \leq 7 + 3 = 10. \]
In conclusion, we have $|\mathcal{B}(\bm{x},\bm{y})|\le 10$ when $d_H(\bm{x}_{[j_1,j_t-1]},\bm{y}_{[j_1+1,j_t]})=0$ and $d_H(\bm{x}_{[j_{t+1},j_{d_H}-1]},\bm{y}_{[j_{t+1}+1,j_{d_H}]})\ge 0$.

 \item If $d_H(\bm{c}_{[1,l_1-1]},\bm{d}_{[2,l_1]})\neq 0$ and $d_H(\bm{w}_{[1,l_3-1]},\bm{v}_{[2,l_3]})=0$.  
By symmetry, considering $R(\bm{y})$ and $R(\bm{x})$ instead of $\bm{x}$ and $\bm{y}$, and using the conclusion of case 2), we also have  
$|\mathcal{B}(\bm{x},\bm{y})|\le 10$.

\item If $d_H(\bm{c}_{[1,l_1-1]},\bm{d}_{[2,l_1]})\neq 0$ and $d_H(\bm{w}_{[1,l_3-1]},\bm{v}_{[2,l_3]})\neq 0$. Since $d_H(\bm{w}_{[2,l_3]},\bm{v}_{[1,l_3-1]})=0$, $d_H(\bm{w}_{[1,l_3-1]},\bm{v}_{[2,l_3]})\neq 0$, and $d_H(\bm{u}w_1,\overline{c_{l_1}}\bm{u})=1$, it follows by Lemma \ref{lem:run_number} and Corollary \ref{cor:run_number} that
\[ d_H(\bm{w}_{[1,l_3-1]}, \bm{v}_{[2,l_3]}) \geq 2, \quad d_H(c_{l_1}\bm{u}, \bm{u}v_1) \geq 1. \] 
By the conclusion of case 2), we can also obtain that
 $|\mathcal{B}(\bm{x},\bm{y})| \leq 10$. 

     \end{enumerate}
Then the proof of this lemma is completed.
\end{IEEEproof}

\subsection{The proof of Theorem \ref{th:B11_nonempty}}

\begin{IEEEproof}
The conclusion is proved only for the case where $B_{11}$ is nonempty. The case where $B_{12}$ is nonempty follows similarly by considering $R(\bm{x})$ and $R(\bm{y})$.

Let $\bm{x}_i\triangleq \bm{x}_{[n]\setminus i}$ for $i\in [n]$ and define $I\triangleq \{8,10,11,12,13,14,15,16,17,18\}$.  
Since $\bm{x}, \bm{y} \in \C_1$ and $B_{11}$ is nonempty, it follows by Lemmas \ref{lem:VT} and \ref{lem:B11B9} that $B_1,B_2,B_3,B_4,B_7$ and $B_9$ are empty, and $d_H\ge 4$. Thus, in this case, we have 
 \begin{align*}
     \mathcal{B}(\bm{x},\bm{y})=\bigcup_{i\in I} B_i.
 \end{align*}

   From Lemma \ref{lem:E_k_odd}, we know that $\bm{x}_{[j_2+1,j_{d_H-1}]}=\bm{y}_{[j_2,j_{d_H-1}-1]}$ when $B_{11}$ is nonempty. Next, we define
\begin{align*}
  d_1 \triangleq d_{H}(\bm{x}_{[j_{d_H-1}+1,j_{d_H}]},\bm{y}_{[j_{d_H-1},j_{d_H}-1]}),~~
    d_2 \triangleq d_{H}(\bm{x}_{[j_1+1,j_2]},\bm{y}_{[j_1,j_2-1]}),  
\end{align*}
\begin{align*}
  d_3\triangleq d_H(\bm{x}_{[j_{d_H-1},j_{d_H}-1]},\bm{y}_{[j_{d_H-1}+1,j_{d_H}]}),~~
    d_4\triangleq d_H(\bm{x}_{[j_1,j_2-1]},\bm{y}_{[j_1+1,j_2]}).  
\end{align*}

 Then by Lemma \ref{lem:E_k_odd}, we can obtain that 
 \begin{align*}
   d_{H}(\bm{x}_{[j_{2}+1,j_{d_H}]},\bm{y}_{[j_2,j_{d_H}-1]})=d_1\ge 1,   
 \end{align*}
 \begin{align*}
    d_{H}(\bm{x}_{[j_{1}+1,j_{d_H-1}]},\bm{y}_{[j_1,j_{d_H-1}-1]})=d_2\ge 1,  
 \end{align*}
  since $B_1$ and $B_3$ are empty.  
\begin{itemize}
\item If $B_{12}$ is a nonempty set. Then we have $\bm{x}_{[j_2,j_{d_H-1}-1]}=\bm{y}_{[j_2+1,j_{d_H-1}]}$ from Lemma \ref{lem:E_k_even}. Moreover, by symmetry, we can also obtain $d_3,d_4 \ge 1 $ by Lemma \ref{lem:E_k} since $B_2$ and $B_4$ are empty. Furthermore, from Lemma \ref{lem:period}, we can obtain that $\bm{x}_{[j_2,j_{d_H-1}]}$ and $\bm{y}_{[j_2,j_{d_H-1}]}$ are alternating sequences with length no less than two. Consequently, the structures of $\bm{x}$ and $\bm{y}$ can be given by
\[
\bm{x} = \bm{a}x_{j_1}\bm{u}\bm{w}\bm{v}x_{j_{d_H}}\bm{b}, \quad 
\bm{y} = \bm{a}\overline{x_{j_1}}\bm{u}\overline{\bm{w}}\bm{v}\overline{x_{j_{d_H}}}\bm{b},
\]
where $\bm{w}$ is an alternating sequence of length at least two.
On the other hand, we have $B_8$ and $B_{10}$ are empty according to Corollary \ref{cor:B11B9}, thus, in this case, we have
 \begin{align*}
     \mathcal{B}(\bm{x},\bm{y})=\bigcup_{i\in I\setminus\{8,10\}} B_i.
 \end{align*}
Then, we will prove this lemma based on the sizes of $d_1,d_2,d_3$ and $d_4$.
\begin{itemize}
    \item If $d_1 \geq 2$, then we have 
    \begin{align*}
      d_H(\bm{x}_{j_1},\bm{y}_{j_{d_H}})=d_H(\bm{x}_{[j_1+1,j_{d_H}]},\bm{y}_{[j_1,j_{d_H}-1]})= d_1 + d_2 \geq 3,
    \end{align*}
    \begin{align*}    
    d_H(\bm{x}_{j_2},\bm{y}_{j_{d_H}})=1+d_H(\bm{x}_{[j_2+1,j_{d_H}]},\bm{y}_{[j_2,j_{d_H}-1]})= 1+d_1 \geq 3.  
    \end{align*}
    Consequently, by Lemmas \ref{lem:sub} and \ref{lem:E_k}, we know that $B_{13}$ and $B_{17}$ are empty. 
    \begin{itemize}
        \item In this case, if both $B_{14}$ and $B_{16}$ are nonempty sets, then by Lemma \ref{lem:E_k_even}, we can obtain that
    \begin{align*}  &d_H(\bm{x}_{[j_2,j_{d_H}-1]},\bm{y}_{[j_2+1,j_{d_H}]})=  d_3=1,  
    \end{align*}
    \begin{align*}
     d&_H(\bm{x}_{[j_1,j_{d_H-1}-1]},\bm{y}_{[j_1+1,j_{d_H-1}]})= d_4=1. 
    \end{align*}
    Thus, we have  $r(\bm{u}),r(\bm{v})\leq 2$ by Lemma \ref{lem:run_number} and Remark \ref{rem:run_number}. Under these conditions, we have 
    \begin{align*}
  \tilde{\bm{x}}=x_{j_1}\bm{u}\bm{w}\bm{v}x_{j_{d_H}}, \quad  \tilde{\bm{y}}=\overline{x_{j_1}}\bm{u}\overline{\bm{w}}\bm{v}\overline{x_{j_{d_H}}},   
    \end{align*}
    consisting of at most five $^{\leq}2$-periodic sequences and two additional symbols from $\{0,1\}$. Moreover, in this case, we can derive that 
    \[
    |\mathcal{B}(\bm{x},\bm{y})|\leq 2 \times 6 = 12,
    \]
    by Lemma \ref{lem:set size}. 
    
    \item If either $B_{14}$ or $B_{16}$ is an empty set, then from Lemma \ref{lem:set size}, we will have
    \[
    |\mathcal{B}(\bm{x},\bm{y})|\leq 2 \times 5 = 10.
    \]
    \end{itemize}
    
    \item By symmetry, considering $R(\bm{y})$ and $R(\bm{x})$ instead of $\bm{x}$ and $\bm{y}$, the same argument applies for $d_2 \geq 2$, yielding $\tilde{\bm{x}},\tilde{\bm{y}}$ are combinations of at most five $^{\le}2$-periodic sequences and two additional symbols from $\{0,1\}$ when $|\mathcal{B}(\bm{x},\bm{y})|\le 12$, otherwise, $|\mathcal{B}(\bm{x},\bm{y})|\le 10$.

    \item If $d_1 = d_2 = 1$, then according to Lemma \ref{lem:run_number} and Remark \ref{rem:run_number}, the structures of $\bm{x}$ and $\bm{y}$ can be rewritten that
    \[
    \bm{x}=\bm{a}x_{j_1} \bm{u}x_{j_2} \bm{w}x_{j_{d_H-1}} \bm{v} x_{j_{d_H}} \bm{b}, \quad
    \bm{y}=\bm{a} \overline{x_{j_1}} \bm{u}  \overline{x_{j_2} \bm{w} x_{j_{d_H-1}}}\bm{v} \overline{x_{j_{d_H}}}\bm{b},
    \]
    where $x_{j_1}=x_{j_2}$, $x_{j_{d_H-1}}=x_{j_{d_H}}$, $ r(\bm{u})\leq 2$, $r(\bm{v})\leq 2$, and $x_{j_2} \bm{w}x_{j_{d_H-1}}$ is an alternating sequence of length at least two. In this case, we have 
    \begin{align*}
  \tilde{\bm{x}}=x_{j_1} \bm{u}x_{j_2} \bm{w}x_{j_{d_H-1}} \bm{v} x_{j_{d_H}}, \quad  \tilde{\bm{y}}=\overline{x_{j_1}} \bm{u}  \overline{x_{j_2} \bm{w} x_{j_{d_H-1}}}\bm{v} \overline{x_{j_{d_H}}},   
    \end{align*}
    consisting of at most five $^{\leq}2$-periodic sequences and two additional symbols from $\{0,1\}$. Then, we shall prove  $|\mathcal{B}(\bm{x},\bm{y})|\le 12$.
  \begin{itemize}
      \item If $r(\bm{u})=2$, then we have $d_4=3$ by Remark \ref{rem:run_number}, {which leads to} 
    \begin{align*}
    d_H&(\bm{x}_{j_{d_H-1}},\bm{y}_{j_{1}})=1+d_H(\bm{x}_{[j_1,j_{d_H-1}-1]},\bm{y}_{[j_1+1,j_{d_H-1}]})= 1+d_4\geq 4, 
     \end{align*}
    \begin{align*}
    &d_H(\bm{x}_{j_{d_H}},\bm{y}_{j_{1}})=d_H(\bm{x}_{[j_1,j_{d_H}-1]},\bm{y}_{[j_1+1,j_{d_H}]})=d_3+d_4\geq 4.    
    \end{align*}
    Consequently, we have $B_{16}$ and $B_{18}$ are empty according to  Lemmas \ref{lem:sub} and \ref{lem:E_k}. Thus, from Lemma \ref{lem:set size}, we can obtain that 
    \[
    |\mathcal{B}(\bm{x},\bm{y})|\leq 2 \times 6 = 12.
    \]
      
   \item If $r(\bm{v})=2$, by symmetry, considering $R(\bm{y})$ and $R(\bm{x})$ instead of $\bm{x}$ and $\bm{y}$, and using the conclusion of case $r(\bm{u})=2$, we can also obtain that $B_{14}$ and $B_{18}$ are empty. Consequently, $|\mathcal{B}(\bm{x},\bm{y})|\leq 12$.
   
   \item If $ r(\bm{u}), r(\bm{v})\leq 1$. Here, $x_{j_1} \bm{u}x_{j_2}$ or $\overline{x_{j_1}} \bm{u}  \overline{x_{j_2}}$ forms a run, and $x_{j_{d_H-1}} \bm{v} x_{j_{d_H}}$ or $\overline{x_{j_{d_H-1}}}\bm{v} \overline{x_{j_{d_H}}}$ forms a run. It can be easily checked that at least two sets among $B_{13},B_{14},B_{15},B_{16},B_{17},B_{18}$ are contained in $B_{11} \cup B_{12}$ according to Lemma \ref{lem:E_k}, thus, by Lemma \ref{lem:set size}, we have 
    \[
    |\mathcal{B}(\bm{x},\bm{y})|\leq 2\times 6= 12.
    \]
  \end{itemize}  
     
\end{itemize}
    
     Consequently, for any possible values of $d_1$ and $d_2$, the conclusions in this lemma are correct. By considering $R(\bm{x})$ and $R(\bm{y})$, the conclusions above hold for $d_3$ and $d_4$. Therefore, when $B_{12}$ is nonempty, this lemma can be proved.        		
    \item If $B_{12}$ is an empty set and $B_{8}$ is a nonempty set. Then we can obtain that $d_H(\bm{x}_{[j_{2},j_{d_H-1}-1]},y_{[j_{2}+1,j_{d_H-1}]})\neq 0$ and  $d_3=0$ by Lemma \ref{lem:E_k_even}. Furthermore, since $d_H(\bm{x}_{[j_{2}+1,j_{d_H-1}]},y_{[j_{2},j_{d_H-1}-1]})=0$ and $d_1\ge 1$, according to Lemma \ref{lem:run_number}, Remark \ref{rem:run_number} and Corollary \ref{cor:run_number}, we have $d_H(\bm{x}_{[j_{2},j_{d_H-1}-1]},y_{[j_{2}+1,j_{d_H-1}]})\ge 2$ is even and  $d_1= 2$. Thus we can obtain that 
    \begin{align*}
     &d_H(\bm{x}_{j_{2}},\bm{y}_{j_{d_H}})=1+d_H(\bm{x}_{[j_2+1,j_{d_H}]},\bm{y}_{[j_2,j_{d_H}-1]})\ge 1+d_1= 3, 
      \end{align*}
    \begin{align*}
    d_H(\bm{x}_{j_{d_H}},\bm{y}_{j_2})&=1+d_H(\bm{x}_{[j_2,j_{d_H}-1]},\bm{y}_{[j_2+1,j_{d_H}]})\ge 1+d_H(\bm{x}_{[j_{2},j_{d_H-1}-1]},y_{[j_{2}+1,j_{d_H-1}]})\ge 3, 
    \end{align*}
    \begin{align*}  
     d_H(\bm{x}_{j_{d_H-1}},\bm{y}_{j_1})&=1+d_H(\bm{x}_{[j_1,j_{d_H-1}-1]},\bm{y}_{[j_1+1,j_{d_H-1}]})\ge 1+d_H(\bm{x}_{[j_{2},j_{d_H-1}-1]},y_{[j_{2}+1,j_{d_H-1}]})\ge 3,  
    \end{align*}
     which implies that $B_{13},B_{14}$ and $B_{16}$ are empty according to Lemmas \ref{lem:sub} and \ref{lem:E_k}. In this case, we have
 \begin{align*}
     \mathcal{B}(\bm{x},\bm{y})=\bigcup_{i\in I\setminus\{12,13,14,16\}} B_i.
 \end{align*}
    \begin{itemize}
     \item  If $B_{15}$ is a nonempty set, then, from Lemma \ref{lem:E_k_odd}, we can obtain that
   \begin{align*}
   d_H(\bm{x}_{[j_1+1,j_{d_{H}-1}]},\bm{y}_{[j_1,j_{d_{H}-1}-1]})=d_2=1.
   \end{align*}
 Thus, we have $d_4\ge 1$ according to Lemma \ref{lem:run_number} and Remark \ref{rem:run_number}. Consequently, 
  \begin{align*}
   &d_H(\bm{x}_{j_{d_H-2}},\bm{y}_{j_1})=2+d_H(\bm{x}_{[j_1,j_{d_{H}-2}-1]},\bm{y}_{[j_1+1,j_{d_{H}-2}]})\ge 2+ d_4\ge 3, 
    \end{align*}
    \begin{align*}
    d_H(\bm{x}_{j_{d_H}},&\bm{y}_{j_1})=d_H(\bm{x}_{[j_1,j_{d_H}-1]},\bm{y}_{[j_1+1,j_{d_H}]})\ge  d_H(\bm{x}_{[j_{2},j_{d_H-1}-1]},y_{[j_{2}+1,j_{d_H-1}]})+d_4\ge 3,  
  \end{align*}   
 which implies that $B_{10}$ and $B_{18}$ are empty according to Lemmas \ref{lem:sub} and \ref{lem:E_k}. Therefore, by Lemma \ref{lem:set size}, we can obtain  $|\mathcal{B}(\bm{x},\bm{y})|\le2\times 4= 8$. 
 
 \item If $B_{15}$ is an empty set, then  we can also obtain  $|\mathcal{B}(\bm{x},\bm{y})|\le2\times 5= 10$ by Lemma \ref{lem:set size}.
    \end{itemize} 
 In conclusion, if $B_{12}$ is an empty set and $B_8$ is a nonempty set, this lemma can be proved.    
			
    \item If $B_{12}$ and $B_{8}$ are empty.  Then we can obtain that $d_H(\bm{x}_{[j_{2},j_{d_H-1}-1]},y_{[j_{2}+1,j_{d_H-1}]})\neq 0$ by Lemma \ref{lem:E_k_even}. Furthermore, since $d_H(\bm{x}_{[j_{2}+1,j_{d_H-1}]},y_{[j_{2},j_{d_H-1}-1]})=0$, by Corollary \ref{cor:run_number}, we have $d_H(\bm{x}_{[j_{2},j_{d_H-1}-1]},y_{[j_{2}+1,j_{d_H-1}]})\ge 2$ is even. Thus, we can obtain that 
    \begin{align*}
     &d_H(\bm{x}_{j_{d_H}},\bm{y}_{j_2})=1+d_H(\bm{x}_{[j_2,j_{d_H}-1]},\bm{y}_{[j_2+1,j_{d_H}]})\ge 1+d_H(\bm{x}_{[j_{2},j_{d_H-1}-1]},y_{[j_{2}+1,j_{d_H-1}]})\ge 3, 
      \end{align*}
    \begin{align*}
     d_H&(\bm{x}_{j_{d_H-1}},\bm{y}_{j_1})=1+d_H(\bm{x}_{[j_1,j_{d_H-1}-1]},\bm{y}_{[j_1+1,j_{d_H-1}]})\ge 1+d_H(\bm{x}_{[j_{2},j_{d_H-1}-1]},y_{[j_{2}+1,j_{d_H-1}]})\ge 3,  
    \end{align*}
     which implies that $B_{14}$ and $B_{16}$ are empty according to Lemmas \ref{lem:sub} and \ref{lem:E_k}. In this case, we have
 \begin{align*}
     \mathcal{B}(\bm{x},\bm{y})=\bigcup_{i\in I\setminus\{8,12,14,16\}} B_i.
 \end{align*}
 
 At this point, if $B_{13}$ is a nonempty set, then, from Lemma \ref{lem:E_k_odd}, we can derive that
   \begin{align*}
   d_H(\bm{x}_{[j_2+1,j_{d_{H}}]},\bm{y}_{[j_2,j_{d_{H}}-1]})=d_1=1.
   \end{align*}
 Thus, we have $d_3\ge 1$ according to Lemma \ref{lem:run_number} and Remark \ref{rem:run_number}. Consequently,
 \begin{align*}
d_H(\bm{x}_{j_{d_H}},\bm{y}_{j_{1}})\ge d_H(\bm{x}_{[j_{2},j_{d_H-1}-1]},y_{[j_{2}+1,j_{d_H-1}]})+d_3\ge 3,   
 \end{align*}
which implies that  $B_{18}$ is an empty set according to Lemmas \ref{lem:sub} and \ref{lem:E_k}. Then, by Lemma \ref{lem:set size}, we can obtain $|\mathcal{B}(\bm{x},\bm{y})|\le2\times 5= 10$. Moreover, if $B_{13}$ is an empty set, then we also have $|\mathcal{B}(\bm{x},\bm{y})|\le 2\times 5= 10$. 
Thereby, if $B_{12}$ and $B_8$ are empty, the conclusion of this lemma can be proved. 
\end{itemize}

Then the proof of this lemma is completed.
\end{IEEEproof}

\subsection{The proof of Theorem \ref{th:B7_nonempty}}

\begin{IEEEproof}
The conclusion is proved only for the case where $B_7$ is nonempty. The case where $B_8$ is nonempty follows similarly by exchanging the positions of $\bm{x}$ and $\bm{y}$.

Let $\bm{x}_i\triangleq \bm{x}_{[n]\setminus i}$ for $i\in [n]$ and define $I\triangleq \{7,8,9,10,12,13,14,16,17,18\}$.  
Since $\bm{x}, \bm{y} \in \C_1$ and $B_{7}$ is nonempty, it follows by Corollary \ref{cor:B11B9} that $B_{11}$ is an empty set. Moreover, we have $d_H\ge 4$ according to Lemma \ref{lem:VT}. From Lemma \ref{lem:E_k_odd}, we know that $\bm{x}_{[j_3+1,j_{d_H}]}=\bm{y}_{[j_3,j_{d_H}-1]}$ when $B_{7}$ is nonempty. Let
   \begin{align*}
d_1\triangleq d_{H}(\bm{x}_{[j_1+1,j_2]},\bm{y}_{[j_1,j_2-1]}), \quad d_2\triangleq d_{H}(\bm{x}_{[j_2+1,j_3]},\bm{y}_{[j_2,j_3-1]}),
 \end{align*}
 \begin{align*}
d_3\triangleq d_{H}(\bm{x}_{[j_2,j_3-1]}, \bm{y}_{[j_2+1,j_3]}), \quad d_4\triangleq d_{H}(\bm{x}_{[j_1,j_2-1]},\bm{y}_{[j_1+1,j_2]}).
\end{align*}
Since $B_1$ and $B_5$ are empty, we have that $d_H(\bm{x}_{[j_2+1,j_{d_H}]},\bm{y}_{[j_2,j_{d_H}-1]})\ge 1$ and $d_H(\bm{x}_{[j_1+1,j_{d_H}]},\bm{y}_{[j_1,j_{d_H}-1]})\ge 2$ according to Lemma \ref{lem:E_k_odd}. On the other hand, it is clear that
  \begin{align*}   &d_H(\bm{x}_{[j_2+1,j_{d_H}]},\bm{y}_{[j_2,j_{d_H}-1]})=d_2,
   \end{align*}
    \begin{align*}
    d_H&(\bm{x}_{[j_1+1,j_{d_H}]},\bm{y}_{[j_1,j_{d_H}-1]})=d_1+d_2,
  \end{align*}
 which implies $d_2\ge 1$ and $d_1+d_2\ge 2$. Consequently, we have 
 \begin{align*}
     d_H(\bm{x}_{j_1},\bm{y}_{j_{d_H-1}})=1+d_H(\bm{x}_{[j_1+1,j_{d_H-1}]},\bm{y}_{[j_1,j_{d_H-1}-1]})= 1+d_1+d_2\ge 3,
 \end{align*}
which implies $B_{15}$ is an empty set according to Lemmas \ref{lem:sub} and \ref{lem:E_k}. Thus, in this case, we have 
 \begin{align*}
     \mathcal{B}(\bm{x},\bm{y})=\bigcup_{i\in I} B_i.
 \end{align*}
Next, we will prove this lemma based on the sets $B_8$ and $B_{12}$.

\begin{itemize}
  \item If $B_8$ is a nonempty set, then  $B_{12}$ is empty and $\bm{x}_{[j_3,j_{d_H}-1]}=\bm{y}_{[j_3+1,j_{d_H}]}$ from Lemma \ref{lem:E_k_even} and Corollary \ref{cor:B11B9}. Combining with $\bm{x}_{[j_3+1,j_{d_H}]}=\bm{y}_{[j_3,j_{d_H}-1]}$, we have that $\bm{x}_{[j_3,j_{d_H}]}$ and $\bm{y}_{[j_3,j_{d_H}]}$ are alternating sequences from Lemma \ref{lem:period}. Thus, the structures of $\bm{x}$ and $\bm{y}$ are as follows:
    \[
    \bm{x}=\bm{a}x_{j_1} \bm{u}x_{j_2}\bm{v}x_{j_{3}} \bm{w} x_{j_{d_H}} \bm{b}, \quad \bm{y}=\bm{a} \overline{x_{j_1}} \bm{u}  \overline{x_{j_2}} \bm{v} \overline{x_{j_{3}}\bm{w}x_{j_{d_H}}}\bm{b},
    \]
    where $x_{j_3} \bm{w}x_{j_{d_H}}$ is an alternating sequence with length at least two. Moreover, since $B_8$ is a nonempty set, and $B_2$ and $B_6$ are empty, by symmetry, swapping $\bm{x}$ and $\bm{y}$, and using the conclusions about $B_7$, we also have  $d_3\ge 1$ and $d_3+d_4\ge 2$. Furthermore, we can obtain that $B_{16}$ is nonempty. Thus, in this case, we have
 \begin{align*}
     \mathcal{B}(\bm{x},\bm{y})=\bigcup_{i\in I\setminus\{12,16\}} B_i.
 \end{align*} 

\begin{itemize}
  \item If $r(\bm{u})\ge 3$, then $d_1,d_4\ge 2$ according to Remark \ref{rem:run_number}. Consequently, we can obtain that
\begin{align*}
&d_H(\bm{x}_{j_{d_H}},\bm{y}_{j_1})=d_H(\bm{x}_{[j_1,j_{d_H}-1]},\bm{y}_{[j_1+1,j_{d_H}]})= d_3+d_4\ge 3,
\end{align*}
\begin{align*} &d_H(\bm{x}_{j_1},\bm{y}_{j_{d_H}})=d_H(\bm{x}_{[j_1+1,j_{d_H}]},\bm{y}_{[j_1,j_{d_H}-1]})= d_1+d_2\ge 3, 
\end{align*}
\begin{align*} 
 d_H(&\bm{x}_{j_1},\bm{y}_{j_{d_H-2}})=2+d_H(\bm{x}_{[j_1+1,j_{d_{H}-2}]},\bm{y}_{[j_1,j_{d_{H}-2}-1]})\ge 2+d_1\ge 4, 
\end{align*}
\begin{align*} 
 d_H(&\bm{x}_{j_{d_H-2}},\bm{y}_{j_1})=2+d_H(\bm{x}_{[j_1,j_{d_{H}-2}-1]},\bm{y}_{[j_1+1,j_{d_{H}-2}]})\ge 2+d_4\ge 4,  
\end{align*}
which implies $B_{9},B_{10},B_{17}$ and $B_{18}$ are empty according to Lemmas \ref{lem:sub} and \ref{lem:E_k}. Thus, from Lemma \ref{lem:set size}, we can conclude  $|\mathcal{B}(\bm{x},\bm{y})|\le 2\times 4= 8$.

  \item If $r(\bm{v})\ge 4$, then $d_2,d_3\ge 3$ according to Remark \ref{rem:run_number}. Consequently, we can obtain that
\begin{align*}
&d_H(\bm{x}_{j_{d_H}},\bm{y}_{j_1})=d_H(\bm{x}_{[j_1,j_{d_H}-1]},\bm{y}_{[j_1+1,j_{d_H}]})= d_3+d_4\ge 3,
\end{align*}
\begin{align*}  &d_H(\bm{x}_{j_1},\bm{y}_{j_{d_H}})=d_H(\bm{x}_{[j_1+1,j_{d_H}]},\bm{y}_{[j_1,j_{d_H}-1]})= d_1+d_2\ge 3, 
\end{align*}
\begin{align*} 
d&_H(\bm{x}_{j_2},\bm{y}_{j_{d_H}})=1+d_H(\bm{x}_{[j_2+1,j_{d_{H}}]},\bm{y}_{[j_2,j_{d_{H}}-1]})= 1+d_2\ge 4, 
\end{align*}
\begin{align*}  d&_H(\bm{x}_{j_{d_H}},\bm{y}_{j_2})=1+d_H(\bm{x}_{[j_2,j_{d_{H}}-1]},\bm{y}_{[j_2+1,j_{d_{H}}]})= 1+d_3\ge 4,   
\end{align*}
which implies $B_{13},B_{14},B_{17}$ and $B_{18}$ are empty according to Lemmas \ref{lem:sub} and \ref{lem:E_k}. Thus, from Lemma \ref{lem:set size}, we can also conclude $|\mathcal{B}(\bm{x},\bm{y})|\le 2\times 4= 8$.

  \item If $r(\bm{u})\le 2$ and $ r(\bm{v})\le 3$, then 
  \begin{align*}
      \tilde{\bm{x}}=x_{j_1} \bm{u}x_{j_2}\bm{v}x_{j_{3}} \bm{w} x_{j_{d_H}},\quad  \tilde{\bm{y}}= \overline{x_{j_1}} \bm{u}  \overline{x_{j_2}} \bm{v} \overline{x_{j_{3}}\bm{w}x_{j_{d_H}}},
  \end{align*}
   are combinations of at most six $^{\le}2$-periodic sequences and two additional symbols from $\{0,1\}$, in this case, we shall prove $|\mathcal{B}(\bm{x},\bm{y})|\le 13$.
\begin{itemize}
  \item If $B_{13}$ is a nonempty set, according to Lemma \ref{lem:E_k_odd},  we have 
  \begin{align*}
  d_H(\bm{x}_{[j_2+1,j_{d_{H}}]},\bm{y}_{[j_2,j_{d_{H}}-1]})= d_2=1,
  \end{align*}
   which means $d_2=1$ and $d_1\ge 1$. Furthermore, we have that $d_3\ge 1$ is odd from Lemma \ref{lem:run_number} and Remark \ref{rem:run_number}. 
  
  If $d_3=1$, then $d_4\ge 1$. Consequently, 
  \begin{align*}
   d_H(\bm{x}_{j_{d_H-2}},\bm{y}_{j_1})=2+d_H(\bm{x}_{[j_1,j_{d_{H}-2}-1]},\bm{y}_{[j_1+1,j_{d_{H}-2}]})\ge 2+d_4\ge 3,
\end{align*}
which means $B_{10}$ is an empty set according to Lemmas \ref{lem:sub} and \ref{lem:E_k}. In addition, we know that $|B_8\cup B_{14}|\le 3$ from Corollary \ref{cor:B7B13}, combining with Lemma \ref{lem:set size}, we have $|\mathcal{B}(\bm{x},\bm{y})|\le 3+2\times 5=13$.

If $d_3>1$, then $d_3\ge 3$. Consequently,
\begin{align*}
&d_H(\bm{x}_{j_{d_H}},\bm{y}_{j_1})=d_H(\bm{x}_{[j_1,j_{d_H}-1]},\bm{y}_{[j_1+1,j_{d_H}]})= d_3+d_4\ge 3,
\end{align*}
\begin{align*}  d&_H(\bm{x}_{j_{d_H}},\bm{y}_{j_2})=1+d_H(\bm{x}_{[j_2,j_{d_{H}}-1]},\bm{y}_{[j_2+1,j_{d_{H}}]})= 1+d_3\ge 4,   
\end{align*}
which implies $B_{14}$ and $B_{18}$ are empty according to Lemmas \ref{lem:sub} and \ref{lem:E_k}. Thus, from Lemma \ref{lem:set size}, we can conclude that $|\mathcal{B}(\bm{x},\bm{y})|\le 2\times 6= 12$.

  \item If $B_{13}$ is an empty set, then we also have $|\mathcal{B}(\bm{x},\bm{y})|\le 3+2\times 5=13$ by Corollary \ref{cor:B7B13}, and Lemma \ref{lem:set size}. 
\end{itemize}
\end{itemize}
In conclusion, when $B_{8}$ is a nonempty set, this lemma can be proved.
 \item If $B_{8}$ is an empty set and $B_{12}$ is a nonempty set.  Then we have that $d_H(\bm{x}_{[j_3,j_{d_H}-1]},\bm{y}_{[j_3+1,j_{d_H}]})\ge 1$ and $\bm{x}_{[j_2,j_{d_H-1}-1]}=\bm{y}_{[j_2+1,j_{d_H-1}]}$ from Lemma \ref{lem:E_k_even}. Consequently, $d_3=0$. In addition, since $\bm{x}_{[j_3+1,j_{d_H}]}=\bm{y}_{[j_3,j_{d_H}-1]}$, thus, we have that $d_H(\bm{x}_{[j_3,j_{d_H}-1]},\bm{y}_{[j_3+1,j_{d_H}]})\ge 2$ is even according to Corollary \ref{cor:run_number}. On the other hand, 
 \begin{align*}
  d_H(\bm{x}_{[j_3,j_{d_H}-1]},\bm{y}_{[j_3+1,j_{d_H}]})= d_H(\bm{x}_{[j_{d_H-1},j_{d_H}-1]},\bm{y}_{[j_{d_H-1}+1,j_{d_H}]}),   
 \end{align*}
 
 since $\bm{x}_{[j_2,j_{d_H-1}-1]}=\bm{y}_{[j_2+1,j_{d_H-1}]}$. Then, we can obtain  $d_H(\bm{x}_{[j_{d_H-1},j_{d_H}-1]},\bm{y}_{[j_{d_H-1}+1,j_{d_H}]})= 2$ from Lemma \ref{lem:run_number} and Remark \ref{rem:run_number} since $\bm{x}_{[j_{d_H-1}+1,j_{d_H}]}=\bm{y}_{[j_{d_H-1},j_{d_H}-1]}$. Furthermore, we can derive that 
 \begin{align*}
 d_H(\bm{x}_{[j_2,j_{d_{H}}-1]},\bm{y}_{[j_2+1,j_{d_{H}}]})= d_H(\bm{x}_{[j_{d_H-1},j_{d_H}-1]},\bm{y}_{[j_{d_H-1}+1,j_{d_H}]})= 2,   
\end{align*}
 which implies $B_{14}$ is an empty set according to Lemma \ref{lem:E_k_even}. 
  Moreover, since $d_3=0$ and $d_2\ge 1$, we have $d_2= 2$ by Lemma \ref{lem:run_number} and Remark \ref{rem:run_number}. Consequently, we can obtain that
   \begin{align*}
     d_H(\bm{x}_{j_2},\bm{y}_{j_{d_H}})=1+d_H(\bm{x}_{[j_2+1,j_{d_{H}}]},\bm{y}_{[j_2,j_{d_{H}}-1]})= 1+d_2= 3,  
   \end{align*}
   which implies $B_{13}$ is an empty set according to Lemmas \ref{lem:sub} and \ref{lem:E_k}. Thus, in this case, we have
   \begin{align*}
     \mathcal{B}(\bm{x},\bm{y})=\bigcup_{i\in I\setminus\{8,13,14\}} B_i.
 \end{align*} 
   \begin{itemize}
       \item If $B_{16}$ is nonempty. Then we have $d_H(\bm{x}_{[j_1,j_{d_{H-1}}-1]},\bm{y}_{[j_1+1,j_{d_{H-1}}]})=d_4=1$ from Lemma \ref{lem:E_k_even}. Thus, we can obtain that $d_1\ge 1$ is odd by Lemma \ref{lem:run_number} and Remark \ref{rem:run_number}. Consequently,
  \begin{align*}
   &d_H(\bm{x}_{j_1},\bm{y}_{j_{d_H-2}})=2+d_H(\bm{x}_{[j_1+1,j_{d_{H}-2}]},\bm{y}_{[j_1,j_{d_{H}-2}-1]})\ge 2+d_1\ge 3, 
   \end{align*}
\begin{align*} 
  &d_H(\bm{x}_{j_{d_H-2}},\bm{y}_{j_1})=2+d_H(\bm{x}_{[j_1,j_{d_{H}-2}-1]},\bm{y}_{[j_1+1,j_{d_{H}-2}]})= 2+d_4= 3,
 \end{align*}
\begin{align*}     d_H(\bm{x}_{j_{d_H}},&\bm{y}_{j_1})=d_H(\bm{x}_{[j_1,j_{d_H}-1]},\bm{y}_{[j_1+1,j_{d_H}]})= d_4+d_H(\bm{x}_{[j_{d_H-1},j_{d_H}-1]},\bm{y}_{[j_{d_H-1}+1,j_{d_H}]})\ge 3,
 \end{align*}
     which implies $B_{9}, B_{10}$ and $B_{18}$ are empty according to Lemmas \ref{lem:sub} and \ref{lem:E_k}. Then, from Lemma \ref{lem:set size}, we can conclude that $|\mathcal{B}(\bm{x},\bm{y})|\le 2\times 4= 8$.
      \item If $B_{16}$ is empty. Then we have $d_H(\bm{x}_{[j_1,j_{d_{H-1}}-1]},\bm{y}_{[j_1+1,j_{d_{H-1}}]})=d_4\ge 2$ from Lemma \ref{lem:E_k_even}. Consequently,
 \begin{align*}
      &d_H(\bm{x}_{j_{d_H-2}},\bm{y}_{j_1})=2+d_H(\bm{x}_{[j_1,j_{d_{H}-2}-1]},\bm{y}_{[j_1+1,j_{d_{H}-2}]})= 2+d_4\ge 4,
 \end{align*}
 \begin{align*}       d_H(\bm{x}_{j_{d_H}},&\bm{y}_{j_1})=d_H(\bm{x}_{[j_1,j_{d_H}-1]},\bm{y}_{[j_1+1,j_{d_H}]})= d_4+d_H(\bm{x}_{[j_{d_H-1},j_{d_H}-1]},\bm{y}_{[j_{d_H-1}+1,j_{d_H}]})\ge 4,
 \end{align*}
      which implies $ B_{10}$ and $B_{18}$ are empty according to Lemmas \ref{lem:sub} and \ref{lem:E_k}. Thus, from Lemma \ref{lem:set size}, we can also conclude that $|\mathcal{B}(\bm{x},\bm{y})|\le 2\times 4= 8$.
   \end{itemize}
   In conclusion, if $B_{8}$ is an empty set and $B_{12}$ is a nonempty set, we have $|\mathcal{B}(\bm{x},\bm{y})|\le 8$.
 \item If $B_8$ and  $B_{12}$ are empty. Then we have  $d_H(\bm{x}_{[j_3,j_{d_H}-1]},\bm{y}_{[j_3+1,j_{d_H}]})\ge 1$ and $d_H(\bm{x}_{[j_2,j_{d_H-1}-1]},\bm{y}_{[j_2+1,j_{d_H-1}]})\ge 1$ from Lemma \ref{lem:E_k_even}. On the other hand, since $\bm{x}_{[j_3+1,j_{d_H}]}=\bm{y}_{[j_3,j_{d_H}-1]}$,  we have that $d_H(\bm{x}_{[j_3,j_{d_H}-1]},\bm{y}_{[j_3+1,j_{d_H}]})\ge 2$ is even according to Corollary \ref{cor:run_number}. Consequently, we can obtain  
 \begin{align*}
 d_H(\bm{x}_{[j_2,j_{d_{H}}-1]},\bm{y}_{[j_2+1,j_{d_{H}}]})\ge d_H(\bm{x}_{[j_{3},j_{d_H}-1]},\bm{y}_{[j_{3}+1,j_{d_H}]})\ge 2,   
\end{align*}
 which implies $B_{14}$ is an empty set according to Lemma \ref{lem:E_k_even}. Thus, in this case, we have
   \begin{align*}
     \mathcal{B}(\bm{x},\bm{y})=\bigcup_{i\in I\setminus\{8,12,14\}} B_i.
 \end{align*} 
 \begin{itemize}
     \item If $d_H=4$. Then we have $j_3=j_{d_H-1}$. Moreover, the structures of $\bm{x}$ and $\bm{y}$ are as follows:
    \[
    \bm{x}=\bm{a}x_{j_1} \bm{u}x_{j_2}\bm{v}x_{j_{3}} \bm{w} x_{j_{4}} \bm{b}, \quad \bm{y}=\bm{a} \overline{x_{j_1}} \bm{u}  \overline{x_{j_2}} \bm{v} \overline{x_{j_{3}}}\bm{w}\overline{x_{j_{4}}}\bm{b}.
    \]
    From the analysis above, we know that $d_H(x_{j_2}\bm{v},\bm{v} \overline{x_{j_{3}}})\ge 1$ and $d_H(\bm{w} x_{j_{4}},\overline{x_{j_{3}}}\bm{w})=0$. Consequently, we can conclude $d_H(x_{j_{3}} \bm{w},\bm{w}\overline{x_{j_{4}}})=2$, $r(\bm{w})=1$ and $r(\bm{v})\ge 0$ according to Lemma \ref{lem:run_number} and Remark \ref{rem:run_number}. Furthermore, we can obtain that
    \begin{align*}
        d_H(\bm{x}_{j_{d_H}},\bm{y}_{j_1})=d_H(x_{j_1} \bm{u}x_{j_2}\bm{v}x_{j_{3}}\bm{w},\bm{u}  \overline{x_{j_2}} \bm{v} \overline{x_{j_{3}}}\bm{w}\overline{x_{j_{4}}}) \ge d_H(x_{j_2}\bm{v},\bm{v} \overline{x_{j_{3}}})+d_H(x_{j_{3}} \bm{w},\bm{w}\overline{x_{j_{4}}})\ge 3,
        \end{align*}
      which implies $B_{18}$ is an empty set according to Lemmas \ref{lem:sub} and \ref{lem:E_k}.
      \begin{itemize}
          \item If $r(\bm{v})\ge 3$, then $d_2,d_3\ge 2$ according to Remark \ref{rem:run_number}. Consequently, we can obtain 
         \begin{align*}
               &d_H(\bm{x}_{j_2},\bm{y}_{j_{d_H}})=1+d_H(\bm{v}x_{j_{3}} \bm{w} x_{j_{4}},\overline{x_{j_2}} \bm{v} \overline{x_{j_{3}}}\bm{w})= 1+d_2\ge 3, 
         \end{align*}
         \begin{align*} 
              d_H(&\bm{x}_{j_{d_H-1}},\bm{y}_{j_1})=1+d_H(x_{j_1} \bm{u}x_{j_2}\bm{v},\bm{u}  \overline{x_{j_2}} \bm{v} \overline{x_{j_{3}}})=1+ d_3+d_4\ge 3,   
         \end{align*}
     which implies $B_{13}$ and $B_{16}$ are empty according to Lemmas \ref{lem:sub} and \ref{lem:E_k}. Thus, from Lemma \ref{lem:set size}, we can conclude that $|\mathcal{B}(\bm{x},\bm{y})|\le 2\times 4= 8$.
     \item If $r(\bm{u})\ge 2$, then $d_1,d_4\ge 1$ according to Remark \ref{rem:run_number}. Consequently, we can obtain 
     \begin{align*}
         d_H(&\bm{x}_{j_1},\bm{y}_{j_{d_H-2}})=2+d_H(\bm{u}x_{j_2},\overline{x_{j_1}} \bm{u})= 2+d_1\ge 3, 
    \end{align*}
    \begin{align*} 
         d_H(&\bm{x}_{j_{d_H-2}},\bm{y}_{j_1})=2+d_H(x_{j_1} \bm{u},\bm{u}  \overline{x_{j_2}})= 2+d_4\ge 3,  
    \end{align*}
   which implies $B_{9}$ and $B_{10}$ are empty according to Lemmas \ref{lem:sub} and \ref{lem:E_k}. Thus, from Lemma \ref{lem:set size}, we can also conclude that $|\mathcal{B}(\bm{x},\bm{y})|\le 2\times 4= 8$.
   \item If $ r(\bm{v})\le 2$ and $r(\bm{u})=1$, then
    \begin{align*}
      \tilde{\bm{x}}=x_{j_1} \bm{u}x_{j_2}\bm{v}x_{j_{3}} \bm{w} x_{j_{4}},\quad  \tilde{\bm{y}}= \overline{x_{j_1}} \bm{u}  \overline{x_{j_2}} \bm{v} \overline{x_{j_{3}}}\bm{w}\overline{x_{j_{4}}},
   \end{align*}
   are combinations of at most four $^{\le}2$-periodic sequences and four additional symbols from $\{0,1\}$. Furthermore, from Lemma \ref{lem:set size}, we can obtain that $|\mathcal{B}(\bm{x},\bm{y})|\le 2\times 6= 12$.
      \end{itemize}
  \item If $d_H\ge 5$. Then we have $B_9$ is an empty set since $B_7$ is nonempty by Lemma \ref{lem:B7B9}. 
     \begin{itemize}
         \item If $d_3\ge 1$, then we have
        \begin{align*}
             &d_H(\bm{x}_{j_{d_H-2}},\bm{y}_{j_1})=2+d_H(\bm{x}_{[j_1,j_{d_{H}-2}-1]},\bm{y}_{[j_1+1,j_{d_{H}-2}]})\ge 2+d_3\ge 3,
         \end{align*}
         \begin{align*} d_H(\bm{x}&_{j_{d_H}},\bm{y}_{j_1})=d_H(\bm{x}_{[j_1,j_{d_H}-1]},\bm{y}_{[j_1+1,j_{d_H}]})\ge  d_3+d_H(\bm{x}_{[j_{3},j_{d_H}-1]},\bm{y}_{[j_3+1,j_{d_H}]})\ge 3, 
           \end{align*}
      which implies $B_{10}$ and $B_{18}$ are empty according to Lemmas \ref{lem:sub} and \ref{lem:E_k}. Thus, from Lemma \ref{lem:set size}, we can conclude that $|\mathcal{B}(\bm{x},\bm{y})|\le 2\times 4= 8$. 
      \item If $d_3=0$, then since $d_2\ge 1$, we have $d_2=2$ from Lemma \ref{lem:run_number} and Remark \ref{rem:run_number}. In addition, since $d_H(\bm{x}_{[j_2,j_{d_H-1}-1]},\bm{y}_{[j_2+1,j_{d_H-1}]})\ge 1$ and $d_3=0$, we have $d_H(\bm{x}_{[j_3,j_{d_H-1}-1]},\bm{y}_{[j_3+1,j_{d_H-1}]})\ge 1$. On the other hand, we know that $\bm{x}_{[j_3+1,j_{d_H-1}]}=\bm{y}_{[j_3,j_{d_H-1}-1]}$, then, we can obtain that $d_H(\bm{x}_{[j_3,j_{d_H-1}-1]},\bm{y}_{[j_3+1,j_{d_H-1}]})\ge 2$ is even by Corollary \ref{cor:run_number}. Consequently, 
      \begin{align*}
      &d_H(\bm{x}_{j_2},\bm{y}_{j_{d_H}})=1+d_H(\bm{x}_{[j_2+1,j_{d_{H}}]},\bm{y}_{[j_2,j_{d_{H}}-1]})= 1+d_2= 3,
      \end{align*}
      \begin{align*} 
      d_H(\bm{x}_{j_{d_H-1}},\bm{y}_{j_1})&=1+ d_H(\bm{x}_{[j_{1},j_{d_H-1}-1]},\bm{y}_{[j_1+1,j_{d_H-1}]})\ge 1+d_H(\bm{x}_{[j_3,j_{d_H-1}-1]},\bm{y}_{[j_3+1,j_{d_H-1}]})\ge 3,
     \end{align*}
   which implies $B_{13}$ and $B_{16}$ are empty according to Lemmas \ref{lem:sub} and \ref{lem:E_k}. Thus, from Lemma \ref{lem:set size}, we can also conclude that $|\mathcal{B}(\bm{x},\bm{y})|\le 2\times 4= 8$.
     \end{itemize}
 \end{itemize}
 In conclusion, when $B_{8}$ and $B_{12}$ are empty, this lemma can be proved.
\end{itemize}

Then the proof of this lemma is completed.
\end{IEEEproof}

\subsection{The proof of Lemmas \ref{lem:B12_empty} and  \ref{lem:B16_empty}}
We note that Lemma \ref{lem:B16_empty} follows by Lemmas \ref{lem:VT} and \ref{lem:set size} directly. Thus, we only need to prove Lemma \ref{lem:B12_empty}. 
\begin{IEEEproof}
    Let $I\triangleq\{13,14,15,16,17,18\}$, then from Lemma \ref{lem:VT}, we have that
    \begin{align*}
     \mathcal{B}(\bm{x},\bm{y})=\bigcup_{i\in I} B_i.
 \end{align*} 
    Since $B_{11}$ is empty, we have $d_{H}(\bm{x}_{[j_{2}+1,j_{d_H-1}]},\bm{y}_{[j_{2},j_{d_H-1}-1]})\ge 1$ according to Lemma \ref{lem:E_k_odd}. In this case, we claim that 
    $B_{13}$ is empty or $B_{15}$ is empty. Otherwise, if $B_{13}$ and $B_{15}$ are both nonempty, then by Lemma \ref{lem:E_k_odd}, we can obtain that $d_{H}(\bm{x}_{[j_{2}+1,j_{d_H}]},\bm{y}_{[j_{2},j_{d_H}-1]})= 1$ and $d_{H}(\bm{x}_{[j_{1}+1,j_{d_H-1}]},\bm{y}_{[j_{1},j_{d_H-1}-1]})= 1$. Since $d_{H}(\bm{x}_{[j_{2}+1,j_{d_H-1}]},\bm{y}_{[j_{2},j_{d_H-1}-1]})\ge 1$, we can conclude that $d_{H}(\bm{x}_{[j_{2}+1,j_{d_H-1}]},\bm{y}_{[j_{2},j_{d_H-1}-1]})= 1$ and $d_{H}(\bm{x}_{[j_1+1,j_2]},\bm{y}_{[j_1,j_2-1]})= d_{H}(\bm{x}_{[j_{d_H-1}+1,j_{d_H}]},\bm{y}_{[j_{d_H-1},j_{d_H}-1]})=0$. Consequently, 
    \begin{align*}
     d_{H}(\bm{x}_{[j_{1}+1,j_{d_H}]},\bm{y}_{[j_{1},j_{d_H}-1]})=& d_{H}(\bm{x}_{[j_1+1,j_2]},\bm{y}_{[j_1,j_2-1]})\\
     &+ d_{H}(\bm{x}_{[j_{2}+1,j_{d_H-1}]},\bm{y}_{[j_{2},j_{d_H-1}-1]}) \\
     &+ d_{H}(\bm{x}_{[j_{d_H-1}+1,j_{d_H}]},\bm{y}_{[j_{d_H-1},j_{d_H}-1]})= 1, 
    \end{align*}
  which implies that $B_5$ is nonempty by Lemma \ref{lem:E_k}. This contradicts the condition $B_5$ is empty, thus, $B_{13}$ is empty or $B_{15}$ is empty. By considering $R(\bm{x})$ and $R(\bm{y})$, we can also obtain that $B_{14}$ is empty or $B_{16}$ is empty. Consequently, we have $|\mathcal{B}(\bm{x},\bm{y})|\le 2\times 4= 8$ according to Lemma \ref{lem:set size}. 
\end{IEEEproof}	

\section{conclusion}\label{sec:concl}

In this paper, we study \((n, N; \mathcal{B})\)-reconstruction codes when the intersections between distinct single-deletion balls and single-substitution balls are empty.  
Our results show that when \(N\) is set to \(14\), \(11\), \( 9\), and \(5\), the redundancy of binary \((n, N; \mathcal{B})\)-reconstruction codes can be \(\log n +3\), \(\log n+12\log\log n+O(1)\), \(2\log n+12\log\log n+O(1)\), and \(3 \log n + 4\), respectively.  We leave the construction of 
\((n, N; \mathcal{B})\)-reconstruction codes for other values of \(N\) for future work. 

\section*{appendix}\label{sec:app}

\subsection{The proof of Lemma \ref{lem:VT}}\label{subsec:VT}\label{app:A}

\begin{IEEEproof}
    It is clear that $\C_1$ can correct a single edit error since it is a subcode of the single-edit correcting code $\mathrm{VT}_{s}(n)$. Thus, we can obtain that $|\mathcal{D}(\bm{x},\bm{y})|=0$ and $|\mathcal{S}(\bm{x},\bm{y})|=0$. Furthermore, from Lemma \ref{lem:sub}, we know that \(d_H(\bm{x},\bm{y})\ge 3\) since $|\mathcal{S}(\bm{x},\bm{y})|=0$. If \(d_H(\bm{x},\bm{y})=3\),  then 
    \begin{align*}
      |\mathrm{VT}^{0}(\bm{x})-\mathrm{VT}^{0}(\bm{y})|=|x_{j_1}+x_{j_2}+x_{j_3}-y_{j_1}-y_{j_2}-y_{j_3}|\in \{1,3\},  
    \end{align*} 
    which contradicts the condition $\mathrm{VT}^{0}(\bm{x}) \equiv \mathrm{VT}^{0}(\bm{y}) \equiv s_0 \pmod{4}$. Thus, we have $d_H(\bm{x},\bm{y})\ge 4$. On the other hand, if $B_1$ is a nonempty set, then we have $\bm{x}_{[j_{2}+1,j_{d_H}]}=\bm{y}_{[j_{2},j_{d_H}-1]}$ according to Lemma \ref{lem:E_k_odd}. Under this condition, if $y_{j_1}=1$, then $x_{j_1} = 0$, we claim that $x_{j_2} = 1$ and $y_{j_{d_H}} = 0$, otherwise
    \[
    |\mathrm{VT}^{0}(\bm{x}) - \mathrm{VT}^{0}(\bm{y})| = |x_{j_1} - y_{j_1} + x_{j_2} - y_{j_{d_H}}| \in \{1,2\},
    \]
    which also contradicts the condition $\mathrm{VT}^{0}(\bm{x}) \equiv \mathrm{VT}^{0}(\bm{y}) \equiv s_0 \pmod{4}$. However, when $x_{j_1} = y_{j_{d_H}} = 0$ and $y_{j_1} = x_{j_2} = 1$, we will have
    \[
    \sum_{i=1}^{n} i x_i - \sum_{i=1}^{n} i y_i = j_2+\sum_{i=j_2+1}^{j_{d_H}}x_i-j_1.
    \]
    It is clear that $0<j_2+\sum_{i=j_2+1}^{j_{d_H}}x_i-j_1<2n$, which contradicts $\bm{x}, \bm{y} \in \mathrm{VT}_{s}(n)$. 

    Similarly, if $x_{j_1} = 1$, then $y_{j_1} = 0$, and we can claim that $x_{j_2} = 0$ and $y_{j_{d_H}} = 1$. Under this condition, we have 
    \begin{align*}
   \sum_{i=1}^{n} i x_i - \sum_{i=1}^{n} i y_i = j_1+\sum_{i=j_2+1}^{j_{d_H}}x_i-j_{d_H}.     
    \end{align*}
It is clear that $-2n<j_1+\sum_{i=j_2+1}^{j_{d_H}}x_i-j_{d_H} \le j_1-j_2<0$, which also contradicts $\bm{x},\bm{y} \in \mathrm{VT}_{s}(n)$. Thus $B_1$ must be an empty set. By swapping $\bm{x}$ and $\bm{y}$, we can obtain that $B_2$ is empty. Furthermore, by considering $R(\bm{x})$ and $R(\bm{y})$, we can prove that $B_3$ and $B_4$ are empty.
\end{IEEEproof}

\subsection{The proof of Lemma \ref{lem:B11B9}}\label{subsec:B11B9}

\begin{IEEEproof}
For any two different sequences $\bm{x},\bm{y} \in \C_1$, let $d_H\triangleq d_H(\bm{x},\bm{y})$, then we have $d_H\ge 4$ by Lemma \ref{lem:VT}. Since $B_{11}$ is a nonempty set, we have $\bm{x}_{[j_2+1,j_{d_H-1}]}=\bm{y}_{[j_2,j_{d_H-1}-1]}$ by  Lemma \ref{lem:E_k_odd}. 

Suppose that $B_5$ is a nonempty set, then by Lemma \ref{lem:E_k_odd}, we have $d_H(\bm{x}_{[j_1+1,j_{d_H}]},\bm{y}_{[j_1,j_{d_H}-1]})=1$. Since $\bm{x}_{[j_2+1,j_{d_H-1}]}=\bm{y}_{[j_2,j_{d_H-1}-1]}$, it follows that $d_H(\bm{x}_{[j_1+1,j_2]},\bm{y}_{[j_1,j_2-1]})=1$ or $d_H(\bm{x}_{[j_{d_H-1}+1,j_{d_H}]},\bm{y}_{[j_{d_H-1},j_{d_H}-1]})=1$. Thus, $\bm{x}_{[j_2+1,j_{d_H}]}=\bm{y}_{[j_2,j_{d_H}-1]}$ or $\bm{x}_{[j_1+1,j_{d_H-1}]}=\bm{y}_{[j_1,j_{d_H-1}-1]}$, which implies that $B_1$ or $B_3$ is a nonempty set, contradicting the assumption \(\bm{x},\bm{y} \in \C_1\) by Lemma \ref{lem:VT}.
Therefore, $B_5$ is an empty set.

Similarly, if $B_7$ is a nonempty set, Lemma \ref{lem:E_k_odd} gives that  $\bm{x}_{[j_3+1,j_{d_H}]}=\bm{y}_{[j_3,j_{d_H}-1]}$. Since $d_H\ge 4$, we have $j_3\le j_{d_H-1}$, combining with $\bm{x}_{[j_2+1,j_{d_H-1}]}=\bm{y}_{[j_2,j_{d_H-1}-1]}$, we can obtain that $\bm{x}_{[j_2+1,j_{d_H}]}=\bm{y}_{[j_2,j_{d_H}-1]}$, which suggests $B_1$ is a nonempty set. This contradicts the assumption \(\bm{x},\bm{y} \in \C_1\) by Lemma \ref{lem:VT}.
Thus, $B_7$ is an empty set. 

By symmetry, considering $R(\bm{y})$ and $R(\bm{x})$ instead of $\bm{x}$ and $\bm{y}$, we can obtain that \(B_3\) must be nonempty if \(B_9\) is nonempty,  leading to a contradiction.  
Therefore, \(B_9\) must be empty. 
\end{IEEEproof}

Before proving Lemmas \ref{lem:B7B13} and \ref{lem:intersect}, we need the following observation.

\begin{observation}\label{obs:E_k}
  For any $d_x,d_y\in [n]$, and any two distinct sequences $\bm{x},\bm{y}\in \Sigma^n$. If there exists $\lambda\in [n-1]$ such that $x(d_x,0)_{\lambda}\neq y(d_y,0)_{\lambda}$, then when $d_x<d_y$, we have
  \begin{equation*}
        \begin{cases}
            x_{\lambda}\neq y_{\lambda},\quad & \text{if}~~ \lambda\in [1,d_x-1],\\
            x_{\lambda+1}\neq y_{\lambda},\quad &\text{if}~~\lambda\in [d_x,d_y-1],\\
            x_{\lambda+1}\neq y_{\lambda+1},\quad &\text{if}~~\lambda\in [d_y,n-1].	
        \end{cases}		
    \end{equation*}

When $d_x>d_y$, we have
  \begin{equation*}
        \begin{cases}
            x_{\lambda}\neq y_{\lambda},\quad & \text{if}~~ \lambda\in [1,d_y-1],\\
            x_{\lambda}\neq y_{\lambda+1},\quad &\text{if}~~\lambda\in [d_y,d_x-1],\\
            x_{\lambda+1}\neq y_{\lambda+1},\quad &\text{if}~~\lambda\in [d_x,n-1].	
        \end{cases}		
    \end{equation*}
\end{observation}
\subsection{The proof of Lemma \ref{lem:B7B13}}\label{subsec:B7B13}

\begin{IEEEproof}
We prove the first part of this lemma, and the second follows by considering $R(\bm{x})$ and $R(\bm{y})$.
From Lemma \ref{lem:E_k_odd}, we know that $d_H(\bm{x}_{[j_2+1,j_{d_H}]},\bm{y}_{[j_2,j_{d_H}-1]})=1$ if $B_{13}$ is nonempty. Then, by Lemmas \ref{lem:sub} and \ref{lem:E_k}, and Observation \ref{obs:E_k}, we have that
\begin{align*}
    B_7 &\subseteq \left\{\bm{x}( j_3, j_1), \bm{y}(j_{d_H}, j_1)\right\} = \{\bm{x}( j_3, j_2),\bm{y}( j_{d_H}, j_2)\}, 
\end{align*}
\begin{align*} 
    B_{13} &\subseteq \{\bm{x}( j_2, j_1), \bm{y}( j_{d_H}, j_1)\} = \{\bm{x}(j_2, \lambda), \bm{y}( j_{d_H}, \lambda-1)\},  
\end{align*}
where $\lambda \in [j_2 + 1, j_{d_H}]$ and $x_{\lambda} \neq y_{\lambda-1}$.  
Since $B_7$ and $B_{13}$ share a common element, it follows that $|B_7 \cup B_{13}| \leq 3$.  
\end{IEEEproof}

 \subsection{The proof of Lemma \ref{lem:intersect}}\label{subsec:intersect}

 \begin{IEEEproof}
Given any two distinct sequences $\bm{x},\bm{y}\in \C_1$, from Lemma \ref{lem:E_k_odd}, we know that $B_1$ and $B_3$ are empty,  thus, we have $d_H(\bm{x}_{[j_2+1,j_{d_H}]},\bm{y}_{[j_2,j_{d_H}-1]})\ge 1$ and $d_H(\bm{x}_{[j_1+1,j_{d_H-1}]},\bm{y}_{[j_1,j_{d_H-1}-1]})\ge 1$ according to Lemma \ref{lem:E_k_odd}. In addition, if $B_{5}$ is nonempty, then we have $d_H(\bm{x}_{[j_1+1,j_{d_H}]},\bm{y}_{[j_1,j_{d_H}-1]})=1$, thus, we can obtain that $B_{17}$ is nonempty. On the other hand, we have
\begin{align*}
d_H(\bm{x}_{[j_1+1,j_{d_H}]},\bm{y}_{[j_1,j_{d_H}-1]})&= d_H(\bm{x}_{[j_1+1,j_{2}]},\bm{y}_{[j_1,j_{2}-1]})+ d_H(\bm{x}_{[j_2+1,j_{d_H}]},\bm{y}_{[j_2,j_{d_H}-1]})\\
&=d_H(\bm{x}_{[j_1+1,j_{d_H-1}]},\bm{y}_{[j_1,j_{d_H-1}-1]})+d_H(\bm{x}_{[j_{d_H-1}+1,j_{d_H}]},\bm{y}_{[j_{d_H-1},j_{d_H}-1]}),
\end{align*}
which means that 
$d_H(\bm{x}_{[j_2+1,j_{d_H}]},\bm{y}_{[j_2,j_{d_H}-1]})= d_H(\bm{x}_{[j_1+1,j_{d_H-1}]},\bm{y}_{[j_1,j_{d_H-1}-1]})=1$, $d_H(\bm{x}_{[j_1+1,j_{2}]},\bm{y}_{[j_1,j_{2}-1]})=0$ and $d_H(\bm{x}_{[j_{d_H-1}+1,j_{d_H}]},\bm{y}_{[j_{d_H-1},j_{d_H}-1]})=0$. Thus, we know that $B_{13}$ and $B_{15}$ are nonempty from Lemma \ref{lem:E_k_odd}. In addition, we also have $d_H(\bm{x}_{[j_2+1,j_{d_H-1}]},\bm{y}_{[j_2,j_{d_H-1}-1]})=1$ since $d_H(\bm{x}_{[j_2+1,j_{d_H}]},\bm{y}_{[j_2,j_{d_H}-1]})= d_H(\bm{x}_{[j_1+1,j_{d_H-1}]},\bm{y}_{[j_1,j_{d_H-1}-1]})=1$. 
Then, by Lemmas \ref{lem:sub} and \ref{lem:E_k}, and Observation \ref{obs:E_k}, we can obtain that
\begin{align*}
    B_5  &\subseteq \{\bm{x}( j_1, \lambda), \bm{y}( j_{d_H}, \lambda - 1)\}  
          = \{\bm{x}( j_1, 0), \bm{y}( j_{d_H}, 0)\}, 
\end{align*}
\begin{align*} 
    B_7  &\subseteq \{\bm{x}( j_3, j_1), \bm{y}( j_{d_H}, j_1)\}  
          = \{\bm{x}(j_3, j_2), \bm{y}( j_{d_H}, j_2)\}, 
\end{align*}
\begin{align*} 
     B_9  &\subseteq \{\bm{x}(j_1, j_{d_H}), \bm{y}(j_{d_{H}-2}, j_{d_H})\}=\{\bm{x}(j_1, j_{d_H-1}), \bm{y}( j_{d_{H}-2}, j_{d_H-1})\}, 
\end{align*}
\begin{align*} 
    B_{13} &\subseteq \{\bm{y}( j_{d_H}, \lambda - 1), \bm{y}( j_{d_H}, j_1)\}  
          = \{\bm{x}( j_2, \lambda),\bm{x}( j_2, j_1)\}, 
\end{align*}
\begin{align*} 
    B_{15} &\subseteq \{\bm{x}( j_1, j_{d_H}), \bm{x}( j_1, \lambda)\} 
          = \{\bm{y}(j_{d_H-1}, j_{d_H}), \bm{y}( j_{d_H-1}, \lambda - 1)\},
\end{align*}
where $\lambda \in [j_2 + 1, j_{d_H-1}]$ and $x_{\lambda}\neq y_{\lambda-1}$.  

It is clear that $B_{13} \subseteq B_5 \cup B_7$ and $B_{15} \subseteq B_5 \cup B_9$. Moreover,  
\begin{align*}
 B_{13} \cup B_5&\subseteq \{\bm{x}( j_1, \lambda), \bm{y}( j_{d_H}, \lambda - 1), \bm{y}( j_1, j_{d_H})\}, 
\end{align*}
\begin{align*} 
 B_{15} \cup B_5&\subseteq \{\bm{x}( j_1, \lambda), \bm{y}( j_{d_H}, \lambda - 1),\bm{x}( j_1, j_{d_H})\}.
\end{align*}

Next, consider $(d_x, d_y) \in E_{17}$. Since $B_5$ is a nonempty set, we have $E_{17}\subseteq \{(\leq j_1 - 1, j_{d_H}),(j_1, \geq j_{d_H} + 1)\}$ from Lemma \ref{lem:E_k}. Furthermore,
 when $(d_x, d_y) = (\leq j_1 - 1, j_{d_H})$, we have  
\[
B_{17} \subseteq  \{\bm{y}( j_{d_H}, \lambda_1 - 1), \bm{y}( j_{d_H}, \lambda - 1)\}=\{\bm{x}( d_x, \lambda_1), \bm{x}( d_x, \lambda)\}.
\]
 When $(d_x, d_y) = (j_1, \geq j_{d_H} + 1)$, we have  
\[
B_{17} \subseteq \{\bm{x}( j_1, \lambda), \bm{x}( j_1, \lambda_2)\}  
= \{\bm{y}( d_y, \lambda - 1),\bm{y}( d_y, \lambda_2 - 1)\},
\]
where $\lambda_1 \leq j_1,\lambda_2 \geq j_{d_H}$ and $x_{\lambda_1}\neq y_{\lambda_1-1},x_{\lambda_2}\neq y_{\lambda_2-1}$. Thus, we can obtain that
\begin{align*}
  B_{17} \subseteq \{\bm{x}( j_1, \lambda), \bm{x}( j_1, \lambda_2),\bm{y}( j_{d_H}, \lambda_1 - 1), \bm{y}( j_{d_H}, \lambda - 1)\}  
= \{\bm{x}( d_x, \lambda_1), \bm{x}( d_x, \lambda),\bm{y}( d_y, \lambda - 1),\bm{y}( d_y, \lambda_2 - 1)\}.  
\end{align*}
It is clear that  $B_5\subseteq B_{17}$, which implies  $|B_{13} \cup B_{17}| \leq 5$, $|B_{15} \cup B_{17}| \leq 5$ and $|B_{13}\cup B_{15} \cup B_{17}| \leq 6$. Moreover, we have $B_{13}\subseteq B_7\cup B_{17}$ and $B_{15}\subseteq B_9\cup B_{17}$, it follows that $\left|\bigcup_{i\in I\setminus\{7,9\}} B_i\right|\leq 6$ and 
\begin{align*}
 \left|\bigcup_{i\in I} B_i\right|&=| B_7\cup B_9\cup B_{17}|\leq |B_7|+|B_{9}|+ |B_{17}|\le 8,
 \end{align*}
\begin{align*} 
 \left|\bigcup_{i\in I\setminus\{7\}} B_i\right|&=| B_9\cup B_{13}\cup B_{17}|\leq |B_9|+|B_{13}\cup B_{17}|\le 7,
 \end{align*}
\begin{align*} 
 \left|\bigcup_{i\in I\setminus\{9\}} B_i\right|&=| B_7\cup B_{15}\cup B_{17}|\leq |B_7|+|B_{15}\cup B_{17}|\le 7.
\end{align*}
\end{IEEEproof}

\end{document}